\let\accentvec\vec
\documentclass[cmp,final,envcountsame]{mystyle}

\let\vec\accentvec

\usepackage{ctable}
\usepackage{graphicx}
\usepackage[cmex10]{amsmath}
\usepackage{amssymb}

\newcommand{\nb}{\nobreakdash-\hspace{0pt}}

\usepackage{url}
\usepackage{color}
\usepackage{colordvi}
\usepackage{xspace}

\newcommand{\opvec}{\operatorname{vec}}
\newcommand{\VEC}[1]{\operatorname{vec(#1)}\xspace}
\newcommand{\N}{\ensuremath{\mathbb N}{}}
\newcommand{\R}{\ensuremath{\mathbb R}}
\newcommand{\C}{\ensuremath{\mathbb C}}

\newcommand{\unity}{\ensuremath{{\rm 1 \negthickspace l}{}}}
\newcommand{\ord}{\ensuremath{{\rm ord}{}}}

\newcommand{\ket}[1]{\ensuremath{| #1 \rangle}{}}

\newcommand{\ketbra}[2]{\ensuremath{| #1 \rangle \langle #2 |}{}}

\newcommand{\expt}[1]{\ensuremath{\langle #1 \rangle}{}}

\newcommand{\adr}{\operatorname{ad}}

\newcommand{\rep}{\ensuremath{\overset{\rm rep}{=}}\xspace}

\newcommand{\tr}{\operatorname{tr}}

\newcommand{\rank}{\operatorname{rank}{}}

\newcommand{\reach}{\mathfrak{reach}{}}
\newcommand{\hoplus}{\ensuremath{\, {\widehat{\oplus}} \,}{}}

\newcommand{\spanR}{\operatorname{span_{\mathbb R}}{}}

\newcommand{\uu}{\mathfrak{u}}
\newcommand{\su}{\mathfrak{su}}
\newcommand{\SU}{\operatorname{SU}}
\newcommand{\oo}{\mathfrak{o}}
\newcommand{\so}{\mathfrak{so}}

\newcommand{\spp}{\mathfrak{sp}}

\newcommand{\usp}{\mathfrak{sp}}
\newcommand{\sll}{\mathfrak{sl}}
\newcommand{\gl}{\mathfrak{gl}}
\newcommand{\GL}{\mathrm{GL}}
\newcommand{\e}{{\rm e}}

\newcommand{\fa}{\mathfrak{a}}
\newcommand{\fb}{\mathfrak{b}}
\newcommand{\fc}{\mathfrak{c}}
\newcommand{\fd}{\mathfrak{d}}
\newcommand{\fe}{\mathfrak{e}}
\newcommand{\ff}{\mathfrak{f}}
\newcommand{\fg}{\mathfrak{g}}
\newcommand{\fh}{\mathfrak{h}}

\newcommand{\fk}{\mathfrak{k}}
\newcommand{\fm}{\mathfrak{m}}

\newcommand{\fp}{\mathfrak{p}}
\newcommand{\fs}{\mathfrak{s}}

\newcommand{\fz}{\mathfrak{z}}
\newcommand{\bG}{\mathbf{G}}
\newcommand{\bK}{\mathbf{K}}

\newcommand{\Alt}{\mathrm{Alt}}
\newcommand{\Sym}{\mathrm{Sym}}
\newcommand{\XX}{{\sf XX}}
\newcommand{\XXX}{{\sf XXX}}
\newcommand{\XY}{{\sf XY}}
\newcommand{\ZZ}{{\sf ZZ}}
\newcommand{\ZZZ}{{\sf ZZZ}}

\newcommand{\xy}{\ensuremath{\left\{\begin{smallmatrix} x\\y\end{smallmatrix}\right\}}\xspace }

\newcommand{\ueber}[2]{\genfrac{}{}{0pt}{}{#1}{#2}}
\newcommand{\uebertwo}[2]{\genfrac{}{}{0pt}{0}{#1}{#2}}

\usepackage[bookmarks=false,pdfstartview={FitH},hyperfootnotes=false]{hyperref}

\begin{document}
\title{Symmetry Principles in Quantum Systems Theory}
\subtitle{}
\author{
R.~Zeier
\and 
T.~Schulte-Herbr{\"u}ggen
}                     
%
\institute{Department of Chemistry, Technical University of Munich (TUM)\\ 
Lichtenbergstrasse 4, 85747 Garching, Germany\\
\email{robert.zeier@ch.tum.de; tosh@ch.tum.de}
}
%
\date{Date: July 21, 2011}
%
%
\maketitle

\begin{abstract}
General dynamic properties like controllability and simulability of spin systems, 
fermionic and bosonic systems are investigated in terms of symmetry. Symmetries may be
due to the interaction topology or due to the structure and representation of
the system and control Hamiltonians. 
In either case, they obviously entail constants of motion.
Conversely, the absence of symmetry implies irreducibility and
provides a convenient necessary condition for full controllability
much easier to assess than the well-established Lie-algebra rank condition.
We give a complete lattice of irreducible simple subalgebras of $\su(2^n)$
for up to $n=15$ qubits. It complements the symmetry condition by allowing for
easy tests solving homogeneous linear equations to filter irreducible
unitary representations of other candidate algebras of classical type 
as well as of exceptional types. ---
The lattice of irreducible simple subalgebras given also determines mutual 
simulability of dynamic systems of spin or fermionic or bosonic nature.
We illustrate how controlled quadratic fermionic (and bosonic) systems can be 
simulated by spin systems and in certain cases also vice versa.
\end{abstract}

{\small
\tableofcontents
}

\section{Introduction}
Experimental control over quantum dynamics of manageable systems
is paramount to exploiting the great potential of quantum systems.
Both in simulation and computation the complexity of a problem
may reduce upon going from a classical to a quantum setting \cite{Fey82,Fey96,Kit02}.
On the computational end, where quantum algorithms 
efficiently solving hidden subgroup problems \cite{EHK04}
have established themselves, the demands for accuracy (\/`error-correction threshold\/')
may seem daunting at the moment. In contrast, the quantum simulation end is by far less sensitive.
Thus simulating quantum systems \cite{Lloyd96}---in particular at phase-transitions 
\cite{Sachdev99,QPT10}---has recently shifted into focus \cite{BCL+02,DNB+02,JC03,ZGB,PC04}.
In view of experimental progress in cold atoms in optical lattice potentials \cite{GMEHB02,BDZ08} 
as well as in trapped ions \cite{LBM+03,BW08}, Kraus {\em et al.} have explored whether
target quantum systems can be universally simulated on translationally invariant lattices
of bosonic, fermionic, and spin systems \cite{kraus-pra71}. In some respect, their work can also be seen as a
follow-up on a study by Schirmer {\em et al.} \cite{SchiPuSol01} 
(see also recent work by Wang~et~al. \cite{WPRS10})
specifically addressing
controllability of systems with degenerate transition frequencies.
Many experimental tasks are engineering problems that profit
from quantum systems theory as a framework and optimal control algorithms
for solving the actual problem.

As compared to an abstract point of view \cite{PrimMH78}, the flavour of quantum systems theory pursued here
is meant to be very pragmatic: it takes the causal formulation of dynamic systems \cite{KFA69} and
does not care about specifics of the quantum measurement problem beyond the basic notions
\cite{Helstrom76} and some recent developments \cite{HeinWolf10}.
Yet it is for these reasons that quantum systems and control has quite generally been recognised
as a key generic tool \cite{DowMil03,dAless08,WisMil09} needed for advances in experimentally exploiting
quantum systems for simulation or computation and even more so in future quantum technology.
It paves the way for constructively optimising strategies
for experimental implementions in realistic settings.
Moreover, since such realistic quantum systems are mostly beyond analytical
tractability, numerical methods are often indispensable.
To this end, gradient flows can be implemented on the control amplitudes thus
iterating an initial guess into an optimised pulse scheme \cite{Rabitz87,Krotov,GRAPE}.
This approach has proven useful in spin systems \cite{PRA05} as well as in solid-state
systems \cite{PRA07}. Moreover, it has recently been generalised from closed systems to
open ones \cite{PRL_decoh}, which are known to be a challenge to control \cite{VioLloyd01}, where the Markovian setting can also be used as embedding of
explicitly non-Markovian subsystems \cite{PRL_decoh2}.

However, in closed systems, the numerical tools 
usually require the system is universal or fully operator controllable \cite{RaRa95,AA03}.
For a plethora of systems with symmetry constraints we have recently determined explicit dynamic system
algebras \cite{SS09} (as subalgebras of $\su(N)$), and conversely, we have derived design rules for the experimenter 
as guidelines ensuring universality of quantum architecture. While extending earlier work on
branching diagrams of simple subalgebras of $\su(N)$ \cite{McKay81,PST09},
here we focus on {\em complete necessary and sufficient conditions} for full controllability
(mostly) confining ourselves to arguments easy to check by inspection or to decide by computationally cheap  
algorithms such as solving a system of homogeneous linear equations.

In view of applications, we illustrate our findings by a comprehensive set of worked
examples on spin chains. Actually Ising-$\ZZ$ coupled $n$\nb{}spin\nb$\tfrac{1}{2}$ chains with
mostly {\em collective controls} or Heisenberg\nb$\XX$ chains with {\em one single local control}
suffice to get {\em exponential growth} of dynamic degrees of freedom (in the sense their
respective dynamic system algebras are $\usp(2^{n-1})$ or $\so(2^n)$). Our work thus adds
to the recent spin-chain literature (see, e.g., \cite{Bose07,Burg08,SPP08,SP09,Kay08,Kay09,BP09,KPR09,PK10,WPRS10} and compare
\cite{KG02,YGK07,YZK08}) 
and---on a more general scale---it is anticipated to have significant impact on quantum simulation
as well as distributed quantum computing (see, e.g., \cite{CB97,CEHM99,YF09}).

\section{Overview and Main Results}
More precisely, the first main part develops, starting from the basic notions of controllability 
(Sec.~\ref{sec:control-basics}) in terms of coupling graphs (Sec.~\ref{sec:tensor})
and their symmetries (Sec.~\ref{sec:sym_cons}), {\em a single
necessary and sufficient symmetry condition for full controllability}
(Sec.~\ref{sec:control-suff}). To this end and in view of practical applications, Sec.~\ref{sec:irred-sub-su}
gives branching diagrams of {\em all irreducible simple subalgebras of the unitary algebras}
$\su(N)$ with $N\leq 2^{15}$. Concomitantly, 
we provide a set of efficient computational {\em algorithms for assessing controllability}
by merely solving systems of homogeneous linear equations. 

The second part focusses on simulability (Sec.~\ref{sec:simulability}) 
in terms of dynamic system algebras. A plethora of worked examples is discussed in Sec.~\ref{sec:w-examples}
including {\em four full series of qubit chains coupled by pair interactions} 
such that their dynamic system algebras for the first three cases are $\so(2n+1)$, $\so(2n+2)$, 
and $\usp(2^{n-1})$, respectively.  
Most remarkably, for $n\geq 4$,
the fourth series results in dynamic system algebras
$\so(2^n)$ if $(n \bmod{4})\in\{0,1\}$ and $\usp(2^{n-1})$ else. 
The findings also interrelate spin systems, fermionic systems (Sec.~\ref{sec:fermionic})
and bosonic systems (Sec.~\ref{sec:bosonic}).
The algebraic conditions for simulability given are sufficient to ensure the existence of
solutions to the actual task
of quantum simulation of closed systems formulated as an observed optimal control problem 
in the outlook (Sec.~\ref{sec:outlook}).

\section{Controllability}\label{sec:control-basics}
Consider the controlled Schr{\"o}dinger equation lifted to unitary maps (quantum gates)
\begin{equation}
        \label{eqn:bilinear_contr}
        {\dot U(t)} = -i\big(H_d + \sum_{j=1}^m u_j(t) H_j\big) \;{U(t)}\;.
\end{equation}
Here the system Hamiltonian $H_d$ denotes a non-switchable
drift term and the control Hamiltonians $H_j$ can be steered by
(piece-wise constant) control amplitudes $u_j(t)\in\R$ taken to
be unbounded henceforth.
The equation of motion governs the evolution of a unitary map of an entire basis set
of vectors representing pure states. Using the short-hand notations
$H:= H_d + \sum_{j=1}^m u_j(t) H_j$ and $\adr_H(\opvec{A}):=[H,A]$,
the Liouville equation $\dot{\rho}(t) = -i [H,\rho(t)]$
can be rewritten 
\begin{equation}\label{eqn:master}
\opvec{\dot{\rho}(t)} = -i \adr_{H}\opvec{\rho(t)}.
\end{equation}
Both equations of motion take the form of a standard
{\em bilinear control system} $(\Sigma)$ known in classical systems and control theory \cite{Elliott09}
\begin{equation}\label{eqn:bilinear_contr2}
        \dot X(t) = \big(A + \sum_{j=1}^m u_j(t) B_j\big) \; X(t)
\end{equation}
with \/`state\/' $X(t)\in\C^N$, drift $A\in \gl(N,\C)$, controls $B_j\in \gl(N,\C)$,
and control amplitudes $u_j\in\R$, where $\gl(N,\C)$ denotes the set of complex $N\times N$ matrices. 
Since all the control systems considered henceforth
are bilinear, we often drop the specification bilinear for short.
Now lifting the (bilinear) control system $(\Sigma)$ to group manifolds \cite{Bro72,Jurdjevic97}
by $X(t) \in \GL(N,\C)$, i.e. the set of non-singular complex $N\times N$ matrices, 
under the action of a compact connected Lie group $\mathbf K$ with Lie algebra $\fk$
while keeping $A,B_j\in \gl(N,\C)$,
the condition for full controllability turns into the
{\em Lie algebra rank condition} \cite{SJ72,JS72,Jurdjevic97}
\begin{equation}
        \langle A, B_j \,|\,j=1,2,\dots,m\rangle_{\rm Lie} = \fk,
\end{equation}
where $\langle \cdot \rangle_{\rm Lie}$ denotes (the linear span over)
the {\em Lie closure} obtained by repeatedly taking mutual commutator brackets.
{\bf Algorithm~1} gives an explicit method to compute the Lie closure, see also~\cite{SchiFuSol01}.

\begin{table}[Ht!]
\label{alg:1} 
\begin{center}
\begin{tabular}{ll}
\hline\hline\\[-1mm]
\multicolumn{2}{l}{{\bf Algorithm 1}: Determine system algebra via Lie closure}\\[1mm]
\hline\\[-1mm]
&{\em Input:} Hamiltonians $I:=\{i H_d; i H_1,\dots, i H_m\} \subseteq \su(N)$\\[1mm]
&1. $B :=$ maximal linearly independent subset of $I$\\[1mm]
&2. $\mathrm{num} := \#B$\\[1mm] 
&3. If $\mathrm{num}=N$ then $O :=B$ else $O := \{\,\}$\\[1mm]
&4. If $\mathrm{num}=N$ or $\#B=0$ then terminate\\[1mm]
&5. $C := [O,B] \cup [B,B]$, where\\
&\phantom{5.}  $[S_1,S_2]=\{[s_1,s_2]\,|\, s_1 \in S_1,s_2 \in S_2\}$\\[1mm]
&6. $O := O \cup B$\\[1mm]
&7. $B :=$ max.\ linear independent extension of $O$\\
&\phantom{7.} with elements from $C$\\[1mm]
&8. $\mathrm{num} := \mathrm{num} + \# B$; Go to 4\\[1mm]
& {\em Output:} basis $O$ of the generated Lie algebra and\\
& \phantom{Output:} its dimension $\mathrm{num}$\\[1mm]
\hline\\[-2mm]
& The complexity is roughly $\mathcal{O}(N^6\cdot N^2)$, as about $N^2$ times \\ 
& a rank-revealing $QR$ decomposition has to be performed in\\
& Liouville space (with dimension $N^2$). For $n$ qubits, $N:=2^n$.\\[1mm]
\hline\hline
\end{tabular}
\end{center}
\end{table}

Transferring the classical result \cite{JS72} to the quantum domain \cite{RSD+95,TOSH-Diss,AA03},
the bilinear system of Eqn.~\eqref{eqn:bilinear_contr} is {\em fully (operator) 
controllable} 
if and only if the drift and controls are a generating set of the special unitary algebra $\su(N)$:
\begin{equation}
\langle{i H_d, i H_j} \,|\,j=1,2,\dots,m\rangle_{\rm Lie} = \fk = \su(N).
\end{equation}
In fully controllable systems, to every initial state $\rho_0$ the {\em reachable set} is
the entire unitary orbit $\mathcal O_{\rm U}(\rho_0):=\{U\rho_0 U^\dagger\;|\; U\in \SU(N)\}$.
With density operators being Hermitian
this means any final state $\rho(t)$ can be reached from any initial state $\rho_0$
as long as both of them share the same spectrum of eigenvalues. Thus reachable sets
and isospectral sets coincide.

In contrast, in systems with restricted controllability the Hamiltonians
generate but a proper subalgebra of the full unitary algebra
\begin{equation}
\langle{i H_d, i H_j} \,|\,j=1,2,\dots,m\rangle_{\rm Lie} = \fk \subsetneq \su(N).
\end{equation}
Then the dynamic group $\bK:=\exp\fk$ is but a proper subgroup $\bK\subsetneq\SU(N)$
of the full unitary group. Therefore the corresponding {\em reachable sets} take the form
of subgroup orbits of initial states 
\begin{equation}\label{eqn:reach-k-orbit}
\reach(\rho_0) = \mathcal O_\bK(\rho_0):=\{K\rho_0 K^\dagger\;|\; K\in \bK\subsetneq\SU(N)\}.
\end{equation}

\section[Tensor-Product Structure and Coupling Graphs in Systems with %
	Pair Interactions]{Natural Tensor-Product Structure and Coupling Graphs in Qubit Systems with %
	Pair Interactions\label{sec:tensor}}

We start out with the case of qubit systems coupled by pair interactions.
Yet quantum simulation of effective many-body interactions in
multi-level systems requires more refined notions, see Appendix~\ref{sec:tensor2}
and~\ref{sec:simple}. Thus
we choose a line-of-thought allowing for the extensions needed later in a natural way
while trying to keep the overhead minimal here. Finally it should be stressed that the results in 
Secs.~\ref{sec:sym_cons}--\ref{sec:control-suff} are valid in full generality of 
Appendix~\ref{sec:tensor2} and~\ref{sec:simple}.

To fix the basic terminology,
observe that the abstract {\em direct sum} of Lie algebras has a matrix representation 
as the {\em Kronecker sum}, e.g.,
$\su(d_1) \hoplus \su(d_2):=\su(d_1)\otimes\unity_{d_2} + \unity_{d_1}\otimes\su(d_2)$
and that it generates a group isomorphic to the {\em Kronecker product} 
(i.e. tensor product) $\bG=\SU(d_1)\otimes \SU(d_2)$.
The abstract direct sum of two algebras $\fh_1$ and $\fh_2$ (each given in an irreducible
representation) has itself an irreducible representation as a single Kronecker
sum $\fh_1 \hoplus \fh_2$ (Thm.~11.6.II of Ref.~\cite{Cornwell:1984}). 
Such an irreducible direct sum representation always exists for every semi-simple 
Lie algebra which is not simple.

Control systems consisting of $n$ qubits are usually embedded in $\su(N)$ with $N:=2^n$. 
Their natural intrinsic \emph{tensor-product structure} takes the form
of the $n$-fold Kronecker sum
$\su(2) \hoplus \su(2) \hoplus \cdots \hoplus \su(2)$.
An $N^2$-$1$ dimensional skew-Hermitian
tensor basis with respect to this tensor-product structure can be
given in terms of the Pauli matrices
\begin{equation}
\mathrm{I}:=\unity_2=
\begin{pmatrix}
1 & 0 \\
0 & 1
\end{pmatrix},
\mathrm{X}:=\sigma_x=
\begin{pmatrix}
0 & 1 \\
1 & 0
\end{pmatrix},
\mathrm{Y}:=\sigma_y=
\begin{pmatrix}
0 & -i \\
i & 0
\end{pmatrix}, \text{ and }
\mathrm{Z}:=\sigma_z=
\begin{pmatrix}
1 & 0 \\
0 & -1
\end{pmatrix}\\[2mm]
\end{equation}
by defining the elements 
$- \tfrac{i}{2}\mathrm{H}_1 \mathrm{H}_2 \cdots \mathrm{H}_n$, where
$\mathrm{H}_1 \mathrm{H}_2 \cdots \mathrm{H}_n
:= \mathrm{H}_1 \otimes \mathrm{H}_2 \otimes  \cdots \otimes \mathrm{H}_n$
and $\mathrm{H}_j \in \{\mathrm{I}, \mathrm{X}, \mathrm{Y}, \mathrm{Z} \}$. 
The element $\mathrm{H}_1 = \mathrm{H}_2 =  \cdots = \mathrm{H}_n = \mathrm{I}$ is 
not traceless and hence cannot occur in $\su(2^n)$.
In terms of this tensor basis,
we write Hamiltonians as linear combinations ($c_k \in \R$) 
\begin{equation}\label{Hsum}
H = \sum_{k=1}^m c_k \mathcal{H}_k
\end{equation}
of elements $\mathcal{H}_k=-\tfrac{i}{2} (\mathcal{H}_{k,1} \otimes 
\mathcal{H}_{k,2} \otimes \cdots \otimes \mathcal{H}_{k,n})$ with $\mathcal{H}_{k,j} \in \{\mathrm{I}, \mathrm{X}, \mathrm{Y}, \mathrm{Z} \}$.
Considering local controls and pairwise coupling interactions
the orders of the constituents are confined, i.e.\
\begin{equation*}
\ord(\mathcal{H}_k):=
\# \{ \ell\, \colon \, \mathcal{H}_{k,\ell} \neq \unity_2 \} \in\{1,2\}.
\end{equation*}
Usually, the control Hamiltonians $H_j$ are local, i.e.\ all terms in Eqn.~\eqref{Hsum} (for $H=H_j$) are of order one,
while the corresponding terms in Eqn.~\eqref{Hsum} for the drift Hamiltonian $H_0\, (= H)$ are of order two comprising the
non-switchable {\em pairwise} coupling terms.
 
Now, in a \emph{coupling graph} the vertices representing the local subsystems
are connected by edges, where each edge stands for a pairwise coupling term
occuring in the drift Hamiltonian $H_d$. An example of a {\em connected} coupling
graph is shown in Figure~\ref{fig:gen-topo-graph}. --- Connected coupling
graphs are essential for full controllability as elucidated by the following theorem.

\begin{theorem}\label{thm:bilinear}
Consider a bilinear control system with pair interactions on $\su(2^n)$, where
all the local subsystems $\su(2)$ are independently fully controllable
so the dynamic algebra $\fk \supseteq \su(2) \hoplus \su(2) \hoplus \cdots \hoplus \su(2)$. 
Then the system is fully controllable, i.e.\ $\fk=\su(2^n)$,
if and only if its coupling graph is connected. In particular, $\fk=\su(2^n)$ is simple.
\end{theorem}
\begin{proof}
A proof is given in Ref.~\cite{AlbAll02}
(see Thm.~2, Remark 5.1, and Thm.~4), see also Ref.~\cite{TOSH-Diss}.
\hfill$\blacksquare$
\end{proof}

\begin{figure}[Ht!]
\begin{center}
\includegraphics{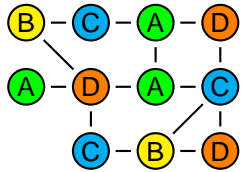}
\end{center}
\caption{General coupling topology represented by a connected graph. The vertices denote the spin-$\tfrac{1}{2}$
qubits, while the edges represent pairwise couplings (e.g.\ of Heisenberg or Ising type).
Qubits of the same colour and letter are taken to be affected by joint local
unitary operations (or none if the color is white), while qubits of different kind can be
controlled independently. For a system to show an outer symmetry brought about by
permutations within subsets of qubits of the same type,
both the graph as well as the system plus all control Hamiltonians have to remain invariant.\label{fig:gen-topo-graph}}
\end{figure}

\section{Symmetry-Constrained Controllability\label{sec:sym_cons}}

A Hamiltonian quantum system is said to have
a symmetry expressed by the skew-Hermitian {\em symmetry operator} $s \in \su(N)$, if
\begin{equation}
[s, H_\nu] = 0 \quad\text{for all}\quad \nu \in \{d; 1,2,\dots,m\}.
\end{equation}
More precisely, we use the term {\em outer symmetry} if $s$
generates a SWAP operation permuting a subset of qubits 
of the same
type (cp.~Fig.~\ref{fig:gen-topo-graph})
such that the coupling graph and all Hamiltonians $\{H_\nu\}$ are left invariant.
Now subsets of 
qubits are termed {\em indistinguishable} if and only if they
can be interchanged by an outer symmetry, i.e.~a SWAP operation that is a symmetry 
of the system; otherwise they are {\em distinguishable}.
In contrast, an {\em inner symmetry} relates to elements $s$
not generating a SWAP operation in the symmetric group of all qubit permutations.

In either case, a symmetry operator is an element of the {\em centraliser}
\begin{equation}
\{H_\nu\}':= \mathcal{Z}_{\su(N)}(\{H_\nu\})=\big\{s\in\su(N)\,|\, [s, H_\nu] = 0 \quad \forall \nu \in \{d; 1,2,\dots,m\}\big\},
\end{equation}
recalling that the centraliser  
of a given subset $\mathfrak m\subseteq \su(N)$ with respect to a Lie algebra $\su(N)$
consists of all elements in $\su(N)$ that commute with all elements in $\mathfrak m$.
Jacobi's identity \mbox{$\big[[a,b],s\big]+\big[[b,s],a\big]+\big[[s,a],b\big]=0$}
gives two useful facts:
(1) an element $s$ that commutes with the Hamiltonians $\{iH_\nu\}$ also commutes
with their Lie closure $\fk$. For the dynamic Lie algebra $\fk$ we have
\begin{equation}
\fk' := \mathcal{Z}_{\su(N)}(\fk) = \{s \in \su(N)\;  |\; [s,k] = 0 \quad\forall k \in \fk\}
\end{equation}
and hence $\{iH_\nu\}'\equiv \fk'$.
Thus in practice it is (most) convenient to just evaluate the centraliser
for a (minimal) generating set $\{iH_\nu\}$
of $\fk$ 
since the overall symmetry properties can be read from the local symmetries
of the constituent Hamiltonians.
Fact (2) means the centraliser $\fk'$ forms itself an invariant 
Lie subalgebra (or ideal) to $\su(N)$ collecting {\em all symmetries}.
In summary, we obtain the following straightforward, yet important result:
\begin{theorem}
Lack of symmetry in the sense of a trivial centraliser is a necessary
condition for full controllability.
\end{theorem}
\begin{proof}
Any non-trivial element in the centraliser
would generate a one-parameter subgroup in \mbox{$\bK'\subset \SU(N)$} 
that is {\em not} in $\bK=\exp \fk$.
$\hfill \blacksquare$
\end{proof}

\begin{table}[Ht!]
\begin{center}
\begin{tabular}{llll}
\hline\hline\\[-1mm]
\multicolumn{2}{l}{{\bf Algorithm 2}:
Determine centraliser (resp.\ commutant)}\\
\multicolumn{2}{l}{\phantom{{\bf Algorithm 2}:}
 to system algebra $\fk$}\\[1mm]
\hline\\[-1mm]
&{\em Input:} Hamiltonians $I:=\{i H_d; i H_1,\dots, i H_m\} \subseteq M$\\[1mm]
&1. For each $H \in I$ solve the homogeneous linear eqn.\\
&\phantom{1.} $\mathcal S_{H}:=\{ s \in M | (\unity\otimes H - H^t\otimes\unity) \opvec(s) = 0\}$\\[1mm]
&2. $R  := \bigcap_{H{\in}I} \mathcal S_{H}$. \\[1mm]
&{\em Output:}\\
&\phantom{Output:} $R = 
\begin{cases}
\text{centraliser $\fk'$} & \text{if $M=\su(N)$}\\
\text{commutant of $\fk$} & \text{if $M=\gl(N,\C)$}\\
\end{cases}$\\[5mm]
\hline\\[-1mm]
& The complexity is roughly $\mathcal{O}(N^6)$, as in Liouville space \\ 
& $N^2$ equations have to be solved by $LU$ decomposition. \\
& For $n$ qubits, $N:=2^n$.\\[1mm]
\hline\hline
\end{tabular}
\end{center}
\end{table}

\medskip
Throughout this paper, we consider 
finite-dimensional complex matrix representations of Lie algebras,
a representation being a map from a given Lie algebra to the set
complex square matrices of appropriate (and finite) dimension.
The matrix
entries are given by complex polynomial (or equivalently holomorphic) functions.
In the following, we will usually not consider the trivial representation,
which maps any element to $1\in \C$.
One particular important example for a representation of a Lie algebra
is the standard representation, which is the lowest-dimensional (non-trivial) representation
(with some exceptions, see the Appendix~\ref{app:repr}) and which is typically used to
define the corresponding Lie algebra in its matrix form.
In analogy to the centraliser, one can define the {\em commutant} relative
to a representation $\phi$ of dimension $\dim(\phi)$
\begin{equation}
\mathrm{comm}_{\phi}(\fm):=\big\{g \in\gl(\dim(\phi),\C)\,|\, [g, \phi(m)]=0\quad \forall m \in\fm\big\}
\end{equation}
for a subset $\fm\subset\fg$ of a Lie algebra $\fg$. Now it is natural to ask
how the notions of centraliser and commutant relate to irreducible representations.

\begin{lemma}\label{lem:centraliser}
Let $\Phi$ denote the standard representation
of $\su(N)$. If $\fk \subseteq \su(N)$, then the following statements are 
equivalent:\\[1mm]
(1) The centraliser 
$\fk'=\mathcal{Z}_{\su(N)}(\fk)$
of $\fk$ in $\su(N)$ is trivial, i.e.\ zero.\\
(2) The restriction of $\Phi$ from $\su(N)$ to $\fk$ is irreducible.\\
(3) The commutant $\mathrm{comm}_{\Phi}(\fk)$ 
of $\fk$ w.r.t.\ $\Phi$ is trivial, i.e.\ $= \{ c\cdot\unity_{N}\; | c\in\C \}$.
\end{lemma}

\begin{proof}
As $\su(N)$ is compact, it follows that $\Phi$ and its restriction to $\fk$ 
are completely reducible in the sense of being a direct sum of irreducible representations
(see Cor.~2.17 of \cite{Sepanski07}).
The representation $\Phi$ is even irreducible and faithful, i.e.\ injective.
Hereafter, we will consider the complexification $\fk_{\C}$ of $\fk$ and
$\su(N)_{\C}=\sll(N,\C)$ as complexification of $\su(N)$.
The representation $\Phi$ has a unique extension $\Phi_{\C}$ to $\sll(N,\C)$, 
which is also irreducible and faithful.
In addition, $\Phi_{\C}$ and its restriction to $\fk_{\C}$ are
completely reducible. These facts can be deduced from
Thm.~1, pp.~111--112 of 
\cite{Zelobenko73} and Prop.~7.5 of \cite{Knapp02}.

Now it follows that (1) is equivalent to $\mathcal{Z}_{\sll(N)}(\fk_{\C})=\{0\}$.
As $\Phi_{\C}$ is faithful, this holds if and only if $\mathrm{comm}_{\Phi_{\C}}(\fk_{\C})$ is trivial.
Relying on the fact that $\Phi_{\C}$ is completely reducible, $\mathrm{comm}_{\Phi_{\C}}(\fk_{\C})$ is trivial if and only if the restriction of $\Phi_{\C}$ from $\sll(N,\C)$ to
$\fk_{\C}$ is irreducible. Using Thm.~1, pp.~111--112 of 
\cite{Zelobenko73}, this is equivalent to (2). As $\Phi$ is completely reducible,
(2) and (3) are equivalent.
\hfill$\blacksquare$
\end{proof}

As a second consequence of a trivial centraliser the corresponding subalgebra
$\fk$ of $\su(N)$ has to be simple or semi-simple:

\begin{lemma}\label{lem:semi-simplicity}
Let $\fk\subseteq\su(N)$ be a subalgebra
to the Lie algebra $\su(N)$. If its centraliser $\fk'$ in $\su(N)$
is trivial, then $\fk$ is simple or semi-simple.
\end{lemma}
\begin{proof}
By compactness, $\fk = \fz_\fk \oplus \fs$ decomposes into its centre $\fz_\fk$ 
and a semi-simple part $\fs$ (see, e.g., Cor.~IV.4.25 of Ref.~\cite{Knapp02}). 
As the centre
$\fz_\fk = \fk' \cap \fk$ is trivial, 
$\fk$ can only be {\em semi-simple} or {\em simple}. 
\hfill$\blacksquare$
\end{proof}

Note that the centraliser is \/`exponentially\/' easier to come by than the Lie closure 
in the sense of comparing the asymptotic complexity $\mathcal O(N^6\cdot N^2)$ (with $N:=2^n$ for $n$ qubits) of
{\bf Algorithm~1} for the Lie closure with the asymptotic complexity $\mathcal O(N^6)$ of
{\bf Algorithm~2} for the centraliser tabulated above. ---
Therefore one would like to fill the gap between lack of symmetry as a necessary 
condition and sufficient conditions for full 
controllability in systems with a connected coupling topology.
For pure-state controllability, this was analysed in \cite{AlbAll02},
for operator controllability the issue has been raised in \cite{dAless08}, {\em inter alia}
following the lines of \cite{TR01,TR03}, however, without a full answer.
Further results in the case of pure-state controllability
can be found in \cite{PST09}.

We have proven that the lack of symmetry is necessary for a control system to be fully controllable.
Yet in turn, a control system without symmetry need not be fully controllable, as the following
elementary (and pathological) example shows:
\begin{example}\label{ex1}
Assume we have a bilinear control system on two qubits, where 
the dynamic Lie algebra 
$\fk=\langle i \mathrm{XI}, i \mathrm{YI}, i \mathrm{ZI}, i \mathrm{IX}, i \mathrm{IY}, i \mathrm{IZ} \rangle_{\rm Lie}=\su(2) \hoplus \su(2)$
is not simple. 
Although it has no symmetry and its centraliser $\fk'$ is in fact trivial, the system
is not fully controllable: all pair terms like $i\mathrm{ZZ}$ cannot be generated,
since its pathological \/`coupling graph\/'
\begin{center}
\includegraphics{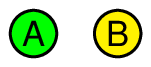}
\end{center}
is clearly not connected.
\end{example}

Nevertheless, the somewhat trivial example is illuminating. While in the context
of $C^*$-algebras, von Neumann's double-commutant theorem recovers the original
algebra from the commutant of its commutant \cite{Dix81,Sak71}, a similar theorem does not
extend to Lie algebras \cite{GW09}. Rather, if the dynamic algebra $\fk\subseteq\su(N)$ 
has a trivial centraliser $\fk'$, then the double centraliser $\fk''$,
i.e.\ the centraliser of the centraliser in $\su(N)$, 
of all compact semi-simple and simple irreducible proper and improper subalgebras $\fk$
of $\su(N)$ is given by $\su(N)$ in line with Lemma~\ref{lem:semi-simplicity}. 
However, if one considers
the associative matrix algebra (with identity) generated by the basis elements 
(including the identity matrix) of a Lie algebra
via its standard representation, then von Neumann's double commutant theorem
still holds, see Thm.~(3.5.D) of Ref.~\cite{Weyl}. --- In the next step, we
will thus add a criterion to single out the simple subalgebras.

Motivated by Example~\ref{ex1} 
one might conjecture that the dynamic algebra $\fk$ is simple 
if $\fk$ acts irreducibly and the coupling graph of the control system is connected.
This is true for control systems in qubits with pairwise coupling interactions:
\begin{theorem}\label{pair_connected}
Consider a bilinear control system with pair interactions on $\su(2^n)$. Assume that
the tensor-product structure is given by $\su(2) \hoplus \su(2) \hoplus \cdots \hoplus \su(2)$
and that the centraliser $\fk'$ of the dynamic algebra $\fk$ is trivial. The dynamic algebra $\fk$
is simple if and only if the coupling graph of the control system is connected.
\end{theorem}
\begin{proof}
See Corollary~\ref{cor:connected}(2) in Appendix~\ref{sec:simple}.
\hfill$\blacksquare$
\end{proof}

The general case beyond pair interactions (and qubit systems) is discussed in Appendix~\ref{sec:simple}.
In the case of pair interactions, we say 
a control system is \emph{connected} if its
coupling graph is connected. 
This definition of a connected control system is a particular case of 
the general definition (see Appendix~\ref{sec:simple}) applicable to control systems which do not have a natural coupling graph.

\section{Irreducible Simple Subalgebras of $\su(N)$\label{sec:irred-sub-su}}

Starting from the knowledge that for a fully controllable system the dynamic algebra $\fk$ 
has to be simple and given in an irreducible representation (see, e.g., Appendix~\ref{sec:simple}), 
it is natural to ask for a classification of all these cases. Following the work of 
Killing, {\'E}lie~Cartan~\cite{Cartan94} classified all
simple (complex) Lie algebras (see, e.g., \cite{Bourb08a,Bourb08b}). 
The corresponding compact real forms (\cite{Bourb08b,Helgason78}) are 
the compact simple  Lie algebras of classical type 
(assuming $\ell \in \N\setminus\{0\}$ henceforth): 
\begin{align*}
\fa_{\ell}: &\quad\su(\ell + 1),\\
\fb_{\ell}: &\quad\so(2\ell + 1),\\ 
\fc_{\ell}: &\quad\usp(\ell):=\spp(2\ell,\C)\cap\uu(2\ell,\C),\\
\fd_{\ell}: &\quad\so(2\ell),
\end{align*}
and of exceptional type  $\fe_6$, $\fe_7$, $\fe_8$, $\ff_4$, $\fg_2$.
Note also that for 
$\fa_{\ell}$ ($\ell\geq 1$), $\fb_{\ell}$ ($\ell\geq 2$), 
	$\fc_{\ell}$ ($\ell \geq 3$) and $\fd_{\ell}$ ($\ell\geq 4$) 
the following isomorphisms (see, e.g., Thm.~X.3.12 in \cite{Helgason78}) 
$\su(2) \cong \so(3) \cong \usp(1)$, $\so(5) \cong \usp(2)$, and 
$\su(4) \cong \so(6)$ are no longer of concern. The same holds 
for the abelian case
$\so(2)$ as well as for the semi-simple one $\so(4) \cong \su(2) \hoplus \su(2)$.

\medskip
\noindent
{\em Pro memoria.}
The classical, compact simple Lie algebras and some forms of their standard matrix representations
may be given as follows:
\begin{center}
\begin{tabular}{l l l l}
\hline\hline\\[-3mm]
algebra & definition and block forms && Lie dimension \\
\hline\\[-2mm]
\multicolumn{3}{l}{$\su(N):=\{a\in\C^{N\times N}\; |\; a^\dagger = - a,\; \tr a=0\}$}  & $N^2 -1$ \\[1mm]
\hline\\[-1mm]
\multicolumn{3}{l}{$\so(N):=\{ a = U\tilde aU^\dagger \in\C^{N\times N}\; |\; \tilde a^t = - \tilde a,\; \tr \tilde a=0\}$} &  $\tfrac{1}{2}N(N-1)$ \\[1mm]
                & \multicolumn{3}{l}{with $U\in SU(N)$} \\[4mm]
\quad $N=2\ell$:         & \multicolumn{3}{l}{$a'=\begin{pmatrix} A & \phantom{-}B\phantom{^t} \\
                        C & -A^t \end{pmatrix}\;$} \\[3mm]
			& \multicolumn{2}{l}{with $ A,B,C\in\C^{\ell\times\ell}, B^t=-B, C^t=-C$} \\[4mm]
\quad $N=2\ell+1$: & \multicolumn{3}{l}{$ a'=\begin{pmatrix} A & \phantom{-}B\phantom{^t} & u\\
                        C & -A^t &v \\ -u^t & -v^t & 0 \end{pmatrix}\;$} \\[5mm]
			& \multicolumn{2}{l}{with $ A,B,C\; \text{as above and}\; u,v\in\C^\ell$} \\[1mm]
			 \multicolumn{3}{l}{\small NB: the representations $\tilde a$ and $a'$ above need not be equal.} \\[2mm]
\hline\\[-2mm]
\multicolumn{3}{l}{$\usp(N/2):=\{a=U\tilde a U^\dagger \in\C^{N\times N}\; |\; J\tilde a = - \tilde a^tJ\}$}
                &  $\tfrac{1}{2}N(N+1)$ \\[1.5mm]
                & \multicolumn{3}{l}{with $J:= \left(\begin{smallmatrix}0\phantom{_\ell} & -\unity_\ell\\ \unity_\ell & \phantom{-}0\phantom{_\ell}\end{smallmatrix}\right), U\in SU(N)$} \\[4mm]
\quad $N=2\ell$: & \multicolumn{3}{l}{$\tilde a=\begin{pmatrix} A & \phantom{-}B\phantom{^t} \\
                        C & -A^t \end{pmatrix}$} \\[4mm]
                        & \multicolumn{2}{l}{with $A,B,C\in\C^{\ell\times\ell}, B=B^t, C=C^t$} \\[4mm]
\hline\hline
\end{tabular}
\end{center}

\medskip
\noindent
By completeness\footnote{
NB: The list of {\em algebras} is indeed complete -- note that in particular spin and pin {\em groups} are
also generated by the algebras $\so(N)$ and $\oo(N)$, respectively \cite{LawMich:1989}.
}
of Cartan's classification above we may summarise as follows:

\begin{corollary}[Candidate List]
Consider a bilinear control system, where the drift and control Hamiltonians
$\{iH_\nu\}$ generate the dynamic system Lie algebra
$\fk \subseteq \su(N)$ in an irreducible representation (\/$\fk'$ trivial\/)
with the additional promise that $\fk$ is simple (e.g.\ due to a connected control system).
Then, being a simple subalgebra of $\su(N)$, the system algebra $\fk$
has to be one of the  candidate compact simple  Lie algebras:
$\su(\ell + 1)$, $\so(2\ell + 1)$,
$\usp(\ell)$, $\so(2\ell)$, $\fe_6$, $\fe_7$, $\fe_8$, $\ff_4$, or $\fg_2$. 
\hfill$\blacksquare$
\end{corollary}

\noindent
For illustration of the Lie algebras of exceptional type, consider
the dimensions of their standard representations (see, e.g., 
p.~218 of Ref.~\cite{Bourb08b}, Ref.~\cite{Min06}, Ref.~\cite{Baez01})
$\fe_6\subset \su(27)$, $\fe_7\subset\usp(28)$,
$\fe_8\subset\so(248)$,
$\ff_4 \subset \so(26)$, and
$\fg_2\subset \so(7)$.
As a final remark on exceptional Lie algebras suffice it to add
that---with the single exception of $\fg_2$---they all
fail to generate groups acting transitively
on the sphere or on $\R^N \setminus \{0\}$. This has been shown in \cite{DiHeGAMM08}
building upon results in \cite{Kra03} to fill gaps in earlier work \cite{Bro73,BW79}.

\medskip
Having listed all the candidates for proper simple subalgebras of $\su(N)$, we now focus
on the set of possible irreducible representations. To this end, in this chapter 
we describe the main results, while all the details shall be explained in 
the Appendix~\ref{app:repr}.
The irreducible representations
of simple (complex) Lie algebras were  already determined by {\'E}lie~Cartan~\cite{Cartan13}. 
This classification is equivalent for the compact simple Lie algebras (or the compact 
connected simple Lie groups), see, e.g., \cite{Bourb08b}. The irreducible simple subalgebras 
of $\su(N)$ are found by enumerating  for all simple Lie algebras all their irreducible 
representations of dimension~$N$. The dimensions of the irreducible representations can 
be  efficiently computed using computer algebra systems such as {\sf LiE}~\cite{LIE222} 
and {\sf MAGMA}~\cite{MAGMA} via Weyl's dimension formula. Following the work of 
Dynkin~\cite{Dynkin57} \nocite{Dynkin2000} (see App.~\ref{incl_rel} and Chap.~6, Sec.~3.2 of Ref.~\cite{GOV94}), 
one can determine the inclusion relations between 
irreducible simple subalgebras of $\su(N)$.
We obtained {\em all} the irreducible simple subalgebras of $\su(N)$ for 
$N\leq 2^{15}=32768$. This significantly extends previous work \cite{McKay81,PST09} 
for $N\leq 9$. The results for $N\leq 16$ are given in Tab.~\ref{tab:su_subalg},
those for $N=2^n$ and $1\leq n\leq 15$ are relegated to Tab.~\ref{tab:two}.
A complete list with all the results for $N\leq 2^{15}$ is attached as Supplementary Material
\cite{anc}.
\begin{table*}[Ht!]
\caption{\label{tab:su_subalg}
The Irreducible Simple Subalgebras of $\su(N)$ for $N\leq 16$} 
\begin{center}
\begin{tabular}{c}
\hline\hline\\[-1mm]
\includegraphics{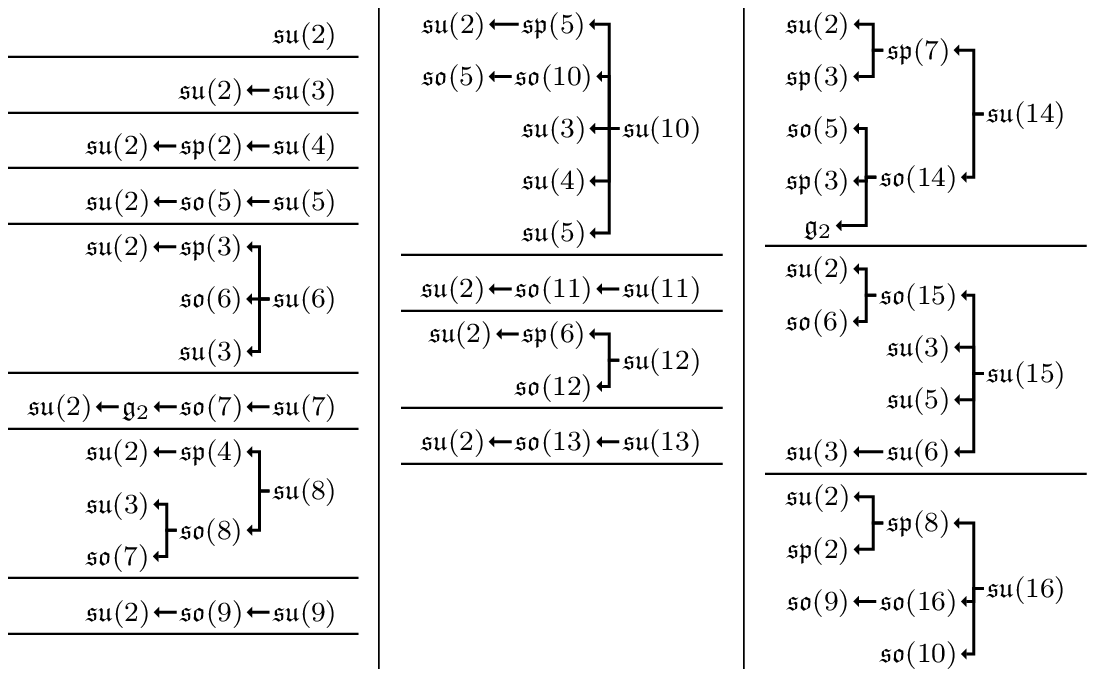}
\\[1mm]
\hline\hline\\[-1mm]
NB: 
$\so(3) \cong \su(2) \cong \usp(1)$, $\so(5) \cong \usp(2)$, and
$\so(6) \cong \su(4)$. 
\end{tabular}
\end{center}
\end{table*}
\begin{remark}
With regard to Tables~\ref{tab:su_subalg} and \ref{tab:two},
the occurrence of $\su(2)$ as an {\em irreducible} simple subalgebra to any $\su(N)$ with $N\geq 2$
is natural from the point of view of spin physics. We identify $\su(N)=\su(2j+1)$, 
where the (non-vanishing) half-integer and integer spin-quantum numbers may take the values 
$j\in\{\tfrac{1}{2}, 1, \tfrac{3}{2}, 2, \dots \}$.
Now to any such $j$ there is an irreducible spin-$j$ representation of the three Pauli
matrices generating $\su(2)$.   
For instance, in $\su(4)$ there is an irreducible spin-$\tfrac{3}{2}$
representation of $\su(2)$ as a proper {\em irreducible} subalgebra ${\su(2)}
\subsetneq\su(4)$.
--- In contrast, the Gell-Mann basis to $\su(2j+1)$ comprises 
a {\em reducible} representation of $\su(2)$ as a subalgebra.
Clearly, the two types of representations are {\em in}\/equivalent.
\end{remark}

\begin{table*}[Hp!]
\caption{\label{tab:two}
The Irreducible Simple Subalgebras of $\su(2^n)$ for $1\leq n\leq 15$} 
\begin{center}
\begin{tabular}{c}
\hline\hline\\[-1mm] 
\includegraphics{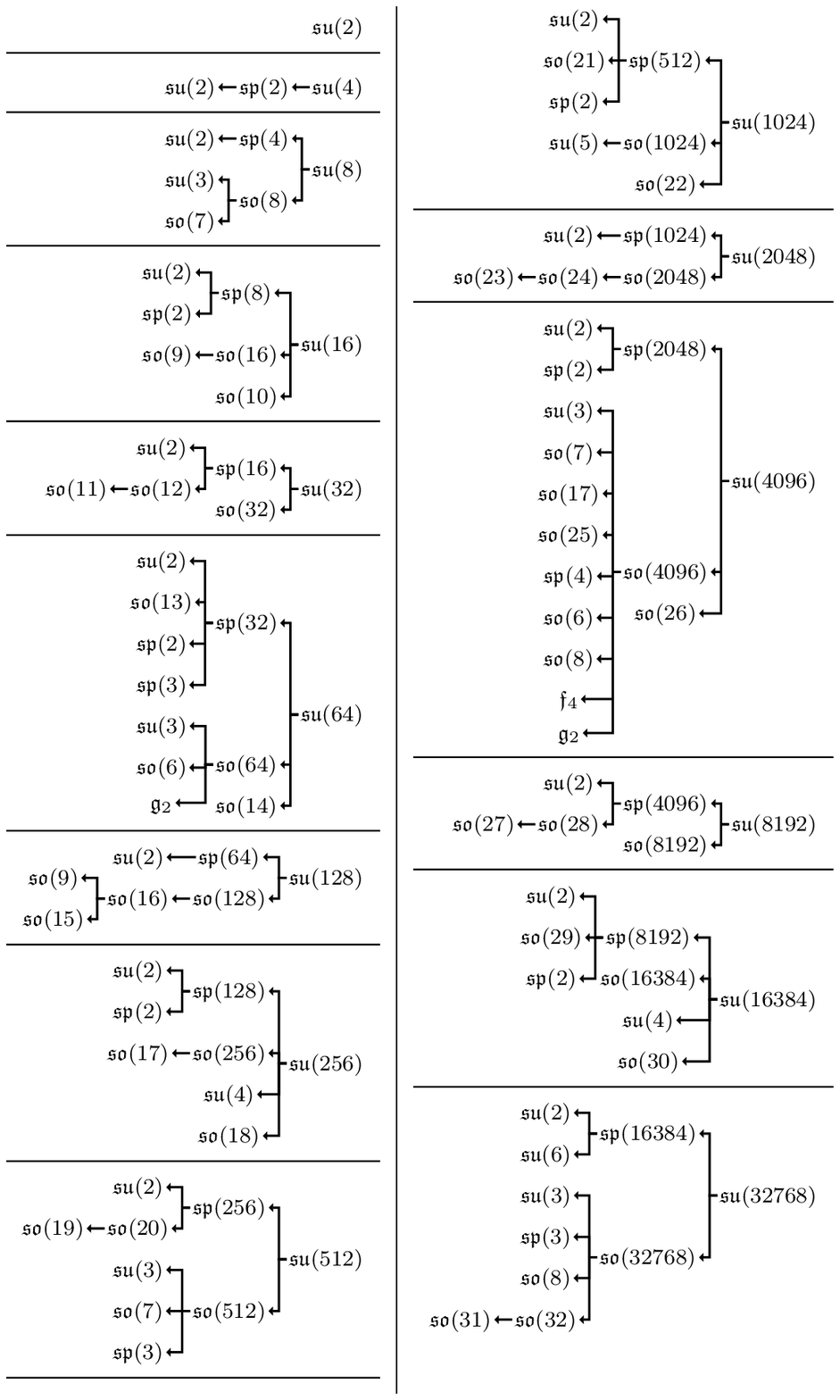}\\[1mm] 
\hline\hline
\end{tabular}
\end{center}
\end{table*}

In the set of irreducible simple subalgebras of $\su(N)$, the subalgebras $\usp(N/2)$ 
with $N$ even and $\so(N)$ play a particularly important role. 
For $N\geq 5$, we discuss the irreducible simple subalgebras of 
$\su(N)$ for $N$ 
even and odd.
If $N\geq 5$ is even, then $\su(N)$ has both
$\usp(N/2)$ and $\so(N)$ as irreducible simple subalgebras.
In addition, $\su(2)\subset\usp(N/2)$ occurs as irreducible
simple subalgebra. We consider two types of trivial cass.
First, if $N\geq 5$ is even and if $\usp(N/2)$, $\so(N)$, and 
$\su(2)\subset\usp(N/2)$ are the only proper irreducible simple subalgebras,
then we say the case is trivial. A trivial example is given by 
$\su(12)$ in Tab.~\ref{tab:su_subalg}. 
If $N\geq 5$ is odd, then $\so(N)$ is an irreducible simple subalgebra 
of $\su(N)$ but $\usp(N/2)$ is not (as $N/2$ is not an integer).
Moreover, $\su(2)\subset\so(N)$ occurs as irreducible
simple subalgebra. 
Second, if $N\geq 5$ is odd and if $\so(N)$ as well as
$\su(2)\subset\so(N)$ are the only proper irreducible simple subalgebras,
then we say the case is trivial.
Examples of such trivial cases are given by 
$\su(5)$, $\su(9)$, $\su(11)$, and $\su(13)$ in Tab.~\ref{tab:su_subalg}.
The irreducible subalgebras $\usp(N/2)$ and $\so(N)$
correspond to the symmetric spaces 
 $\mathrm{SU}(N)/\mathrm{Sp}(N/2)$ and
$\mathrm{SU}(N)/\mathrm{SO}(N)$. These are two of three possible 
symmetric spaces~\cite{Helgason78}
of $\mathrm{SU}(N)$, where the third type does not correspond to a semi-simple subalgebra of $\su(N)$.

We call a representation $\phi$ of a subalgebra $\fk$ symplectic [orthogonal] 
if the subalgebra $\fk$ given in the representation $\phi$ is conjugate to a subalgebra 
of $\usp(N/2)$ [$\so(N)$]. If the representation is neither symplectic nor orthogonal, 
we term it unitary. In abuse of notation, we call
also the subalgebra $\fk$ (w.r.t.\ some fixed but unspecified representation $\phi$) 
symplectic, orthogonal, or unitary, if the respective representation $\phi$
is symplectic, orthogonal, or unitary.\footnote{
	This notation is motivated by the classification of representations 
	and subalgebras as symplectic and orthogonal in Ref.~\cite{Malcev50} and 
	in Chap.~VIII, Sec.~7.5, Def.~2 of Ref.~\cite{Bourb08b}. 
	Classifying representations and subalgebras as unitary appears to be non-standard 
	notation. Unfortunately, the respective representations are also said to be of 
	quaternionic, real, or complex type.}
We emphasise that the classification of a subalgebra
depends on the representations considered,
see also Chap.~IX, App.~II.2, Prop.~3 of Ref.~\cite{Bourb08b}.

The property of a representation to be symplectic [resp.\ orthogonal]
corresponds to the existence of an invariant (non-degenerate)
skew-symmetric [resp.\ symmetric] bilinear form on the space of 
the representation. For irreducible representations in the compact case [e.g.\ for subgroups of 
$\SU(N)$], this correspondence is an equivalence and a proof can be found in Sec.~3.11, Thm.~H of Ref.~\cite{Samelson90}. As an invariant (non-degenerate) bilinear form can either be skew-symmetric or symmetric, it follows that
the same holds for the classification of (irreducible)
symplectic, orthogonal, or unitary representations
(Chap.~IX, Sec.~7.2, Prop.~1 of Ref.~\cite{Bourb08b}):
\begin{lemma}\label{exclusive}
An irreducible representation $\phi(\fk)$ can either be symplectic, or orthogonal, or unitary.
\end{lemma}

\section{From Necessary to Sufficient Conditions for Controllability\label{sec:control-suff}}

While the ramification of {\em mathematically} admissible irreducible simple candidate subalgebras 
may seem daunting, in the following we will eliminate candidates by simple means.
More precisely, we arrive at the following.

\medskip
\medskip
\begin{corollary}[Task List] 
One way of showing full controllability amounts to excluding other candidates of 
irreducible simple subalgebras, which can be\\[1mm]
(1) symplectic, i.e.\ conjugate to a subalgebra of $\usp(N/2)$,\\
(2) orthogonal, i.e.\ conjugate to a subalgebra of $\so(N)$,\\ 
(3) or unitary in the remaining cases.\\[1mm]
In particular, one has to exclude cases like the exceptional ones 
$\fe_6$, $\fe_7$, $\fe_8$ $\ff_4$, $\fg_2$. The unitary, irreducible simple subalgebras
can occur in the cases $\su(\ell+1)\subsetneq\su(N)$ ($\ell\geq 2$), 
$\so(4\ell+2)$, and $\fe_6$.
$\hfill\blacksquare$
\end{corollary}

\medskip
In what follows, the plan is to make use of the fact that in Tabs.~\ref{tab:su_subalg} 
and \ref{tab:two}, most of the irreducible subalgebras are symplectic or orthogonal.
The symplectic and orthogonal ones (including their nested subalgebras!)
will be excluded by merely solving simultaneous systems of linear homogeneous equations, 
which will also exclude the exceptional algebras $\fe_7$, $\fe_8$, $\ff_4$, and 
$\fg_2$, just leaving $\fe_6$.
It appears that for systems of dimension $2^n$, irreducible representations 
of $\fe_6$ cannot be an irreducible simple subalgebra
of $\su(2^n)$ without being a subalgebra to 
an intermediate orthogonal or symplectic algebra.

In principle, the task of identifying the dynamic Lie algebra can also be solved
by algorithms \cite{deGraaf00,Roozemond} available in the computer algebra system  
{\sf MAGMA}~\cite{MAGMA}. 
Yet, here we focus on exploiting algorithms that are more efficient, since they
boil down to solving systems of homogeneous linear equations, which currently
can---in general---be carried to matrix sizes of about $6\cdot 10^4 \times 6\cdot 10^4$ 
(and in extreme cases to $10^5 \times 10^5$)\cite{Huckle}. 
So the algorithms presented here aim at the more specific
task of distinguishing $\su(N)$ from its proper irreducible simple subalgebras, a
task our algorithms are more efficient in.

\subsection{Symplectic and Orthogonal Subalgebras\label{sec:sympl_orth}}
In order to decide on conjugation to irreducible subalgebras which are
symplectic and orthogonal, we need more detail. 
To this end\footnote{Preliminary results were given in the conference papers \cite{SHSZ:2010,ZSSH:2010}.},
recalling the following Lemma will prove useful to apply 
the lines of~\cite{Obata58} in streamlined form leading to an
explicit algorithm.

\begin{lemma}\label{lem:ssym}
(1) Every unitary symmetric matrix $S=S^t\in U(N,\C)$ is unitarily $t$-congruent to the identity, 
i.e.~$S=T^t \unity\, T$ with $T$ unitary.\\[1mm]
(2) Every unitary skew-symmetric matrix $S=-S^t\in U(N,\C)$ with $N$ even
is unitarily $t$-congruent to $J$, i.e.~$S=T^t J T$ with $T$ unitary
and 
\begin{equation}\label{eqn:def-J}
J:= \begin{pmatrix}
0 &-\unity_{N/2}\\
\unity_{N/2} & 0
\end{pmatrix}.
\end{equation}
\end{lemma}
\begin{proof} 
(1) Follows by singular-value decomposition and goes back to Hua (\cite{Hua44}, Thm.~5).
(2) Follows likewise from the same source ({\em ibid.}, Thm.~7).
\hfill$\blacksquare$
\end{proof}

\medskip
\begin{lemma}\label{lem:ssym2}
Suppose $\fk\subset\su(N)$ is simple and $J$ is defined as in Eqn.~\eqref{eqn:def-J}.
Then the element $iH\in\fk$\\[1mm]
 (1) is unitarily conjugate to $i \widetilde H\in \so(N)$,
        where $\widetilde H^t = - \widetilde H$,
	if and only if there exists a symmetric unitary $S$ (so $S\bar{S}=+\unity_N$) satisfying
	$SH + H^t S =0$;\\[1mm]
(2) is unitarily conjugate to $i \widetilde H\in\usp({N}/{2})$ (with $N$ even),
        where $J \widetilde H  = - \widetilde H^t J$,
	if and only if there is a skew-symmetric unitary $S$ (so $S\bar{S}=-\unity_N$) satisfying
	$SH + H^t S =0$.
\end{lemma}

\begin{proof}
First observe that whenever there is a unitary $T$ such that
$THT^\dagger=:\widetilde H$ with $L\widetilde H = - \widetilde H^t L$, 
this is equivalent to 
$$LTHT^\dagger = -(THT^\dagger)^tL \Leftrightarrow
LTH = -\bar{T} H^t T^t L T \Leftrightarrow
(\underbrace{T^tLT}_{S})H = -H^t (\underbrace{T^tLT}_{S}).$$
Now it is easy to establish that the conditions are sufficient (``$\Rightarrow$''):\\[1mm]
(1) Setting $L:=\unity_N$ and $S:=T^tT$ gives $S \bar{S}=T^tT T^\dagger \bar{T} = +\unity_N$.
Thus $S=S^t$ is unitary, complex {\em symmetric} and satifies $SH=-H^tS$.\\[1mm]
(2) Setting $L:=J$ and $S:=T^tJT$ gives 
	$S\bar{S}=T^tJT T^\dagger J \bar{T} = -\unity_N$ by $J^2=-\unity_N$.
	Thus $S=-S^t$ is unitary, {\em skew-symmetric} and satifies $SH=-H^tS$.

\smallskip
\noindent
Moreover the conditions are also necessary (``$\Leftarrow$'') 
by Lemma~\ref{lem:ssym}, because with appropriate respective unitaries $T$\\[1mm]
(1) for $L=\unity_N$ any symmetric unitary matrix $S$ 
can be written as $S=T^tT$; \\[1mm]
(2) for $L=J$ any skew-symmetric unitary matrix $S$ 
can be written as $S=T^tJT$.\\
\phantom{XX}\qquad\hfill$\blacksquare$
\end{proof}

\begin{table}[Ht!]
\begin{center}
\begin{tabular}{llll}
\hline\hline\\[-1mm]
\multicolumn{4}{l}{{\bf Algorithm 3}:
Check conjugation to subalgebras of $\so(N)$ or $\usp(N/2)$}\\[1mm]
\hline\\[-1mm]
&{\em Input:} Hamiltonians $I:=\{i H_d; i H_1,\dots, i H_m\} \subseteq \su(N)$\\[1mm]
&1. For each Hamiltonian $H \in I$ determine all non-singular\\
&\phantom{1.}  solutions to the homogeneous linear equation\\
&\phantom{1.} $\mathcal S_H:=\{ S \in SL(N) | S H + H^t S = 0\}$\\[0mm]
&\phantom{1. $\mathcal S_H:$}$= \{ S \in SL(N) | (H^t\otimes\unity + \unity\otimes H^t)\VEC S  = 0\}$\\[1mm]
&2. $\mathcal S  := \bigcap_{H{\in}I} \mathcal S_H$\\[1mm]
&{\em Output:} (a) $ \exists S \in \mathcal{S} \text{ s.t.\ } S \Bar S = +\unity$ $\Leftrightarrow$ $\fk \subseteq \so(N)$\\
&\phantom{\em Output:} (b) $ \exists S \in \mathcal{S} \text{ s.t.\ } S \Bar S = -\unity$ $\Leftrightarrow$ $\fk \subseteq \usp(N/2)$\\
&\phantom{\em Output:} (c) \hspace{1pt}$ \nexists S \in \mathcal{S}$ $\Rightarrow$ $\fk \not\subseteq \so(N)$ and $\fk \not\subseteq \usp(N/2)$\\
&\phantom{Output:} The cases (a) and (b) are mutually exclusive\\
&\phantom{Output:} if the centraliser of 
$I$ is trivial.\\[1mm] 
\hline\\[-1mm]
& The complexity is roughly $\mathcal{O}(N^6)$, as in Liouville space \\ 
& $N^2$ equations have to be solved by $LU$ decomposition ($N:=2^n$).\\[1mm]
\hline\hline
\end{tabular}
\end{center}
\end{table}

\medskip
In the context of filtering simple subalgebras, 
Lemma~\ref{lem:ssym2} can be turned into the powerful {\bf Algorithm~3}.
It boils down to checking a system of homogeneous linear equations for
solutions $S$ satisfying $SH_\nu=-H_\nu^tS$ for {\em all $iH_\nu\in\fk$ simultaneously}:
if $S$ is a solution with $S \bar{S}=+\unity$, the subalgebra $\fk$ of $\su(N)$ generated by
the $\{i H_{\nu}\}$ is conjugate to a subalgebra of $\so(N)$,
while in case of $S \bar{S}=-\unity$, $\fk$ is conjugate to a subalgebra of $\usp({N}/{2})$.

\begin{remark}
By irreducibility of $\fk$ (via {\bf Algorithm~2}), those subgroups generated by
$\fk\subset\su(2^n)$ with a unitary representation equivalent to its complex conjugate
are limited to orthogonal and  symplectic ones: it 
follows from Schur's Lemma that $S \bar{S}=\pm \unity$ are in fact {\em the only types of solutions}
for $SH=-H^tS$ with $iH\in\fk$,
as nicely explained in Lem.~3 of Ref.~\cite{Obata58}.
Lemma~\ref{exclusive} (of this paper) explains why these solutions
are mutually exclusive.
Due to the irreducibility, the matrix $S$ is unique up to a scalar factor
$c \in \C$ with $c\bar{c}=1$.
\end{remark}

Conjugation to the  symplectic algebras has also been treated in Ref.~\cite{BW79}
by solving a system of linear equations, while Ref.~\cite{SchiSoLea02} resorted to 
determining eigenvalues for discerning the unitary case from conjugate symplectic or orthogonal
subalgebras. --- The results can be summarised and extended as follows:
\begin{theorem}[Candidate Filter I] \label{thm:filter1}
Consider a set of Hamiltonians $\{i H_\nu\}$ generating the dynamic algebra
$\fk \subseteq \su(N)$ with the promise 
(by {\bf Algorithm~2} and e.g.\ due to a connected control system) that
$\fk\subseteq \su(N)$ is given in an irreducible representation 
and $\fk$ is simple.
If in addition {\bf Algorithm~3} has but an empty set of solutions, then 
$\fk$ is neither conjugate to a simple subalgebra 
of $\usp({N}/{2})$ nor of $\so(N)$.
In particular, $\fk$ is none of the following simple Lie algebras:
 $\fe_7$, $\fe_8$, $\ff_4$, or $\fg_2$.
\end{theorem}
\begin{proof}
The cases $\so(N)$ and $\usp(N)$ are settled by Lemma~\ref{lem:ssym2}.
The cases $\fe_7$, $\fe_8$, $\ff_4$, and $\fg_2$ follow from the elaborate classification
of Malcev~\cite{Malcev50} (see also, e.g., \cite{Dynkin57,BP70,McKay81} and 
Theorem~\ref{Malcev} in Appendix~\ref{incl_rel}), as
an irreducible representation of $\fe_8$, $\ff_4$, or $\fg_2$ is always conjugate to
a subalgebra of $\so(N)$, while an irreducible representation of $\fe_7$ is conjugate
either to a subalgebra of $\so(N)$ or of $\usp(N/2)$.
\hfill$\blacksquare$
\end{proof}

\subsection{Unitary Subalgebras\label{complex}}

It also follows from Malcev~\cite{Malcev50} 
(again, see also \cite{Dynkin57,BP70,McKay81} and Theorem~\ref{Malcev} in Appendix~\ref{incl_rel})
that only the subalgebras
$\su(\ell+1)$ ($\ell\geq 2$), $\so(4\ell+2)$, and  $\fe_6$ can have
unitary representations. 
One can immediately deduce
from the Tables~\ref{tab:su_subalg} and \ref{tab:two}
the following
\begin{corollary}
The Lie algebras $\su(2^n)$ do not possess (proper) unitary, irreducible simple subalgebras if
$n\in\{1,2,3,5,7,9,11,13,15\}$. In these cases ($n\neq 1$) and
under the conditions of Theorem~\ref{thm:filter1},
{\bf Algorithm~3} provides a necessary and sufficient
criterion for full controllability. \hfill$\blacksquare$
\end{corollary}

\medskip
We checked by explicit computations that $\fe_6$ does not occur
as a unitary, irreducible simple subalgebra of $\su(2^n)$ for $n\leq 100$, i.e.\ for qubit systems
with up to $100$ qubits. Thus one might conjecture that $\fe_6$ does not occur as a unitary, irreducible simple subalgebra for qubit systems in general.

We present an example of a control system whose dynamic algebra is 
a (proper) unitary subalgebra of $\su(2^4)$:

\begin{example}\label{ex:counter1}
Consider a bilinear control system on $\su(16)$ with four subsystems given by 
$\su(2) \hoplus \su(2) \hoplus \su(2) \hoplus \su(2)$. The local dynamic  algebra
is given by $$\langle i\mathrm{XIII}, i\mathrm{YIII}, i\mathrm{ZIII}, i\mathrm{IIIX}, i\mathrm{IIIY}, i\mathrm{IIIZ}
\rangle_{\rm Lie}.$$
In addition, we have a drift Hamiltonian 
$H_d=\mathrm{XXII}+\mathrm{YYII}+\mathrm{IXXI}+\mathrm{IYYI}+\mathrm{IIXX}+\mathrm{IIYY}$
(Heisenberg-{\XX} interaction).
The control system 
\begin{center}
\includegraphics{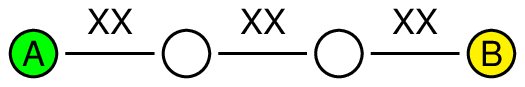}
\end{center}
is connected and  acts irreducibly.
The dynamic algebra $\fk=\so(10)$ is simple and
a (proper) unitary subalgebra of $\su(16)$. 
\end{example}

\subsection{System Algebras Comprising Local Actions $\su(2)^{\oplus n}$}

We now discuss the set of local unitary transformations  
$\mathrm{SU}(2)^{\otimes n}\subseteq \mathrm{SU}(2^n)$
and its Lie algebra $\su(2)^{\oplus n}\subseteq \su(2^n)$ 
where both are given in their respective standard representation, i.e.~as
$n$-fold Kronecker product and $n$-fold Kronecker sum (see Sec.~\ref{sec:tensor})
$$\su(2)\hoplus\su(2)\hoplus\cdots\hoplus\su(2).$$
What is the classification of $\su(2)^{\oplus n}$ w.r.t.\ symplectic, orthogonal, and 
unitary subalgebras? 
We obtain from Thm.~3 of Ref.~\cite{ZGB} (see also~\cite{BB04,BDMB}):
\begin{lemma}\label{lem:loc}
For the algebra $\su(2)^{\oplus n}$ given in its (irreducible) standard representation there are
two cases: 
   (1) if $n$ is odd, it is a symplectic  subalgebra of $\su(2^n)$ in the sense of being
       conjugate to a subalgebra of $\usp(2^{n-1})$, and 
   (2) if $n$ is even, it is an orthogonal subalgebra  in the sense of being
       conjugate to a subalgebra of $\so(2^n)$.
\end{lemma}
\begin{proof}
Let $\phi$ denote an irreducible representation of a compact Lie group $\mathrm{G}$.
Then for the Frobenius-Schur indicator (Chap.~IX, App.~II.2, Prop.~4 of Ref.~\cite{Bourb08b}) 
one finds
\begin{equation*}
\int_G \mathrm{Tr}[\phi^2(g)]\; dg =
\begin{cases}
-1 & \Leftrightarrow \quad\text{$\phi$ is a {\em symplectic} representation}\\
+1 & \Leftrightarrow \quad\text{$\phi$ is an {\em orthogonal} representation}\\
\phantom{+}0 & \Leftrightarrow \quad\text{$\phi$ is a {\em unitary} representation}\\
\end{cases}
\end{equation*}
Let $\psi$ denote the standard representation of the Lie group $\mathrm{H}=\SU(2)^{\otimes n}$.
Ref.~\cite{ZGB} proves that
$\int_H \mathrm{Tr}[\psi^2(h)]\; dh=(-1)^n$. 
\hfill$\blacksquare$
\end{proof}

If the subsystems of a control system  are independently fully controllable then
it follows from Lemma~\ref{lem:loc} that some cases can be excluded:

\begin{lemma}\label{lem:loc2}
Assume that the dynamic algebra $\fk \subseteq \su(2^n)$ is irreducible and simple, and that
the subsystems $\su(2)$ are independently fully controllable 
[i.e.\ $\fk \supseteq \su(2) \hoplus \su(2) \hoplus \cdots \hoplus \su(2)$].
If $n$ is odd [resp.\ even] then $\fk$ 
is not an orthogonal [resp.\ symplectic]
subalgebra. 
\end{lemma}
\begin{proof}
We remark that $\fh=\su(2) \hoplus \su(2) \hoplus \cdots \hoplus \su(2)$ is given in an irreducible representation.
It follows from 
the discussion prior to Lemma~\ref{exclusive} that
$\fh$ has an invariant (non-degenerate) skew-symmetric bilinear form
and no invariant (non-degenerate) symmetric bilinear form if $n$ is odd.
Therefore, the dynamic algebra $\fk$ cannot have an invariant (non-degenerate) 
symmetric bilinear form and the Lemma follows for odd $n$. The case of even $n$ is similar.
\hfill$\blacksquare$
\end{proof}

Unfortunately\footnote{In Thm.~3 and 4 of the conference paper Ref.~\cite{ZSSH:2010} we 
  incorrectly gave more general results for dynamic algebras which contain 
  a non-zero subset of the local operations. But in light of Examples~\ref{ex:counter1} 
  and \ref{ex:counter2} the more general results in Ref.~\cite{ZSSH:2010} are not correct, 
  as the non-zero subset is in general not given in an \emph{irreducible} representation.}, 
Lemma~\ref{lem:loc2} is no longer true if the dynamic algebra contains only
a non-zero subset of the local operations:
\begin{example}\label{ex:counter2}
Consider a bilinear control system on $\su(8)$ with three subsystems given by 
$\su(2) \hoplus \su(2) \hoplus \su(2)$. The local dynamic  algebra is 
$\langle i\mathrm{XII}, i\mathrm{YII}, i\mathrm{ZII}
\rangle_{\rm Lie}$.
In addition, we have a drift Hamiltonian $ H_d=\mathrm{XXI}+\mathrm{YYI}+\mathrm{IXX}+\mathrm{IYY}$.
The control system 
\begin{center}
\includegraphics{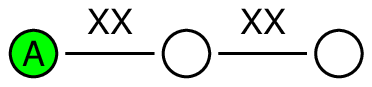}
\end{center}
is  
connected and  acts irreducibly.
The dynamic algebra $\fk=\so(7)$ is simple and
an orthogonal subalgebra. We emphasise that as a consequence of Lemma~\ref{lem:loc2} this
would have not been possible if $\fk \supseteq \su(2) \hoplus \su(2) \hoplus \su(2)$.
\end{example}

\subsection{A Necessary and Sufficient Symmetry Condition\label{sufficient}}
In this subsection we present a necessary and sufficient symmetry criterion
for full controllability of control systems contained in $\su(N)$. 
To this end,
we introduce some additional notation: Assume that $\phi$ is a representation
of a compact Lie algebra of dimension $N$. The tensor square $\phi^{\otimes 2}:=
\phi\otimes \unity_N + \unity_N \otimes \phi$ decomposes as
$\phi^{\otimes 2}=\Alt^2 \phi \oplus \Sym^2 \phi$,
where the alternating square $\Alt^2 \phi$ and the symmetric square $\Sym^2 \phi$ are the restrictions
of $\phi^{\otimes 2}$ to the antisymmetric and the symmetric subspace, respectively.
More details on this notation is given in Appendix~\ref{AltSym}.
We arrive at

\begin{theorem}\label{tensor_sq}
Assume that $\fk$ is a subalgebra of $\su(N)$ and denote
by $\Phi$ the standard representation of $\su(N)$.
Then, the following are equivalent:\\
(1) $\fk = \su(N)$.\\
(2) The restrictions of $\Alt^2 \Phi$ and $\Sym^2 \Phi$ to the subalgebra $\fk$ are both irreducible.\\
(3) The commutant $\mathrm{comm}_{\Phi^{\otimes 2}}(\fk)$ 
w.r.t.\ the tensor square $\Phi^{\otimes 2}$ has dimension two.
\end{theorem}
\begin{proof}
(1) $\Rightarrow$ (2) follows by Theorem~\ref{thm_altsym} in Appendix~D. We prove (2) $\Rightarrow$ (1).
As the restriction $(\Alt^2 \Phi)|_{\fk}$ of $\Alt^2 \Phi$ to 
$\fk$ is irreducible, we get that the restriction $\Phi|_{\fk}$ of $\Phi$ to $\fk$ is also irreducible. Otherwise, 
$\Phi|_{\fk}=\phi_1 \oplus \phi_2$ would be reducible and, as a consequence,
$(\Alt^2 \Phi)|_{\fk}=\Alt^2(\phi_1 \oplus \phi_2)=\Alt^2 \phi_1 \oplus (\phi_1 \otimes \phi_2) \oplus \Alt^2 \phi_2 $
would also be reducible (which is impossible). Lemma~\ref{lem:centraliser} implies the
centraliser $\fk'$ of $\fk$ in $\su(N)$ is trivial, thus
by Lemma~\ref{lem:semi-simplicity} $\fk$ is semisimple.
Now (1) follows by Theorem~\ref{thm_altsym} in Appendix~D.
Moreover Thm.~1.5 of Ref.~\cite{Ledermann} says that the dimension of the commutant of a representation
$\phi$ is given by $\sum_i m_i^2$ where the $m_i$ are the multiplicities of the irreducible components of $\phi$.
As we consider the representation 
$(\Phi^{\otimes 2})|_{\fk}=(\Alt^2 \Phi)|_{\fk} \oplus (\Sym^2 \Phi)|_{\fk}$, 
the equivalence of (2) and (3) readily follows.
\hfill$\blacksquare$
\end{proof}

We now show that condition (3) of Theorem~\ref{tensor_sq} can be easily tested
using a set of Hamiltonians $\{i H_\nu\}$ generating the dynamic algebra 
$\fk \subseteq \su(N)$. Therefore, we prove that the commutant of 
$\{ (iH_\nu)\otimes \unity_N + \unity_N \otimes (iH_\nu)\}$ 
is equal to $\mathrm{comm}_{\Phi^{\otimes 2}}(\fk)$. 
Obviously, the latter commutant is contained in the former.
Let $s\in \gl(N^2,\C)$ be an element of the former commutant. Then by
Jacobi's identity $\big[[a,b],s\big]+\big[[b,s],a\big]+\big[[s,a],b\big]=0$,
$s$ commutes with all commutators
\begin{gather*}
[(iH_{\nu})\otimes \unity_N + \unity_N \otimes (iH_{\nu}),
(iH_{\mu})\otimes \unity_N + \unity_N \otimes (iH_{\mu})]\\
=[iH_{\nu},iH_{\mu}]\otimes \unity_N + \unity_N \otimes [iH_{\nu},iH_{\mu}],
\end{gather*}
and by induction, $s$ is also contained in the latter commutant.

\medskip
Together with Theorem~\ref{tensor_sq},
we thus obtain a {\em necessary and sufficient symmetry condition} for full controllability
as a theoretical main result:

\begin{corollary}\label{cor_tensor}
Consider a set of Hamiltonians $\{i H_\nu\,|\,\nu=d,1,2,\dots\}$ generating the dynamic algebra 
$\fk \subseteq \su(N)$. The corresponding control system is fully controllable in the sense $\fk=\su(N)$,
if and only if the joint commutant of $\{ (iH_\nu)\otimes \unity_N + \unity_N \otimes (iH_\nu)\,|\,\nu=d,1,2,\dots\}$
has dimension two. \hfill$\blacksquare$
\end{corollary}

In spite of the beauty of simplicity of this result,
from an algorithmic point of view the above symmetry condition is
currently not appealing:
In Corollary~\ref{cor_tensor} one would have to compute the commutant of $N^2\times N^2$ matrices
as compared to $N\times N$ matrices in the test for the lack of symmetry in
{\bf Algorithm~2}. Thus the complexity of testing for Corollary~\ref{cor_tensor} 
would be the square of the complexity of {\bf Algorithm 2}.
Even in moderately-sized examples one has to save computer memory 
by methods of sparse matrices due to the larger matrices. 
In larger examples, testing for Corollary~\ref{cor_tensor}
gets impractical. Yet compared with potential conditions involving even higher
tensor powers, one should consider Corollary~\ref{cor_tensor} as a fortunate incidence.

In order to characterise the commutant of Corollary~\ref{cor_tensor}
in further detail, we introduce an $N^2\times N^2$ permutation matrix $K_{N,N}$ also known as
{\em commutation matrix} \cite{HJ2,HS81}. Let $e^a$ denote the vector such that
$(e^a)_b = \delta_{a,b}$ with $a,b \in \{1,\ldots,N^2\}$. We define $K_{N,N}$ by the permutation
$K_{N,N}\cdot e^a = e^{\pi(b)}$ where one has 
$\pi(N\cdot i+j+1)=j\cdot N+i+1$ and $i,j \in \{0,\ldots,N-1\}$. 
The commutation matrix operates on the vec-representation \cite{HJ2} of 
an $N\times N$ matrix $A$ as the transposition operator: $K_{N,N}\cdot \opvec{(A)} = \opvec{(A^t)}$.

\begin{lemma}
The commutant of $\{ (iH_\nu)\otimes \unity_N + \unity_N \otimes (iH_\nu)\}$
always contains the elements $\unity_{N^2}$ and $K_{N,N}$.
\end{lemma}
\begin{proof}
As the identity matrix $\unity_{N^2}$ always commutes, we have only to prove that
$K_{N,N}$ is contained in the commutant. Sec.~3 of Ref.~\cite{HS81}
says that $K_{N,N} (A\otimes B) = (B\otimes A) K_{N,N}$ for all $N\times N$ matrices
$A$ and $B$ and thereby $K_{N,N} (A\otimes B+B\otimes A) = (A\otimes B+B\otimes A) K_{N,N}$. 
In particular one finds
$K_{N,N} (A\otimes \unity_N+\unity_N\otimes A) = (A\otimes \unity_N+\unity_N\otimes A) K_{N,N}$ 
and the Lemma is proven.
\hfill$\blacksquare$
\end{proof}

The operator $K_{N,N}$ has two eigenspaces (see Sec.~4.2 of Ref.~\cite{HS81}):
The first one is given by the symmetric subspace (i.e.\ `bosons') and 
has the eigenvalue $+1$ with multiplicity $N(N+1)/2$. For even $N$, the permutation-symmetric subspace
is equivalent to the Lie algebra $\usp(N/2)$. The second one is given by the antisymmetric
subspace (i.e.\ `fermions') and has the eigenvalue $-1$ with multiplicity $N(N-1)/2$.
The permutation-antisymmetric subspace is equivalent to the Lie algebra $\so(N)$.

\medskip
The methods of this subsection thus shed new light on the symplectic and orthogonal subalgebras
(see Subsection~\ref{sec:sympl_orth}). Prop.~3.5 of Ref.~\cite{KW96}
(see also p.~446 of Ref.~\cite{FH91})
says that an irreducible representation $\phi$ of a compact simple Lie algebra $\fg$ is either symplectic or orthogonal if and only
if its tensor square $\phi^{\otimes 2}$ contains the trivial representation of $\fg$
exactly once.
In particular, the irreducible representation $\phi$ is symplectic (resp.\ orthogonal) if the trivial representation
occurs exactly once in $\Alt^2 \phi$ (resp.\ $\Sym^2 \phi$).  A similar condition is given by
Prop.~4.2 of Ref.~\cite{KW96}: An irreducible representation $\phi$ of a compact simple Lie algebra
$\fg$ is either symplectic or orthogonal if and only if 
its tensor square $\phi^{\otimes 2}$ contains the (irreducible) adjoint representation of $\fg$ at least once.


\section{Simulability}\label{sec:simulability}
Simulating quantum systems \cite{Fey82,Lloyd96,BN09} is a promising mid-term perspective,
because the accuracy demands are easier to come by than the \/`error-correction threshold\/'
for actual quantum computing. Another practical advantage lies in the fact that
sometimes the simulating systems allow for separating control parameters in the analogue 
that in the original (be it classical or quantum) cannot be tuned independently.

This section exploits that the dynamical algebra captures all the
key properties of the dynamical system to be studied. More precisely, the question whether
(and to which extent) one quantum system can simulate another one can be answered 
by analysing the Lie-subalgebra structure of systems with a given dimension.
Recently Kraus {\em et al.}~have explored whether
target quantum systems can be universally simulated on translationally invariant lattices
of bosonic, fermionic, and spin systems \cite{kraus-pra71}. Based on the branching 
diagrams of simple subalgebras to $\su(N)$, here we take a more general
approach pursuing the question which type of quantum system can simulate a given one
with least overhead in state-space dimension. In particular, we also allow for 
effective many-body interactions to be simulated by pair-interactions. --- To this end,
the reader may wish to resort to the more general notion of tensor-product structures 
in Appendix~\ref{sec:tensor2} first.

In quantum simulation, one of the first natural questions to ask is whether and under 
which conditions a controlled quantum dynamical system $\Sigma_a$ can simulate another
(controlled or uncontrolled) dynamical system $\Sigma_b$ given as bilinear control systems
with $\mu=a,b$ on density matrices $\rho_\mu$
\begin{equation}\label{eqn:bilinear_contr3}
        \dot \rho_\mu(t) = -i\Big[\big(H_0^\mu + \sum_{j=1}^m u^\mu_j(t) H^\mu_j\big) \;,\; \rho_\mu(t) \Big]
			\quad\text{with}\quad \rho_\mu(0)=\rho_{\mu o}\quad.
\end{equation}
The dynamic Lie algebras $\fk_a$ and $\fk_b$ are given by the respective Lie closures as
\begin{equation}
\fk_\mu := \expt{iH_0^\mu, iH^\mu_j\, |\, j=1,2,\dots,m}_{\rm Lie}
\end{equation}
thus entailing the reachable sets take the form of $\bK_\nu$-subgroup orbits 
as in Eqn.~\eqref{eqn:reach-k-orbit} 
\begin{align}
\reach(\rho_{ao}) &:= \{K_a \rho_{ao} K_a^\dagger\,|\, K_a\in \bK_a := \exp \fk_a\}\quad\text{and}\\
\reach(\rho_{bo}) &:= \{K_b \rho_{bo} K_b^\dagger\,|\, K_b\in \bK_b := \exp \fk_b\}\;.
\end{align}

An obvious requirement is that for any initial state $\rho_{bo}$ of system $\Sigma_b$ 
leading to the dynamics $\rho_b(t)\in\reach(\rho_{bo})$ there is an initial state 
$\rho_{ao}$ of system $\Sigma_a$ such that under the dynamics of $\Sigma_a$ one has
\begin{equation}\label{eqn:sim1}
\rho_b(t)\in\reach(\rho_{ao})\quad\forall\; t\geq 0\;.
\end{equation}
This requirement is obviously fulfilled by the following sufficient condition:
\begin{proposition}\label{prop:sim-gen}
A dynamic bilinear control system $\Sigma_a$ with dynamical algebra $\fk_a$ can simulate 
another dynamic system $\Sigma_b$ with dynamical algebra $\fk_b$ if $\fk_a\supseteq\fk_b$.
\end{proposition}
\begin{proof}
Clearly $\fk_a\supseteq\fk_b$ implies $\bK_a\supseteq\bK_b$ and thus 
$\reach(\rho_{ao})\supseteq\reach(\rho_{bo})$, which in turn ensures that
Eqn.~\eqref{eqn:sim1} is fulfilled for any choice of initial states. 
\hfill$\blacksquare$
\end{proof}
In particular, if system $\Sigma_b$ is uncontrolled it can be simulated if
its drift Hamiltonian $H_0^b$ can be simulated, i.e.\ provided $iH_0^b \in \fk_a$.

Two dynamic bilinear control systems $\Sigma_a$ and $\Sigma_b$ are said to be 
{\em dynamically equivalent} independent of the respective initial states $\rho_{\mu 0}$
if and only if they can mutually simulate one another, 
i.e.\ if $\fk_a\supseteq\fk_b$ and $\fk_b\supseteq\fk_a$ so $\fk_a=\fk_b$
(up to isomorphism). 

\begin{remark}
It is important to note that in the {\em special case of pure states}, where by construction 
$\rho(t)=\rho^2(t)$, it suffices that, e.g., a system $\Sigma_a$ has 
the dynamic Lie algebra $\fk_a=\usp(N/2)$ in order to simulate 
system $\Sigma_b$ with $\fk_b=\su(N)$, because the unitary orbit of
any pure state $\rho_0=\ketbra \psi \psi$ coincides with its symplectic orbit
for $N$ even
\begin{equation}
\mathcal O_{\SU(N)}(\ketbra \psi \psi) =  \mathcal O_{Sp(N/2)}(\ketbra \psi \psi)
				\quad\forall\; \ket\psi\in\mathcal H.
\end{equation}
This is equivalent to a well-known result stating that for $N$ even, a system
is {\em pure-state controllable} as soon as its system algebra encapsulates
the symplectic one \cite{AA03}. --- Since we are interested in general results beyond
pure states, the notion of full controllability maintained in this work is 
full operator controllability unless specified otherwise. Also for simulability
we do not confine the state space to pure states henceforth. 
\end{remark}

\begin{proposition}\label{prop:sim-least}
Consider two dynamic systems $\Sigma_a$ and $\Sigma_b$ whose respective dynamic Lie algebras 
$\fk_a$ and\/ $\fk_b$ shall be irreducible over a given Hilbert space $\mathcal H$. 
Then $\Sigma_a$ simulates $\Sigma_b$ irreducibly and with least overhead in the
very $\mathcal H$ given, if any interlacing system $\Sigma_i$ with irreducible algebra $\fk_i$ fulfilling
\begin{equation}
\fk_a\supseteq\fk_i\supseteq\fk_b
\end{equation}
enforces (up to isomorphism) $\fk_i=\fk_a$ or $\fk_i=\fk_b$ or trivially both.
\end{proposition}

\noindent
{\em Caveat.}
Note that the term \/`with least overhead\/' crucially depends on the Hilbert space
{\em given a priori}: Thus there may be extreme
realisations.  For instance, in a fully controllable system of say $14$ qubits 
with dynamic algebra $\su(16\, 384)$ there is an {\em irreducible} way 
to simulate a fully controllable $\su(4)$-system of two qubits 
(or just a single spin-$\tfrac{3}{2}$ with control over all multipole moments) 
with \/`least overhead\/' in $\su(2^{14})$,  
see 
the penultimate entry in Tab.~\ref{tab:two}. Realisations of this type may 
not be very useful in practice, yet relate to the context of code spaces.

\medskip
Here, we have dealt with quantum simulation of {\em unobserved control systems}. 
Now we illustrate the above findings by examples.
Later, in Sec.~\ref{sec:outlook}, we will give an outlook on 
a weaker notion of quantum simulation of {\em  observed control systems} 
with respect to expectation values by given sets of observables. 

\section{Worked Examples}\label{sec:w-examples}

\subsection{Dynamic Systems with Orthogonal Algebras}
\begin{figure}[Ht!]
\begin{center}
\includegraphics[height=7mm]{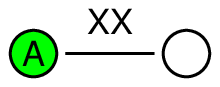}\hspace{9mm}
\includegraphics[height=7mm]{ex4}\hspace{9mm}
\includegraphics[height=7mm]{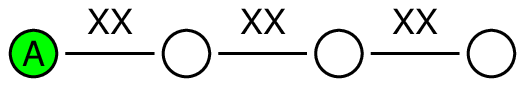}\hspace{2mm}\\[2mm]
\includegraphics[height=7mm]{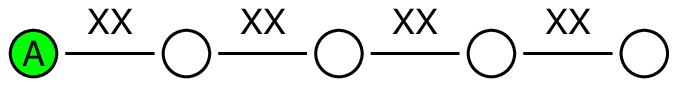}\hspace{2mm}
\end{center}
\caption{Heisenberg-{\XX} spin chains with $n$ spins-$\tfrac{1}{2}$ and odd-order orthogonal system algebras 
$\so(2n+1)$ require one locally controllable qubit at the end. A full series can
be constructed, the first examples being $\so(5)\cong\usp(4/2)$, $\so(7)$, $\so(9)$, and $\so(11)$.
For $n=1$ one gets $\so(3)\cong\su(2)$.
\label{fig:sys-so2k+1}}
\end{figure}

Take spin chains of \mbox{$n$ spins-$\tfrac{1}{2}$} with Heisenberg-{\XX} (and \XY) interactions 
and a {\em single} locally controllable
qubit at one end. These instances serve as convenient topologies to simulate a full series of {\em odd-order} orthogonal 
algebras $\so(2n+1)$ for $n$ qubits. The first instances are shown in Fig.~\ref{fig:sys-so2k+1}.
\begin{proposition}\label{prop:so2n+1}
Heisenberg-{\XX} chains of $n$ spin-$\tfrac{1}{2}$ qubits ($n\geq 1$) and a {\em single} locally controllable
qubit at one end give rise to the dynamic system algebras $\so(2n+1)$ as irreducible 
subalgebras embedded in $\su(2^n)$.
\end{proposition}
\begin{proof}
In view of later applications, the proof is kept constructive. 
For better readability, let $x$, $y$, and $z$ denote Pauli matrices.

First, as a foundation for induction, the case $n=2$ can be settled by
direct calculation to verify 
\begin{equation}\label{eqn:so5}
i\expt{x1,y1,(xx+yy)}_{\rm Lie}=i\{x1,y1,z1,xx,yy,xy,yx,zx,zy,1z\}\rep\so(5)\;,
\end{equation}
where the final identity can be corroborated by {\bf Algorithm~3}
as will be illustrated in Eqn.~\eqref{eqn:so5-S} below.

Second, for the induction from $(n-1)$ to $n$, where the drift Hamiltonian 
is extended by the final Heisenberg coupling between the qubit pair $(n-1),n$ 
to take the form $H_0:= \sum_{k=1}^{n-1} x_kx_{k+1} + y_ky_{k+1}$,
observe that all the algebra elements for $n-1$ qubits re-occur. 
Upon twice commuting with $z1\cdots 1$ arising at the controlled end, 
the first pair coupling term $x_1x_2 + y_1y_2$ can be recovered: 
$
\adr^2_{i\,z_1}(i\sum_{k=1}^{n-1} x_kx_{k+1} + y_ky_{k+1})
	= \adr_{i\,z_1}(-i(y_1x_2-x_1y_2))
	= -i(x_1x_2+y_1y_2)
$
and then by virtue of Eqn.~\eqref{eqn:so5} also $1z1\cdots 1$ and thus recursively
all the terms in the Lie closure at the stage $n-1$.

Third, once having embedded the $(n{-}1)$-qubit algebra into the $n$-qubit system, 
the induction boils down to including the coupling term $x_{n-1}x_n + y_{n-1}y_n$,
which takes $1\cdots 1z_{n-1}1$ to $1\cdots 1z_n$. 
Writing braces $\xy$ whenever one has the choices $\{x,y\}$ one gets
the following complete list:\\[3mm]
\begin{tabular}{l l l}
\multicolumn{2}{l}{The Pauli-basis elements for $\so(2n+1)$}\\
\hline\hline\\[-3mm]
$2$ terms & $\xy 1 \cdots 1$ &\\[1mm]
$n$ terms & $z 1 \cdots 1$ & $1 z 1 \cdots 1$  etc \\
\hline\\[-3mm]
$2$ terms & $z \xy 1 \cdots 1$ &\\[1mm]
$4(n-1)$ terms & $\xy \xy 1 \cdots 1$ & $ 1 \xy \xy 1 \cdots 1$  etc \\[1mm]
\hline
$\vdots$ & $\vdots$ & $\vdots$\\
$2$ terms & $ z z \cdots z \xy_k 1 \cdots 1$ &\\[1mm]
$4(n-k+1)$ terms &$ \xy z \cdots z \xy_k 1 \cdots 1$ & $1 \xy_2 z \cdots z \xy_{k+1} 1 \cdots 1$  etc\\
$\vdots$ & $\vdots$ & $\vdots$\\
\hline\\[-3mm]
$2$ terms & $ z z \cdots z \xy_n$\\[1mm]
$4$ terms & $ \xy z \cdots z \xy_n$\\[1mm]
\hline\hline
\end{tabular}\\[2mm]
Finally counting terms gives a total of $2n + n + 4\sum_{k=1}^{n-1} (n-k) = 3n + 4\sum_{k=1}^{n-1} k 
	= 3n + 2n(n-1) = 2n^2 + n = \dim \so(2n+1) = n(2n+1)$ elements to span the basis
of the Hamiltonians $H_\nu$ generating $\expt{iH_\nu}_{\rm Lie} = \so(2n+1)$.
So for all \mbox{$n$-spin-$\tfrac{1}{2}$} Heisenberg-$\XX$ chains controlled locally at one end 
we have obtained a constructive scheme to determine  
irreducible representations of their respective dynamic Lie algebras $\so(2n+1)$ 
in terms of Pauli bases.
\hfill$\blacksquare$
\end{proof}

In contrast, $n$-spin-$\tfrac{1}{2}$ chains with Heisenberg-{\XX} interactions and {\em two} independently
controllable qubits, one at each end, provide a realisation of a series of {\em even-oder} orthogonal
algebras $\so(2n+2)$ for $n$ qubits, the first examples being shown in Fig.~\ref{fig:sys-so2k}.
\begin{figure}[Ht!]
\begin{center}
\includegraphics[height=7mm]{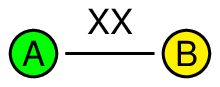}\hspace{9mm}
\includegraphics[height=7mm]{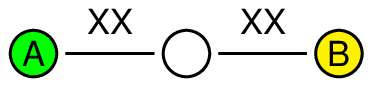}\hspace{9mm}
\includegraphics[height=7mm]{ex3}\\[2mm]
\includegraphics[height=7mm]{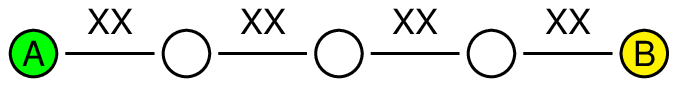}\hspace{2mm}
\end{center}
\caption{$n$-Spin-$\tfrac{1}{2}$ Heisenberg-{\XX} chains with even-order orthogonal system algebras 
\mbox{$\so(2n+2)$} result by allowing just two locally controllable qubits at the ends. A full series can
be constructed, the first examples of which are shown $\so(6)\cong\su(4)$, $\so(8)$, $\so(10)$, and $\so(12)$.
For $n=2$ one gets $\so(6)\cong\su(4)$ as a fully controllable two-qubit system.
\label{fig:sys-so2k}}
\end{figure}

\begin{proposition}\label{prop:so2n+2}
Heisenberg-{\XX} chains of $n$ spin-$\tfrac{1}{2}$ qubits ($n\geq 2$) and {\em two} individually 
locally controllable qubits, one at each end, give rise to the dynamic system algebras $\so(2n+2)$
as irreducible subalgebras embedded in $\su(2^n)$.
\end{proposition}
\begin{proof}
The constructive proof follows in entire analogy to the one of Proposition~\ref{prop:so2n+1}: 
however, the local controls at the second end 
imply that the Lie closure comprises each term occuring in the above list also read
from right to left thus duplicating the first line in each category from two terms to 
four terms. Since the second lines in each category already comprise the reverse
terms, one obtains the following complete list of elements:\\[3mm]
\begin{tabular}{l l l}
\multicolumn{2}{l}{The Pauli-basis elements for $\so(2n+2)$}\\
\hline\hline\\[-3mm]
$4$ terms & $\xy 1 \cdots 1$ & $1\cdots1\xy$\\[1mm]
$n$ terms & $z 1 \cdots 1$ & $1 z 1 \cdots 1$  etc \\
\hline\\[-3mm]
$4$ terms & $z \xy 1 \cdots 1$ & $1 \cdots 1 \xy z$\\[1mm]
$4(n-1)$ terms & $\xy \xy 1 \cdots 1$ & $ 1 \xy \xy 1 \cdots 1$  etc \\[1mm]
\hline
$\vdots$ & $\vdots$ & $\vdots$\\
$4$ terms & $ z  \cdots z \xy_k 1 \cdots 1$ & $ 1 \cdots 1 \xy_{n-k} z \cdots  z $\\[1mm]
$4(n-k+1)$ terms &$ \xy z \cdots z \xy_k 1 \cdots 1$ & $1 \xy_2 z \cdots z \xy_{k+1} 1 \cdots 1$  etc\\
$\vdots$ & $\vdots$ & $\vdots$\\
\hline\\[-3mm]
$4$ terms & $ z  \cdots z \xy_n$ & $\xy_1 z \cdots  z $ \\[1mm]
$4$ terms & $ \xy z \cdots z \xy_n$\\[1mm]
$1$ term  & $ z z \cdots z z$\\[1mm]
\hline\hline
\end{tabular}\\[8mm]
Finally, by the commutator $[(z \cdots z \xy_k 1\cdots 1),(1\cdots 1\xy_{n-k'}z\cdots z)]$
with $k=n-k'$, the longitudinal spin-order term $z_1z_2\cdots z_n$ listed last arises.
Counting terms, one arrives at a total of
$4n + n + 1 + 4\sum_{j=1}^{n-1} (n-j) = 5n + 1 + 4\sum_{j=1}^{n-1} j 
= 5n + 1 + 2n(n-1) = 2n^2 + 3n + 1 = \dim \so(2n+2)$ elements.
Thus also for all \mbox{$n$-spin-$\tfrac{1}{2}$} Heisenberg-$\XX$ chains individually 
controlled locally at the two ends we have provided a constructive scheme to determine  
irreducible representations of their respective dynamic Lie algebras $\so(2n+2)$ 
in terms of Pauli bases.
\hfill$\blacksquare$
\end{proof}

In both instances of
Heisenberg-$\XX$ chains controlled locally at one end [Fig.~\ref{fig:sys-so2k+1} with $\so(2n+1)$] 
or at two ends [Fig.~\ref{fig:sys-so2k} with $\so(2n+2)$] 
there are convenient Cartan decompositions $\fg=\fk\oplus\fp$:
the $\fk$-parts consist of per-{\em anti}symmetric matrices, while the $\fp$-parts comprise the
per-{\em symmetric} matrices, recalling that per-symmetry relates to reflection at the minor diagonal.
In both of the above listings, the respective subalgebras $\fk$ to $\so(2n+1)$ or $\so(2n+2)$
encompass the Hamiltonians with {\em odd} numbers of $z$-terms, while the respective 
subspaces $\fp$ contain the elements with {\em even} numbers of $z$-terms (including zero $z$-terms).

\medskip
For illustration, in the first example, i.e. the two-qubit Heisenberg-$\XX$ chain 
of Fig.~\ref{fig:sys-so2k+1},  
the transformation matrix $S$ satisfying $SH+H^tS=0$ according 
to {\bf Algorithm~3} is given by
\begin{equation}\label{eqn:so5-S}
S_1=\left(\begin{smallmatrix} 0 &0 &0 &+1\\ 0 &0 &+1 &0\\ 
0 &-1 &0 &0\\ -1 &0 &0 &0 \end{smallmatrix}\right) \equiv J_2\;.
\end{equation}
Here $S_1\bar{S}_1=-\unity$ reconfirms $\so(5)\cong\usp(4/2)$. 

\medskip
As a second example, for both of the three-qubit cases in Fig.~\ref{fig:sys-so2k+1} 
and Fig.~\ref{fig:sys-so2k}  
corresponding to $\so(7)$ and $\so(8)$, {\bf Algorithm~3} provides 
\begin{equation}
S_2=\left(\begin{smallmatrix} 
0 &0 &0 &0 &0 &0 &0 &+1  \\
0 &0 &0 &0 &0 &0 &-1 &0  \\
0 &0 &0 &0 &0 &+1 &0 &0  \\
0 &0 &0 &0 &-1 &0 &0 &0  \\
0 &0 &0 &-1 &0 &0 &0 &0  \\
0 &0 &+1 &0 &0 &0 &0 &0  \\
0 &-1 &0 &0 &0 &0 &0 &0  \\
+1 &0 &0 &0 &0 &0 &0 &0  \\
\end{smallmatrix}\right)\;,
\end{equation}
where $S_2\bar{S}_2=+\unity$ shows the orthogonal type of the respective irreducible representations.

\subsection{Dynamic Systems with Symplectic Algebras}\label{sec:symplectic}
Based on the smallest examples of qubit systems with Ising-{\ZZ}
interactions shown in Fig.~\ref{fig:sys-sp4-16}, even on the basis of {\em collective controls} one may construct
a full sequence of $n$ spin-$\tfrac{1}{2}$ chains with $n$ odd, the dynamic system algebras of which 
are the symplectic ones $\usp(2^{n-1})$. Note again that the bilinear control systems with
symplectic system algebras are {\em pure-state controllable} \cite{AA03}, whereas they fail
to be fully operator controllable.\\
\begin{figure}[Ht!]
\begin{center}
\begin{tabular}{c@{\hspace{20mm}}c}
{\sf (a)}  & {\sf (b) }\\[3mm]
\includegraphics[height=7mm]{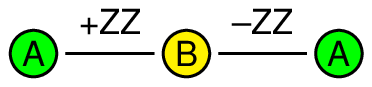} &
\includegraphics[height=7mm]{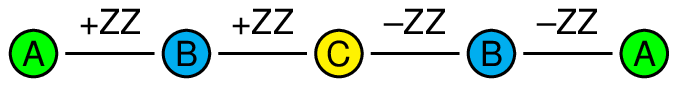}\\[4mm]
\includegraphics[height=7mm]{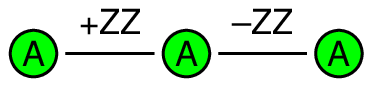} &
\includegraphics[height=7mm]{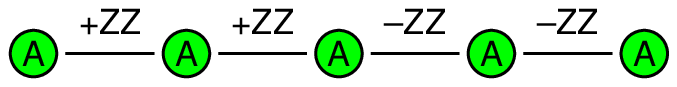}
\end{tabular}
\end{center}
\caption{Quantum systems with dynamic Lie algebras $\usp(4)$ [see (a)] and $\usp(16)$ [see (b)] as examples
of a series of Ising chains of $n=2k+1$ qubits with positive $\ZZ$ coupling terms on one
branch and negative couplings on the other. They give rise to the dynamic algebras $\usp(2^{n-1})$
irreducibly embedded in $\su(2^{n})$, respectively. The limiting case $k=0$ gives 
$\usp(1)\cong\su(2)$ as a single fully controllable qubit.
\label{fig:sys-sp4-16}}
\end{figure}
\begin{proposition}\label{prop:sp-n}
Ising-{\ZZ} chains of $n=2k+1$ spin-$\tfrac{1}{2}$ qubits ($k\geq 1$) including $k$ pairs of 
qubits which can be controlled simultaneously and one qubit in the middle of the chain which can 
be controlled independently as in the first row of Fig.~\ref{fig:sys-sp4-16}
give rise to the dynamic system algebras $\usp(2^{n-1})=\usp(2^{2k})$
as irreducible subalgebras embedded in $\su(2^n)=\su(2^{2k+1})$.
We obtain the same dynamic algebras when all qubits can only be controlled simultaneously
as in the second row of Fig.~\ref{fig:sys-sp4-16}.
\end{proposition}
\begin{proof}
We focus on the dynamic algebra $\fk_k$
corresponding to the case when all $2k+1$ qubits can only be controlled simultaneously
as in the second row of Fig.~\ref{fig:sys-sp4-16}. We denote by
$\bar{\fk}_k$ the dynamic algebra corresponding to the first row of Fig.~\ref{fig:sys-sp4-16}. We use the notation 
\begin{equation}
\mathrm{X}_j:= \underbrace{\mathrm{I}\cdots\mathrm{I}}_{j-1}\mathrm{X}\underbrace{\mathrm{I}\cdots \mathrm{I}}_{n-j},\quad \mathrm{Y}_j:= \underbrace{\mathrm{I}\cdots\mathrm{I}}_{j-1}\mathrm{Y}\underbrace{\mathrm{I}\cdots \mathrm{I}}_{n-j},\quad \text{ and }\quad \mathrm{Z}_j:= \underbrace{\mathrm{I}\cdots\mathrm{I}}_{j-1}\mathrm{Z}\underbrace{\mathrm{I}\cdots \mathrm{I}}_{n-j}\label{eqn:notation}
\end{equation}
to denote the operators which act, respectively, as $\mathrm{X}$, $\mathrm{Y}$, and $\mathrm{Z}$
on the $j$-th qubit and as the identity on all other qubits. We remark that the statements
of the Theorem can be directly verified for $k\in \{0,1\}$. We organize the proof in steps:
first we prove that $\fk_{k-1} \subseteq \fk_{k}$, second
we prove that 
$\bar{\fk}_k=\fk_k$, later we show
that $\fk_k$ is given in an irreducible (third step) and symplectic (fourth step) representation, and
in the end we prove that $\fk_{k}$ is not a proper subalgebra of
$\usp(2^{n-1})=\usp(2^{2k})$. 
Recall, that $\fk_k$ is generated by the operators
\begin{gather*}
f_1=-\tfrac{i}{2}\sum_{j=1}^{2k+1} \mathrm{X}_j,\;
f_2=-\tfrac{i}{2}\sum_{j=1}^{2k+1} \mathrm{Y}_j,\;
f_3=-\tfrac{i}{2} \left(
\sum_{j=1}^{k} \mathrm{Z}_j \mathrm{Z}_{j+1} 
  -  \sum_{j=k+1}^{2k} \mathrm{Z}_j \mathrm{Z}_{j+1}\right).
\end{gather*}
The corresponding algebra $\fk_{k-1}$ on $2k-1$ qubits can be embedded into $2k+1$ qubits 
using the operators
\begin{gather*}
g_1=-\tfrac{i}{2}\sum_{j=2}^{2k} \mathrm{X}_j,\;
g_2=-\tfrac{i}{2}\sum_{j=2}^{2k} \mathrm{Y}_j,\;
g_3=-\tfrac{i}{2}\left(\sum_{j=2}^{k} \mathrm{Z}_j \mathrm{Z}_{j+1} 
  -  \sum_{j=k+1}^{2k-1} \mathrm{Z}_j \mathrm{Z}_{j+1}\right).
\end{gather*}
We compute repeated commutators of $f_3$ with $f_1$. 
In the first two iterations, we get $f_4=[f_3,f_1]=-\tfrac{i}{2}(
\sum_{j=1}^{k} \left[\mathrm{Y}_j \mathrm{Z}_{j+1} + \mathrm{Z}_{j}\mathrm{Y}_{j+1}\right]
- \sum_{j=k+1}^{2k}\left[ \mathrm{Y}_j \mathrm{Z}_{j+1} + \mathrm{Z}_{j}\mathrm{Y}_{j+1} \right])
$ and 
$f_5=[f_3,f_4]=-\tfrac{i}{2}[
- \mathrm{X}_1 
 - 2\sum_{j=2}^{2k} \mathrm{X}_j 
- \mathrm{X}_{2k+1}
- 2( \sum_{j=2}^k \mathrm{Z}_{j-1}\mathrm{X}_j\mathrm{Z}_{j+1}
-  \mathrm{Z}_{k}\mathrm{X}_{k+1}\mathrm{Z}_{k+2}+ \sum_{j=k+2}^{2k} \mathrm{Z}_{j-1}\mathrm{X}_j\mathrm{Z}_{j+1})]
$.
Repeating this process, we obtain the element
$f_6=[f_3,f_5]=-\tfrac{i}{2}[
- \mathrm{Y}_1\mathrm{Z}_2 + \mathrm{Z}_{2k}\mathrm{Y}_{2k+1}
- 4( \sum_{j=2}^{k} \mathrm{Y}_j\mathrm{Z}_{j+1}
- \sum_{j=k+1}^{2k} \mathrm{Y}_j\mathrm{Z}_{j+1} 
+ \sum_{j=1}^{k} \mathrm{Z}_j\mathrm{Y}_{j+1}
- \sum_{j=k+1}^{2k-1} \mathrm{Z}_j\mathrm{Y}_{j+1})]
$. Finally, we compute the next element
$f_7=[f_3,f_6]=-\tfrac{i}{2}(
 \mathrm{X}_1 
+ 8  \sum_{j=2}^{2k} \mathrm{X}_j
+ \mathrm{X}_{2k+1}
+ 8 \sum_{j=2}^k \mathrm{Z}_{j-1}\mathrm{X}_j\mathrm{Z}_{j+1}
- 8 \mathrm{Z}_{k}\mathrm{X}_{k+1}\mathrm{Z}_{k+2}
+ 8 \sum_{j=k+2}^{2k} \mathrm{Z}_{j-1}\mathrm{X}_j\mathrm{Z}_{j+1})
$.
We obtain that $f_8=-(4f_5+f_7)/3=-\tfrac{i}{2}(\mathrm{X}_1+\mathrm{X}_{2k+1})$ and
$g_1=f_1-f_8$. The proof for $f_9=-\tfrac{i}{2}(\mathrm{Y}_1+\mathrm{Y}_{2k+1})$ and $g_2=f_2-f_9$ is similar.
We compute a few more commutators: First, we set
$f_{10}=[f_3,g_1]=-\tfrac{i}{2}( \sum_{j=2}^{k} \left[ \mathrm{Y}_j \mathrm{Z}_{j+1} + 
\mathrm{Z}_j\mathrm{Y}_{j+1} \right]-\sum_{j=k+1}^{2k-1} \left[ \mathrm{Y}_j \mathrm{Z}_{j+1} + 
\mathrm{Z}_j\mathrm{Y}_{j+1} \right]
+\mathrm{Z}_1\mathrm{Y}_2
-\mathrm{Y}_{2k} \mathrm{Z}_{2k+1})$.
The other commutators are
$f_{11}=[f_8,f_{10}]=-\tfrac{i}{2}(-\mathrm{Y}_1 \mathrm{Y}_2 + \mathrm{Y}_{2k} \mathrm{Y}_{2k+1})$,
$f_{12}=[f_1,f_{11}]=-\tfrac{i}{2}(-\mathrm{Z}_1\mathrm{Y}_2-\mathrm{Y}_1\mathrm{Z}_2
+\mathrm{Z}_{2k}\mathrm{Y}_{2k+1}+\mathrm{Y}_{2k}\mathrm{Z}_{2k+1})$, and
$f_{13}=[f_1,f_{12}]=-\tfrac{i}{2}(2 \mathrm{Y}_1\mathrm{Y}_2 - 2 \mathrm{Z}_1\mathrm{Z}_2 
-2\mathrm{Y}_{2k}\mathrm{Y}_{2k+1} + 2\mathrm{Z}_{2k}\mathrm{Z}_{2k+1})$.
It follows that
$f_{14}=-\tfrac{1}{2}f_{13}-f_{11}=-\tfrac{i}{2}(\mathrm{Z}_1\mathrm{Z}_2-\mathrm{Z}_{2k}\mathrm{Z}_{2k+1})$
and $g_3=f_3-f_{14}$. We obtain $\fk_{k-1} \subseteq \fk_{k}$ completing
the first step of the proof. Relying on the form of $f_8$ and $f_9$ we can prove by 
induction that $\bar{\fk}_k = \fk_k$ (second step). Assuming by induction that
$\fk_{k-1}$ is irreducibly embedded on $2k-1$ qubits, we obtain that the centralizer
of $\fk_{k-1}$ (embedded on $2k+1$ qubits) is given by all operators $O$ which operate only on the 
two outer qubits. But the generators $f_1$, $f_2$, and $f_3$ of $\fk_{k}$ do not simultaneously commute with operators $O$. Therefore, $\fk_{k}$ is irreducibly embedded on $2k+1$ qubits
(third step).
We switch to a new basis by reordering the qubits
according to the numbers in the figure:
\begin{center}
\includegraphics[height=7mm]{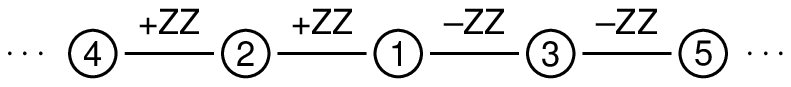}
\end{center}
In this basis, we can provide 
a matrix 
$$S :=
\begin{pmatrix}
0 & -1\\
1 & 0
\end{pmatrix}
\otimes M^{\otimes k}
=\begin{pmatrix}
0 & -1\\
1 & 0
\end{pmatrix}
\otimes
\begin{pmatrix}
0 & 0 & 0 & 1\\
0 & -1 & 0 & 0\\
0 & 0 & -1 & 0\\
1 & 0 & 0 & 0
\end{pmatrix}^{\otimes k}
$$ 
which satisfies $SH+H^{t}S=0$ for all elements $iH$ of $\fk_k$.
In particular, we have $S\bar{S}=-\unity_{2^n}$.
This can be readily verified on the generators for $k \in \{0,1\}$.  
Using our commutator computations we obtain that   $\fk_k=
\langle f_1,f_2,f_3 \rangle_{\rm Lie}$  is equal to $\langle g_1,g_2,g_3,f_8,f_9,f_{14} \rangle_{\rm Lie}$. 
Thus we can prove $SH+H^{t}S=0$ by induction on
$k$: Assuming the equation holds for $g_1$,$g_2$, and $g_3$
(i.e.\ for $k-1$), we need to prove that it also holds for $f_8$,$f_9$, and $f_{14}$
which are respectively given in the new basis by
$-\tfrac{i}{2}(\mathrm{X}_{2k}+\mathrm{X}_{2k+1})$, $-\tfrac{i}{2}(\mathrm{Y}_{2k}+\mathrm{Y}_{2k+1})$, 
and $-\tfrac{i}{2}(\mathrm{Z}_{2k-2}\mathrm{Z}_{2k}-\mathrm{Z}_{2k-1}\mathrm{Z}_{2k+1})$.
But this can be directly checked on the four outer qubits using 
$S_2=M\otimes M$. As $\fk_k$ is given in an irreducible representation,
the matrix $S$ is unique up to a scalar factor.
 This shows that $\fk_k$  is given in a symplectic representation 
and that $\fk_k \subseteq \usp(2^{2k})$ (fourth step).
Staying in our new basis, we prove  that $\fk_k$ 
contains the elements $P_j:=-\tfrac{i}{2}(\mathrm{X}_j \mathrm{Z}_{j+1}
-\mathrm{Z}_j\mathrm{X}_{j+1})$ and $Q_j:=-\tfrac{i}{2}(\mathrm{X}_j \mathrm{Y}_{j+1}
-\mathrm{Y}_j \mathrm{X}_{j+1})$ for all even $j\in \{2,\ldots,2k\}$
by induction on $j$.
This can be readily verified for $j=2$ considering $\fk_1\subseteq \fk_k$. Assuming that
$\fk_k$ contains the elements $P_{j-2}$ and $Q_{j-2}$ for $j\leq k$, we show that
it also contains the elements $P_{j}$ and $Q_{j}$. Recall that
$\fk_k$ contains the elements $f_8$, $f_9$, and $f_{14}$. In addition, the elements $v_1=-\tfrac{i}{2}(\mathrm{X}_{2k-2}+\mathrm{X}_{2k-1})$ and $v_2=-\tfrac{i}{2}(\mathrm{Y}_{2k-2}+\mathrm{Y}_{2k-1})$
are contained in $\fk_k$. But one can directly check on the four outer qubits
that $P_j$ and $Q_j$ are contained in the algebra 
$\fm=\langle 
f_8, f_9, f_{14}, P_{j-2}, Q_{j-2}, v_1, v_2
\rangle_{\rm Lie}
=\so(2^4)$.
Assuming that $\fk_j=\usp(2^{2j})$ holds for $j<k$, it follows
that $\usp(2^{2k-4}) \otimes \so(2^4) \subseteq \fk_k \subseteq \usp(2^{2k})$.
As $\usp(2^{2k-4}) \otimes \so(2^4)$ is a maximal subalgebra of 
$\usp(2^{2k})$ (see Thm.~1.4 of Ref.~\cite{Dynkin57}) and $f_3\in \fk_k$ is not of product form,
we obtain by induction that $\fk = \usp(2^{2k})$. 
\hfill$\blacksquare$
\end{proof}

Note the Cartan decomposition in the antisymmetric Ising chains of Fig.~\ref{fig:sys-sp4-16} 
can be taken with respect to the joint permutation of the qubits in the two branches with positive
and negative $\ZZ$ couplings: the $\fk$-part consists of all terms with {\em odd}
numbers of Pauli operators deviating from the identity, while the $\fp$-part collects
the ones with {\em even} numbers.

\begin{figure}[Ht!]
\begin{center}
\begin{tabular}{c@{\hspace{8mm}}c}
{\sf (a)}  & {\sf (b) }\\[3mm]
\raisebox{1.5mm}{\includegraphics[width=.45\linewidth]{antisym1}}
& 
\includegraphics[width=.45\linewidth]{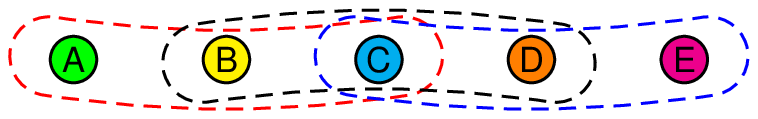}\\[2mm]
\raisebox{1.5mm}{\includegraphics[width=.45\linewidth]{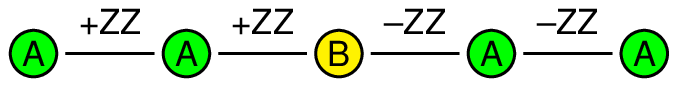}}
& 
\includegraphics[width=.45\linewidth]{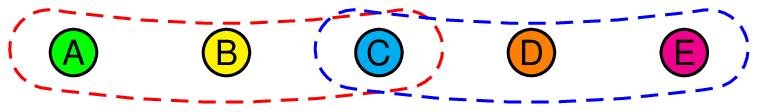}\\[2mm]
\raisebox{1.5mm}{\includegraphics[width=.45\linewidth]{antisym3}}
& 
\includegraphics[width=.45\linewidth]{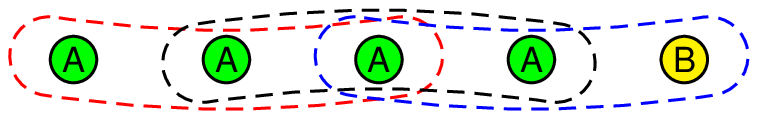}\\[2mm]
& 
\includegraphics[width=.45\linewidth]{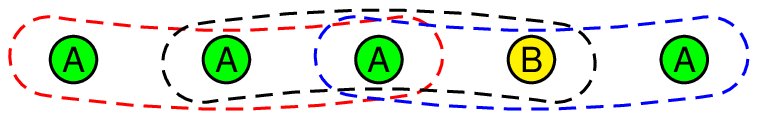}
\end{tabular}
\end{center}
\caption{Quantum systems with dynamic Lie algebra $\usp(16)$.
(a) Examples with pairwise Ising-{\ZZ} couplings and
(b) examples with three-body {\ZZZ}-interactions.
\label{fig:sys-sp16}}
\end{figure}

\medskip
As a third example, consider the first Ising chain in Fig.~\ref{fig:sys-sp4-16} corresponding to
$\usp(8/2)$. Here {\bf Algorithm~3} gives 
\begin{equation}
S_3=\left(\begin{smallmatrix} 
0 &0 &0 &0 &0 &0 &0 &-1  \\
0 &0 &0 &+1 &0 &0 &0 &0  \\
0 &0 &0 &0 &0 &+1 &0 &0  \\
0 &-1 &0 &0 &0 &0 &0 &0  \\
0 &0 &0 &0 &0 &0 &+1 &0  \\
0 &0 &-1 &0 &0 &0 &0 &0  \\
0 &0 &0 &0 &-1 &0 &0 &0  \\
+1 &0 &0 &0 &0 &0 &0 &0  \\
\end{smallmatrix}\right)\quad
\end{equation}
with $S_3\bar{S}_3=-\unity$ underscoring the irreducible representation is symplectic.

Moreover since all the dynamic systems in Figures~\ref{fig:sys-sp16}(a) and \ref{fig:sys-sp16}(b) share the same system
algebra $\usp(16)$, any two can mutually simulate eachother by Proposition~\ref{prop:sim-gen}. 
So remarkably enough, the spin chains in Fig.~\ref{fig:sys-sp16}(a) can simulate the 
effective three-qubit {\ZZZ}-interactions shown in Fig.~\ref{fig:sys-sp16}(b). In particular, note the 
lowest instance in Fig.~\ref{fig:sys-sp16}(a): even only the {\em collective local controls} on all the qubits
suffice to generate the three-body interactions with full local control shown at the
top of Fig.~\ref{fig:sys-sp16}(b). In turn, it may be astonishing at first sight that the system on
top of Fig.~\ref{fig:sys-sp16}(b) does not provide more dynamic degrees of freedom than the collective
system at the bottom of Fig.~\ref{fig:sys-sp16}(a), where the simulating power roots in the opposite
signs of the couplings.

\subsection{Dynamic Systems with Alternating Orthogonal and Symplectic Algebras}\label{sec:alternating}
Based on the smallest examples of Heisenberg-{\XX} chains with one single local control on the second qubit as 
shown in Fig.~\ref{fig:sys-sosp}, one may construct
a full sequence of $n$ spin-$\tfrac{1}{2}$ chains, whose dynamic system algebras 
are orthogonal or symplectic depending on the value of $n\not\in\{1,3\}$.
Again, observe symplectic system algebras ensure {\em pure-state controllability} \cite{AA03} 
without full operator controllability.

Quite remarkably, full local control on a single qubit suffices to get a dynamic algebra,
where the number of dynamic degrees of freedom scales exponentially with number of qubits,
a finding described only for full isotropic Heisenberg\nb$\XXX$ coupling up to now \cite{Burg08}. 
More precisely, one arrives at the following:

\begin{figure}[Ht!]
\begin{center}
\includegraphics[height=7mm]{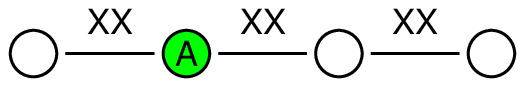}\hspace{9mm}
\includegraphics[height=7mm]{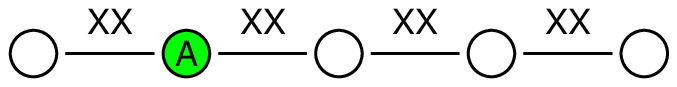}\\[2mm]
\includegraphics[height=7mm]{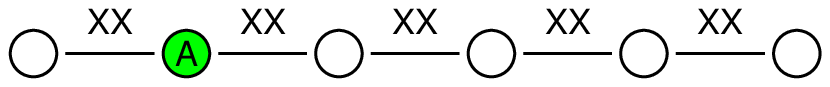}\hspace{9mm}
\includegraphics[height=7mm]{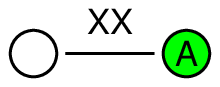}\hspace{2mm}
\end{center}
\caption{$n$-Spin-$\tfrac{1}{2}$ Heisenberg-{\XX} chains with $n\not\in\{1,3\}$ and 
only {\em one locally controllable qubit at the second position}
have orthogonal system algebras $\so(2^n)$  if $(n\bmod{4}) \in \{0,1\}$
and symplectic system algebras $\usp(2^{n-1})$ otherwise.
A full series can
be constructed for $n>3$, and
the examples shown for $n\in\{4,5,6,2\}$
correspond to
$\so(16)$, $\so(32)$, $\usp(32)$, $\so(5)\cong\usp(4/2)$.
In the single case of $n=3$, central symmetry arises, which makes the respective algebra {\em reducible}.
\label{fig:sys-sosp}}
\end{figure}

\begin{proposition}\label{prop:sosp}
Heisenberg-{\XX} chains of $n\not\in\{1,3\}$ spin-$\tfrac{1}{2}$ qubits
with only one locally controllable qubit at the second position
give rise to the dynamic algebras
$$
\fk_n=
\begin{cases}
\so(2^n) & \text{if $(n\bmod{4}) \in \{0,1\}$},\\
\usp(2^{n-1}) & \text{if $(n\bmod{4}) \in \{2,3\}$}
\end{cases}
$$
which are irreducibly embedded in $\su(2^n)$.
\end{proposition}
\begin{proof}
In the notation of Eqn.~\eqref{eqn:notation} the generators of the dynamic algebra
$\fk_n$ can be written as $f_1=-\tfrac{i}{2}\mathrm{X}_2$, $f_2=-\tfrac{i}{2}\mathrm{Y}_2$,
and $f_3=-\tfrac{i}{2}(\sum_{j=1}^{n-1} \mathrm{X}_j\mathrm{X}_{j+1} + \mathrm{Y}_j\mathrm{Y}_{j+1})$.
We remark that the statements
of the Theorem can be directly verified for $n\in \{2,4,5\}$. 
Assuming $n\geq 6$ from now on, 
we complete the proof by induction. We organize the proof in steps:
first we prove that $\fk_{n} \supsetneq \fk_{2}\hoplus\fk_{n-2}$, second
we show that $\fk_{n}$ is given in an irreducible representation, and
in the end we prove that $\fk_{n}$ is equal to $\so(2^n)$ or $\usp(2^{n-1})$.
By computing sums of commutators we identify certain elements of $\fk_{n}$.
The first elements are $f_4:=[f_1,[f_3,f_1]]+[f_2,[f_3,f_2]]=
-\tfrac{i}{2}(
\mathrm{X}_1\mathrm{X}_2+\mathrm{Y}_1\mathrm{Y}_2+
\mathrm{X}_2\mathrm{X}_3+\mathrm{Y}_2\mathrm{Y}_3)$,
$f_5:=[[f_2,f_1],[f_1,[f_2,f_4]]]=
-\tfrac{i}{2}(
\mathrm{X}_1\mathrm{X}_2+\mathrm{X}_2\mathrm{X}_3)$, and
$f_6:=[[f_1,f_2],[f_2,[f_1,f_4]]]=
-\tfrac{i}{2}(
\mathrm{Y}_1\mathrm{Y}_2+\mathrm{Y}_2\mathrm{Y}_3)$.
Next we compute the 
elements $f_7:=[f_4,[f_3,f_4]]+[f_2,[f_6,[f_2,[f_3,f_4]]]]+[f_1,[f_5,[f_1,[f_3,f_4]]]]
=-\tfrac{i}{2}(
\mathrm{X}_3\mathrm{X}_4+\mathrm{Y}_3\mathrm{Y}_4)$ and
$f_8:=[[f_2,f_1],[[f_1,f_2],[f_3,[f_4,f_3]]]]-f_7+[f_6,[f_7,f_3]]+[f_5,[f_7,f_3]]+[f_2,[f_6,[f_2,[f_3,f_4]]]]+[f_1,[f_5,[f_1,[f_3,f_4]]]]=
-\tfrac{i}{2}(\mathrm{X}_2\mathrm{X}_3+\mathrm{Y}_2\mathrm{Y}_3)$
leading to the elements $f_9:=f_4-f_8=
-\tfrac{i}{2}(\mathrm{X}_1\mathrm{X}_2+\mathrm{Y}_1\mathrm{Y}_2)$,
$f_{10}:=f_3-f_4=
-\tfrac{i}{2}(\sum_{j=3}^{n-1} \mathrm{X}_j\mathrm{X}_{j+1} + \mathrm{Y}_j\mathrm{Y}_{j+1})$, 
and $f_{11}:=f_4+f_7=
-\tfrac{i}{2}(
\mathrm{X}_1\mathrm{X}_2+\mathrm{Y}_1\mathrm{Y}_2+
\mathrm{X}_2\mathrm{X}_3+\mathrm{Y}_2\mathrm{Y}_3+
\mathrm{X}_3\mathrm{X}_4+\mathrm{Y}_3\mathrm{Y}_4)
$. By explicit computations on the first four qubits one can show
that the elements $f_{12}=-\tfrac{i}{2}\mathrm{X}_4$ and $f_{13}=-\tfrac{i}{2}\mathrm{Y}_4$ are contained
in $\fk_4=\langle 
f_1,f_2,f_{11}
\rangle_{\rm Lie}\subseteq \fk_n$. We obtain that $\fk_2=\langle 
f_1,f_2,f_{9}
\rangle_{\rm Lie}\subseteq \fk_n$ and $\fk_{n-2}=\langle 
f_{12},f_{13},f_{10}
\rangle_{\rm Lie}\subseteq \fk_n$. Therefore, $\fk_n=\langle
f_1,f_2,f_9, f_{12}, f_{13}, f_{10}, f_8
\rangle_{\rm Lie} \supsetneq \fk_{2}\hoplus\fk_{n-2}$. This completes the first step of the proof.
By induction $\fk_{2}$ and $\fk_{n-2}$ are given in an irreducible representation. Therefore, this holds also for $\fk_{2}\hoplus\fk_{n-2}$ and $\fk_n$,
which completes the second step of the proof. Using the matrices
$$
S_2:=
\left(\begin{smallmatrix}
0 & 0 & 0 & +1\\
0 & 0 & -1 & 0\\
0 & +1 & 0 & 0\\
-1 & 0 & 0 & 0
\end{smallmatrix}\right)\; \text{ and }\;
S_3:=
\left(\begin{smallmatrix}
0 & 0 & 0 & 0 & 0 & 0 & 0 & +1\\
0 & 0 & 0 & 0 & 0 & 0 & +1 & 0\\
0 & 0 & 0 & 0 & 0 &-1 & 0 & 0\\
0 & 0 & 0 & 0 &-1 & 0 & 0 & 0\\
0 & 0 & 0 & +1 & 0 & 0 & 0 & 0\\
0 & 0 & +1 & 0 & 0 & 0 & 0 & 0\\
0 &-1 & 0 & 0 & 0 & 0 & 0 & 0\\
-1 & 0 & 0 & 0 & 0 & 0 & 0 & 0
\end{smallmatrix}\right)
$$
we define the matrices 
$S_{2k}=(S_2)^{\otimes k}$ and
$S_{2k+1}=(S_2)^{\otimes (k-1)}\otimes S_3$.
We obtain that $S_{2k}\bar{S}_{2k}=(-1)^k \unity_{2^{2k}}$ and
$S_{2k+1}\bar{S}_{2k+1}=(-1)^k \unity_{2^{2k+1}}$. Relying on direct computations in the case of 
$j\in\{2,4,5\}$, one can verify
that $S_j H+H^{t}S_j=0$ holds for all elements $iH$ of $\fk_j$.
Assuming by induction that $S_j H+H^{t}S_j=0$ holds for all elements $iH$ of $\fk_j$ where
$j\in \{2,n-2\}$, we show that $S_n H+H^{t}S_n=0$ holds also for all elements
$iH$ of the algebra $\fk_n=\langle
f_1,f_2,f_9, f_{12}, f_{13}, f_{10}, f_8
\rangle_{\rm Lie} \supsetneq \fk_{2}\hoplus\fk_{n-2}$ by directly verifying
$S_4 f_8+f_8^{t} S_4=0$ on the first four qubits. In summary,
we proved that $\fk_{2}\hoplus\fk_{n-2} \subsetneq \fk_n \subseteq$
$\so(2^n)$ or $\usp(2^{n-1})$ depending on the value of $n$.
But Thm.~1.4 of Ref.~\cite{Dynkin57} says that $\fk_{2}\hoplus\fk_{n-2}$ is a maximal subalgebra
of $\so(2^n)$ or $\usp(2^{n-1})$. Thus, $\fk_n$ is equal to $\so(2^n)$ 
if $(n\bmod{4}) \in \{0,1\}$ and equal to $\usp(2^{n-1})$ otherwise.
This completes the last step of the proof.
\hfill$\blacksquare$
\end{proof}

\subsection{Dynamic Systems with Unitary Algebras}\label{sec:unitary}
We close the series of worked examples by considering  $n$\nb{}spin\nb$\tfrac{1}{2}$ Heisenberg-{\XX} chains with
$n\geq 2$, where the first two qubits can be independently, locally controlled
(see Fig.~\ref{fig:sys-su}). This case was recently studied in Refs.~\cite{KPR09,BP09,Diss-Sander}. We show that these systems are fully controllable for arbitrary $n\geq 2$.
\begin{figure}[Ht!]
\begin{center}
\includegraphics[height=7mm]{ex8}\hspace{9mm}
\includegraphics[height=7mm]{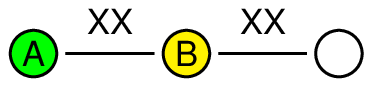}\hspace{9mm}
\includegraphics[height=7mm]{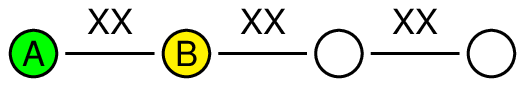}\\[2mm]
\includegraphics[height=7mm]{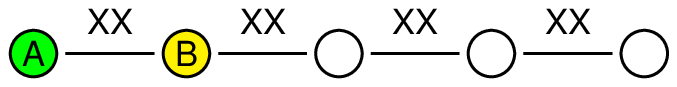}\hspace{2mm}
\end{center}
\caption{$n$-Spin\nb$\tfrac{1}{2}$ Heisenberg-{\XX} chains with $n\geq 2$,
where the first two qubits can be independently, locally controlled
have fully-controllable system algebras $\su(2^n)$. 
A full series can be constructed: the first examples shown 
correspond to the algebras $\so(6)\cong\su(4)$, $\su(8)$, $\su(16)$, and $\su(32)$.
\label{fig:sys-su}}
\end{figure}
\begin{corollary} 
Assume that the first two qubits of
a Heisenberg-{\XX} chain of $n$ spin-$\tfrac{1}{2}$ qubits with $n\geq 2$
can be independently, locally controlled. 
Then, the dynamic algebras is $\fk_n=\su(2^n)$.
\end{corollary}
\begin{proof}
The Theorem can be directly verified
for $n\in\{2,3,4,5\}$. 
Building on the proof of Proposition~\ref{prop:sosp}, we prove the Theorem
for $n\geq 6$ by induction. We first show that $\fk_n\supsetneq \fk_{2}\hoplus\fk_{n-2}$.
From the proof of Proposition~\ref{prop:sosp} it is only left to show
that the elements $\mathrm{X}_3$ and $\mathrm{Y}_3$ can be generated.
But this can be directly verified by computations on the first four qubits. Thus we proved
that $\su(2^n)\supseteq \fk_n \supsetneq \su(2^2) \hoplus \su(2^{n-2})$.
Thm.~1.3 of Ref.~\cite{Dynkin57} says that  $\su(2^2) \hoplus \su(2^{n-2})$
is a maximal subalgebra of $\su(2^n)$. The Theorem follows immediately.
\hfill$\blacksquare$
\end{proof}

\section{Fermionic Quantum Systems}\label{sec:fermionic}
Fermionic $d$-level systems with any kind of quadratic (pair-interaction) 
Hamiltonians give rise to dynamic system Lie algebras limited to subalgebras
like $\so(2d)$ or $\so(2d+1)$. By making use of the Jordan-Wigner transformation, 
which links the number of levels $d$ with the number of qubits $n$,
we show how these systems can be simulated
by $n$\nb spin\nb $\tfrac{1}{2}$ chains with partial local control. --- 
For keeping the relation to mathematical literature, 
references are more extensive in this section. 

\subsection{Quadratic Hamiltonians}
To fix notations, consider the fermionic creation and annihilation operators
$f^{\dagger}_p$ and $f_p$ which operate on a finite-dimensional quantum system of $d$ levels
and satisfy the anticommutation relations (with $1\leq p,q \leq d$ and
$\{a,b\}_+:=ab+ba$ and the Kronecker function $\delta_{p,q}$)
\begin{equation}\label{anticomm}
\{f^{\dagger}_p, f_q\}_+ = \delta_{p,q} 
\quad\text{and}\quad
\{f^{\dagger}_p, f^{\dagger}_q\}_+ = 0 = \{f_p, f_q\}_+\quad.
\end{equation}
For the $p$-th level of the system, $f^{\dagger}_p$ 
and $f_p$ change the occupation numbers $n_p$ labelling the respective states $\ket{n_p}$
such as to give
$
f^{\dagger}_p \ket{0} = \ket{1} = 
\left(\begin{smallmatrix}
1\\
0
\end{smallmatrix}
\right)
=
\ket{{\uparrow}}
$
and
$ 
f_p \ket{1} = \ket{0} = 
\left(\begin{smallmatrix}
0\\
1
\end{smallmatrix}
\right)
=
\ket{{\downarrow}},
$
where by the Pauli principle $n_p \in \{0,1\}$, $(f^{\dagger}_p)^2\equiv 0$, and 
$(f_p)^2\equiv 0$. 
Now the properties of the usual {\em quadratic Hamiltonians}
(see, e.g., \cite{LSM61,Berezin,BR86,kraus-pra71,ADGH:2010})
\begin{equation}\label{quadHam}
H := \sum_{p,q=1}^d A_{pq} f_p f_q + B_{pq} f_p f^{\dagger}_q 
+C_{pq} f^{\dagger}_p f_q + D_{pq} f^{\dagger}_p f^{\dagger}_q.
\end{equation}
can be discussed in terms of their pair-interaction coupling coefficients $A_{pq}$, $B_{pq}$, $C_{pq}$,
and $D_{pq}$ seen as (possibly complex) entries of the $d{\times}d$-matrices
$A$, $B$, $C$, and $D$. Hermiticity of $H$ requires 
$A=D^{\dagger}$, $B=B^{\dagger}$, and  $C=C^{\dagger}$, while
in addition, the commutator relations of Eqn.~\eqref{anticomm} imply
\begin{align*}
H  = & \sum_{p,q=1}^d -A_{pq} f_q f_p - D_{pq} f^{\dagger}_q f^{\dagger}_p \quad
 + \sum_{p,q=1, p\neq q}^d -B_{pq} f^{\dagger}_q f_p -C_{pq} f_q f^{\dagger}_p \\
 & + \sum_{p=1}^d B_{pp} (1- f^{\dagger}_p f_p) + C_{pp} (1-f_p f^{\dagger}_p)\;,
\end{align*}
which upon identification with Eqn.~\eqref{quadHam} enforces
$A=-A^{t}$, $B=-C^{t}$, and $D=-D^{t}$.
Finally keeping a widely-used custom (see, e.g., p.~452 of Ref.~\cite{LSM61} or p.~173 of Ref.~\cite{Berezin}) 
we also assume the entries of $A$, $B$, $C$, and $D$ are real.
Summing up, $A$ is real skew-symmetric following
$ A=\bar A=-A^t=-D$ and 
$B$ is real symmetric with $B=\bar B=B^t=-C $.
So $H$ of Eqn.~\eqref{quadHam} can be given in `symmetrised' form (see, e.g., p.~2 of Ref.~\cite{OO09}) as
\begin{equation}\label{symquadHam}
H= \sum_{p,q=1}^d (-B_{pq}) \left[f^{\dagger}_p f_q - f_p f^{\dagger}_q\right]
+ (-A_{pq}) \left[f^{\dagger}_p f^{\dagger}_q - f_p f_q\right].
\end{equation}

\subsection{Jordan-Wigner Transformation} 
For simplicity, first recall the map from the non-compact, (real) special linear algebra
$\sll(2,\R)$ to the compact, special unitary algebra $\su(2)$. The generators of
$\sll(2,\R)$ are given by 
$
\mathrm{E}=
\left(\begin{smallmatrix}
0 & 1\\
0 & 0
\end{smallmatrix} \right)
=
\tfrac{1}{2}(\mathrm{X}+i\mathrm{Y}),\;
\mathrm{H}=
\left(\begin{smallmatrix}
1 & 0\\
0 & -1
\end{smallmatrix} \right)
= \mathrm{Z},\; \text{ and }\;
\mathrm{F}=
\left(\begin{smallmatrix}
0 & 0\\
1 & 0
\end{smallmatrix} \right)
= \tfrac{1}{2}(\mathrm{X}-i\mathrm{Y})
$
and the generators of $\su(2)$ can be chosen as $i\mathrm{X}$, $i\mathrm{Z}$, and $i\mathrm{Y}$ where
$
\mathrm{X}=
\left(\begin{smallmatrix}
0 & 1\\
1 & 0
\end{smallmatrix} \right)
=
\mathrm{F}+\mathrm{E},\;
\mathrm{Z}=
\left(\begin{smallmatrix}
1 & 0\\
0 & -1
\end{smallmatrix} \right)
= \mathrm{H},\; \text{ and }\;
\mathrm{Y}=
\left(\begin{smallmatrix}
0 & -i\\
i & 0
\end{smallmatrix} \right)
= i(\mathrm{F}-\mathrm{E}).
$
Obviously, this also defines a map between the Lie algebras (see, e.g.,
p.~127 of Ref.~\cite{Sepanski07} or Chap.~IX, Sec.~3.6 of Ref.~\cite{Bourb08b})
readily serving as a prototype for maps between {\em non-compact} normal real forms
and {\em compact} real forms of Lie algebras (cp.~\cite{Helgason78}). 

\medskip
\noindent
Likewise, the Jordan-Wigner transformation
\cite{JW28} maps the 
fermionic operators (for $1\leq p \leq d$)
\begin{align}
f^{\dagger}_p=\tfrac{1}{2}(c_p+ic_{p+d})\quad \text{and}\quad
&f_p=\tfrac{1}{2}(c_p-ic_{p+d})\quad\text{to the operators}\\
c_p:=f_p + f^{\dagger}_p\quad \text{and}\quad
&c_{p+d}:=i(f_p-f^{\dagger}_p)\;.
\end{align}
Now $c_p$ and $c_{p+d}$ can be given explicitly as $d$-qubit operators
\begin{equation}
c_p=\underbrace{\mathrm{Z}\cdots\mathrm{Z}}_{p-1}\mathrm{X}\underbrace{\mathrm{I}\cdots \mathrm{I}}_{d-p}
\quad \text{ and }\quad
c_{p+d}=\underbrace{\mathrm{Z}\cdots\mathrm{Z}}_{p-1}\mathrm{Y}\underbrace{\mathrm{I}\cdots \mathrm{I}}_{d-p}\;.
\end{equation}
We refer to Chap.~VIII, Sec.~3 of Ref.~\cite{Boerner67},
Sec.~9.6 of Ref.~\cite{Miller72}, or Sec.~44 of Ref.~\cite{SW86},
where more information on this construction can be found and where connections
to Clifford algebras are discussed. In the context of Clifford algebras this
construction is sometimes named after Brauer and Weyl \cite{BW35} (see, e.g., p.~183
of Ref.~\cite{Hladik99}). 

\subsection{Quadratic Hamiltonians in Qubit Form}
Now one can readily apply the Jordan-Wigner transformation to
fermionic quadratic Hamiltonians. Assuming that the number of levels is $d$, the
Hamiltonian of Eqn.~\eqref{symquadHam} is mapped to (see, e.g., Thm.~VI.I of Ref.~\cite{OO09})
\begin{subequations}
\label{qubitform}
\begin{align}
\label{qubitform_A}
H=&-\sum_{p=1}^d  B_{pp} \underbrace{\mathrm{I}\cdots\mathrm{I}}_{p-1}\mathrm{Z}\underbrace{\mathrm{I}\cdots \mathrm{I}}_{d-p}\\
\label{qubitform_B}
&+\sum_{p,q=1, p>q}^d  
B_{pq} \Big(
\underbrace{\mathrm{I}\cdots\mathrm{I}}_{q-1}
\mathrm{X} \underbrace{\mathrm{Z}\cdots\mathrm{Z}}_{p-q-1}
\mathrm{X} \underbrace{\mathrm{I}\cdots \mathrm{I}}_{d-p}
+
\underbrace{\mathrm{I}\cdots\mathrm{I}}_{q-1}
\mathrm{Y} \underbrace{\mathrm{Z}\cdots\mathrm{Z}}_{p-q-1}
\mathrm{Y} \underbrace{\mathrm{I}\cdots \mathrm{I}}_{d-p}
\Big) \\
\label{qubitform_C}
& - \sum_{p,q=1, p>q}^d A_{pq}  \Big(
\underbrace{\mathrm{I}\cdots\mathrm{I}}_{q-1}
\mathrm{X} \underbrace{\mathrm{Z}\cdots\mathrm{Z}}_{p-q-1}
\mathrm{X} \underbrace{\mathrm{I}\cdots \mathrm{I}}_{d-p}
-
\underbrace{\mathrm{I}\cdots\mathrm{I}}_{q-1}
\mathrm{Y} \underbrace{\mathrm{Z}\cdots\mathrm{Z}}_{p-q-1}
\mathrm{Y} \underbrace{\mathrm{I}\cdots \mathrm{I}}_{d-p}
\Big).
\end{align}
\end{subequations}
This determines the dynamic algebra of a general fermionic Hamiltonian containing quadratic terms:

\begin{theorem}\label{thm:gen_quad}
Let the entries of the real antisymmetric matrix $A$ and the real symmetric matrix $B$
denote the control functions of the Hamiltonian  
given in Eqn.~\eqref{qubitform}. We assume $d\geq 2$. The dynamic algebra $\so(2d)$ of the corresponding control system is
embedded in $\su(2^d)$. The centraliser of the dynamic algebra is one-dimensional and is given 
by the $d$-qubit operator $-\tfrac{i}{2}\mathrm{Z}\cdots\mathrm{Z}$. The embedding of $\so(2d)$ into
$\su(2^d)$ splits into two irreducible representations of equal dimension.
\end{theorem}
\begin{proof}
Let $\fk_d$ denote the dynamic algebra of the control system.
The generators  $-\tfrac{i}{2}\mathrm{I}\cdots \mathrm{I}\mathrm{X}\mathrm{Z}\cdots\mathrm{Z}\mathrm{X}\mathrm{I}\cdots \mathrm{I}$ and  $-\tfrac{i}{2}\mathrm{I}\cdots \mathrm{I}\mathrm{Y}\mathrm{Z}\cdots\mathrm{Z}\mathrm{Y}\mathrm{I}\cdots \mathrm{I}$ arise from linear combinations of Eqns.~\eqref{qubitform_B}-\eqref{qubitform_C}, and
computing commutators with the generators $-\tfrac{i}{2}\mathrm{I}\cdots\mathrm{I}\mathrm{Z}\mathrm{I}\cdots \mathrm{I}$ from Eqn.~\eqref{qubitform_A} reveals the
generators
$-\tfrac{i}{2}\mathrm{I}\cdots \mathrm{I}\mathrm{X}\mathrm{Z}\cdots\mathrm{Z}\mathrm{Y}\mathrm{I}\cdots \mathrm{I}$ and  $-\tfrac{i}{2}\mathrm{I}\cdots \mathrm{I}\mathrm{Y}\mathrm{Z}\cdots\mathrm{Z}\mathrm{X}\mathrm{I}\cdots \mathrm{I}$.
By comparing with the (independent) proof of Theorem~\ref{thm:gen_quad_lin} it follows that
$\fk_d$ is a subalgebra of $\so(2d+1)$. The statements of the Theorem can be directly verified
for $d\in\{2,3,4,5\}$. We assume by induction that the $(d{-}1)$-qubit operator 
$a=-\tfrac{i}{2}\mathrm{Z}\cdots\mathrm{Z}$  is the only element in the centraliser of $\fk_{d-1}$.
Considering $a$ as an $d$-qubit operator operating on the first $d-1$ qubits we obtain
that the centraliser of $\fk_d$ can only contain linear combinations of elements from the set $-\tfrac{i}{2}\{\mathrm{I}\cdots\mathrm{I}\mathrm{X},\mathrm{I}\cdots\mathrm{I}\mathrm{Y},\mathrm{I}\cdots\mathrm{I}\mathrm{Z},\mathrm{Z}\cdots\mathrm{Z}\mathrm{I},
\mathrm{Z}\cdots\mathrm{Z}\mathrm{X},\mathrm{Z}\cdots\mathrm{Z}\mathrm{Y},\mathrm{Z}\cdots\mathrm{Z}\mathrm{Z}\}$. The second statement of the Theorem follows from the
fact that only the last element in the set commutes with all generators.
We obtain from the structure of the generators that $\so(2d-4)\oplus\so(4)\subsetneq \fk_d$.
We remark that $\so(2d-4)\oplus\so(4)$ is a maximal subalgebra of $\so(2d)$ and that
$\so(2d)$ is a maximal subalgebra of $\so(2d+1)$ (see, e.g., p.~219 of Ref.~\cite{BS49} or Table~12 on p.~150 of Ref.~\cite{Dynkin57b}). Therefore, $\fk_d$ is equal to $\so(2d)$ or $\so(2d+1)$. But
$\fk_d\neq\so(2d+1)$ as the corresponding embedding would be irreducible, and the first statement of the Theorem follows.
We already showed that the centraliser is one-dimensional which is equivalent to the fact that
the commutant is two-dimensional. Theorem~1.5 of Ref.~\cite{Ledermann} says that the dimension of the commutant of a representation
$\phi$ is given by $\sum_i m_i^2$ where the $m_i$ are the multiplicities of the irreducible components of $\phi$. Thus, the embedding of $\so(2d)$ into
$\su(2^d)$ splits into two irreducible representations. The third statement of the Theorem
follows now by proving that the simultaneous eigenvalues of the Hamiltonians corresponding to all
generators are given by $\pm 1$ occurring each with multiplicity $2^{d-1}$. This can be directly verified for $d\in\{2,3,4,5\}$. Assuming the statement by induction for all $\fk_{d'}$ with $d' < d$ we obtain the simultaneous eigenvalues of Hamiltonians corresponding to the algebras $\fk_{d-2}$ and $\fk_2$
acting on the first $d-2$ qubits and the last two qubits, respectively. As the eigenvalues of a tensor product of two matrices are given by the product of the eigenvalues of each matrix, we can prove the statement by induction.
\hfill$\blacksquare$
\end{proof}

The first statement of Theorem~\ref{thm:gen_quad} is related to the fact that the
canonical transformations of fermionic systems are given by orthogonal transformations
$\oo(2d)$ (see, e.g., p.~118 of Ref.~\cite{Berezin} and p.~39 of Ref.~\cite{BR86}).
In more mathematical literature, the first statement of Theorem~\ref{thm:gen_quad} can be found
in Sec.~9.6 of Ref.~\cite{Miller72},
pp.~182-186 of Ref.~\cite{SW86},
pp.~499-501 of Ref.~\cite{Gilmore94}, and
pp.~180-186 of Ref.~\cite{Hladik99}. Recently, this came again into focus
\cite{kraus-pra71,ADGH:2010}. The dynamic algebra $\so(2d)$ was also discussed as symmetry
of Hamiltonians in Refs.~\cite{KM04,KM05}.

The polynomial growth (in $d$) of the algebra in Theorem~\ref{thm:gen_quad} 
to the dynamic system of Eqns.\eqref{qubitform}
was identified as the reason why 
efficient classical simulation of quadratic, fermionic
quantum systems is possible (cp., e.g., pp.~9-10 of Ref.~\cite{Knill01},
pp.~5-6 of Ref.~\cite{TD02}, and
pp.~3 of Ref.~\cite{SBOK06}). --- Setting $n:=d-1$ in Proposition~\ref{prop:so2n+2}
we obtain
\begin{corollary}
Heisenberg-{\XX} chains of $d-1$ spin-$\tfrac{1}{2}$ qubits ($d\geq 3$) and {\em two} individually 
locally controllable qubits, one at each end, have the dynamic algebra 
$\so(2d)$ and can thus simulate a general fermionic quadratic Hamiltonian on $d$ levels and vice versa.
\end{corollary}
By the second and third statement of Theorem~\ref{thm:gen_quad},
the embedding of $\so(2d)$ into $\su(2^d)$ splits into two irreducible representations
of equal dimension, and thus $\so(2d)$ acts simultaneously on both components.
Therefore, this embedding is---besides a doubling of the dimension---equivalent 
to the embedding of  $\so(2d)$ into $\su[2^{(d-1)}]$ via
Proposition~\ref{prop:so2n+2} as readily verifiable by resorting to the Pauli basis for
$\so(2d)$ given there. Referring to Tab.~\ref{tab:two}, we
further remark that the embedding of $\so(2d)$ into $\su[2^{(d-1)}]$ 
can arise from a symplectic representation (for $d=4k+2$), an orthogonal one (for $d=4k$), or
a unitary one (for $d$ odd).

\medskip
Now we illustrate that a controlled spin chain can simulate a quadratic
fermionic system, while the converse does not hold.
To this end, consider the case where the general quadratic (\/`physical\/')
Hamiltonian is supplemented by the (\/`unphysical\/') linear terms
$\sum_{p=1}^d j_p f^{\dagger}_p + k_p f_p$ with $j_p,k_p \in \C$. 
Hermiticity implies $j_p=\bar{k}_p$. Assume again that the coefficients $j_p$ and
$k_p$ are real, i.e.\ $j_p=\bar{j}_p$ and $k_p=\bar{k}_p$ to obtain
$j_p=\bar{j}_p=k_p$. Thus the linear terms can be written as
$\sum_{p=1}^d j_p (f^{\dagger}_p + f_p)$; they are mapped via the Jordan-Wigner 
transformation to the operators
\begin{equation}\label{linH}
H_2 = \sum_{p=1}^d j_p \underbrace{\mathrm{Z}\cdots\mathrm{Z}}_{p-1}
\mathrm{X} \underbrace{\mathrm{I}\cdots \mathrm{I}}_{d-p}\;.
\end{equation}
As will be shown next, this determines the dynamic algebra of a fictitious Hamiltonian system
containing quadratic {\em and} linear terms (see also the Pauli basis for $\so(2d+1)$ given in the proof
to Proposition~\ref{prop:so2n+1}):
\begin{theorem}\label{thm:gen_quad_lin}
Let $j_p\in\R$ ($1\leq p \leq d$) and
the entries of the real antisymmetric matrix $A$ and the real symmetric matrix $B$
denote the control functions in the Hamiltonian $H+H_2$, where
$H$ and $H_2$ are given by
Eqn.~\eqref{qubitform} and Eqn.~\eqref{linH}, respectively. 
The dynamic algebra $\so(2d+1)$ of the corresponding control system 
is irreducibly embedded in $\su(2^d)$. 
\end{theorem}
\begin{proof}
Computing commutators of generators 
$-\tfrac{i}{2}\mathrm{Z}\cdots\mathrm{Z}\mathrm{X}\mathrm{I}\cdots \mathrm{I}$
from Eqn.~\eqref{linH} with generators 
$-\tfrac{i}{2}\mathrm{I}\cdots\mathrm{I}\mathrm{Z}\mathrm{I}\cdots \mathrm{I}$
from Eqn.~\eqref{qubitform_A}, we obtain the additional generators 
$-\tfrac{i}{2}\mathrm{Z}\cdots\mathrm{Z}\mathrm{Y}\mathrm{I}\cdots \mathrm{I}$.
Furthermore, the generators  $-\tfrac{i}{2}\mathrm{I}\cdots \mathrm{I}\mathrm{X}\mathrm{Z}\cdots\mathrm{Z}\mathrm{X}\mathrm{I}\cdots \mathrm{I}$ and  $-\tfrac{i}{2}\mathrm{I}\cdots \mathrm{I}\mathrm{Y}\mathrm{Z}\cdots\mathrm{Z}\mathrm{Y}\mathrm{I}\cdots \mathrm{I}$ arise from linear combinations of Eqns.~\eqref{qubitform_B}-\eqref{qubitform_C} and
computing commutators with generators $-\tfrac{i}{2}\mathrm{I}\cdots\mathrm{I}\mathrm{Z}\mathrm{I}\cdots \mathrm{I}$
from Eqn.~\eqref{qubitform_A} reveals the generators
$-\tfrac{i}{2}\mathrm{I}\cdots \mathrm{I}\mathrm{X}\mathrm{Z}\cdots\mathrm{Z}\mathrm{Y}\mathrm{I}\cdots \mathrm{I}$ and  $-\tfrac{i}{2}\mathrm{I}\cdots \mathrm{I}\mathrm{Y}\mathrm{Z}\cdots\mathrm{Z}\mathrm{X}\mathrm{I}\cdots \mathrm{I}$.
Now the Theorem follows by comparing all the generators with the table in the proof of Proposition~\ref{prop:so2n+1}.
\hfill$\blacksquare$
\end{proof}

\noindent
Moreover, setting $n:=d$ in Proposition~\ref{prop:so2n+1} we get
\begin{corollary}
Heisenberg-{\XX} chains of $d$ spin-$\tfrac{1}{2}$ qubits ($d\geq 1$) and a {\em single} locally controllable
qubit at one end have the dynamic algebra $\so(2d+1)$ and
can simulate a general fermionic quadratic Hamiltonian on $d$ levels with its dynamic algebra $\so(2d)$, 
but obviously not vice versa\footnote{as is also illustrated by the unphysical linear terms above}.
\end{corollary}

\subsection{Discussion}
We analyse three further cases of fermionic Hamiltonians.
First, consider quadratic Hamiltonians (without linear terms) which are particle-number preserving, i.e.
$A=0$ in Eqn.~\eqref{qubitform}. Assuming the elements of $B$ are control functions, we obtain
$\uu(d)$ as dynamic algebra (cp.\ p.~501 of Ref.~\cite{Gilmore94}). 
Second, the diagonal normal form
(see, e.g., App.~A of Ref.~\cite{LSM61}, Sec.~III.8 of Ref.~\cite{Berezin}, Sec.~3.3 of Ref.~\cite{BR86}, 
and Theorem~II.1 of Ref.~\cite{OO09})
 for the Hamiltonian $H$ of Eqn.~\eqref{symquadHam} 
\begin{equation}
\sum_{p=1}^d E_p \left(f^{\dagger}_p f_p - \tfrac{1}{2}\right) =
\sum_{p=1}^d \frac{E_p}{2} \left(f^{\dagger}_p f_p -f_p f^{\dagger}_p\right)
\end{equation}
(with $E_p$ positive and real) is mapped by
the Jordan-Wigner transformation to the $d$-qubit operator
\begin{equation}
\sum_{p=1}^d \frac{E_p}{2} \underbrace{\mathrm{I}\cdots\mathrm{I}}_{p-1}
\mathrm{Z} \underbrace{\mathrm{I}\cdots \mathrm{I}}_{d-p}.
\end{equation}
Considering $E_p$ ($1\leq p \leq d$) as controls, we get a $d$-dimensional abelian Lie algebra
as dynamic algebra. Third, if we allow for fermionic operators of arbitrary order 
(less than or equal to $d$), we get $\su(2^d)$ as dynamic algebra.

In summary, the sequence of dynamic Lie algebras
\begin{equation}
\su(2^d) \supseteq \so(2d+1) \,\supset \,\so(2d) \,\supseteq \, \uu(d)
\end{equation}
plays a prominent role for fermionic quantum systems as pointed out in
Chapter~22 of Ref.~\cite{Wybourne}
and
Sec.~IV of Ref.~\cite{WL02}. This sequence of Lie algebras is also widely studied 
in particle physics
\cite{WZ82,Georgi99}.

\subsection{Spinless Hubbard Model with Periodic Boundary Conditions\label{spinless}}
First, specify a spinless version of the Hubbard model (see p.~20 in Ref.~\cite{EFGKK05} 
or p.~61 in Ref.~\cite{AS10}) with periodic boundary conditions where
$d+1=1$ due to the periodicity, $d\geq 2$, and 
$t,u\in \R$:
\begin{subequations}\label{Hubbard}
\begin{align}
H = & -t \sum_{p=1}^d (f^{\dagger}_p f_{p+1} - f_p f^{\dagger}_{p+1})
\label{Hubbard_a} \\
& +u \sum_{p=1}^d (f^{\dagger}_p f_p- \tfrac{1}{2})\;.\label{Hubbard_b}
\end{align}
\end{subequations}
Equation~\eqref{Hubbard_a} resembles a spinless tight-binding model
(see p.~437 in Ref.~\cite{EFGKK05} or p.~59 in Ref.~\cite{AS10})
and equals the quadratic Hamiltonian
of Eqn.~\eqref{symquadHam} with $A=0$ and
$$
B=-\frac{1}{2}
\begin{pmatrix}
0 & 1 &   &   &   &   &   & 1\\
1 & 0 & 1 &   &   &   &   & \\
  & 1 & 0 & \cdot &   &   & \\
  &   & \cdot  & \cdot & \cdot &   &   &\\
  &   &   & \cdot   &  \cdot  & \cdot   &\\
  &   &   &   & \cdot  & 0 & 1 &  \\
  &   &   &   &   & 1 & 0 & 1\\
1 &   &   &   &   &   & 1 & 0
\end{pmatrix}.
$$
The Jordan-Wigner transform of the Hamiltonian $H$
of Eqn.~\eqref{Hubbard} takes the form
\begin{subequations}
\label{unitary_both}
\begin{align}
\label{unitary_A}
H =& \frac{t}{2}
\Big[ \Big( \sum_{p=1}^{d-1} \underbrace{\mathrm{I}\cdots\mathrm{I}}_{p-1}
\mathrm{Y} \mathrm{Y} \underbrace{\mathrm{I}\cdots\mathrm{I}}_{d-1-p} +
\underbrace{\mathrm{I}\cdots\mathrm{I}}_{p-1}
\mathrm{X} \mathrm{X} \underbrace{\mathrm{I}\cdots\mathrm{I}}_{d-1-p}
\Big)+ 
\Big( 
 \mathrm{X} \underbrace{\mathrm{Z}\cdots\mathrm{Z}}_{d-2} \mathrm{X}+
 \mathrm{Y} \underbrace{\mathrm{Z}\cdots\mathrm{Z}}_{d-2} \mathrm{Y}
\Big)
\Big] \\
\label{unitary_B}
&+  u \sum_{p=1}^d
\underbrace{\mathrm{I}\cdots\mathrm{I}}_{p-1}
\mathrm{Z} \underbrace{\mathrm{I}\cdots\mathrm{I}}_{d-p}\;.
\end{align}
\end{subequations}
Extending Eqn.~\eqref{Hubbard_b} to $
\sum_{p=1}^d u_p (f^{\dagger}_p f_p- \tfrac{1}{2})$ such that it contains site-dependent
control functions $u_p\in \R$, we obtain (by building on App.~B of Ref.~\cite{ADGH:2010})
\begin{lemma}
The dynamic control system corresponding to the Hamiltonian 
\begin{equation}
H=-t \sum_{p=1}^d (f^{\dagger}_p f_{p+1} - f_p f^{\dagger}_{p+1})
+ \sum_{p=1}^d u_p (f^{\dagger}_p f_p-\tfrac{1}{2})
\end{equation}
has $\uu(d)$ as dynamic Lie 
algebra assuming that $u_p\in \R$ are controls and $d\geq 2$.
\end{lemma}
\begin{proof}
Let $\fk_d$ denote the dynamic algebra of the control system.
We obtain from Eqn.~\eqref{unitary_A} one generator $a_1$ and from Eqn.~\eqref{unitary_B}
the generators ($0\leq p \leq d$)
$$\mathrm{z}_p:=-\tfrac{i}{2}\mathrm{Z}_p=-\tfrac{i}{2}\underbrace{\mathrm{I}\cdots\mathrm{I}}_{p-1}\mathrm{Z}\underbrace{\mathrm{I}\cdots \mathrm{I}}_{d-p}.$$
One can verify on the generators that the $d$-qubit operators $-\tfrac{i}{2}\mathrm{Z}\cdots\mathrm{Z}$
and $-\tfrac{i}{2}\sum_{p=1}^d \mathrm{Z}_{p}$ are both elements of the centraliser of $\fk_d$.
By comparison with Theorem~\ref{thm:gen_quad} we obtain that $\fk_d\subseteq \so(2d)$.
As the centraliser in Theorem~\ref{thm:gen_quad} is one-dimensional and the centraliser of $\fk_d$ is at least two-dimensional, it follows that
$\fk_d \subsetneq \so(2d)$. We remark that 
$\uu(d)$ is a maximal subalgebra of $\so(2d)$ and that $\su(q)\oplus \uu(d-q)$ is a maximal subalgebra
of $\su(d)$
(see, e.g., p.~219 of Ref.~\cite{BS49}). In particular,  
$\uu(q)\oplus \uu(d-q)$ is a maximal subalgebra
of $\uu(d)$. The Theorem can be directly verified for $d\in\{2,3,4,5\}$. 
We assume by induction that the Theorem is true for all $\fk_{d'}$ with $d' < d$.
The Theorem follows by induction if we show that $\fk_d \supsetneq \uu(q)\oplus \uu(d-q)$
holds for any $q$. We compute the commutators
$a_2:=[\mathrm{z}_1,[\mathrm{z}_2,a_1]]=-\tfrac{i}{2}(\mathrm{X}_1\mathrm{X}_2+\mathrm{Y}_1\mathrm{Y}_2)$
and $a_3:=[\mathrm{z}_3,[\mathrm{z}_2,a_1]]=-\tfrac{i}{2}(\mathrm{X}_2\mathrm{X}_3+\mathrm{Y}_2\mathrm{Y}_3)$
[using again the notation of Eqn.~\eqref{eqn:notation}]. For $3\leq j \leq d-1$ we have
$g_j:=[\mathrm{z}_{j+1},[\mathrm{z}_j,a_1]]=-\tfrac{i}{2}(\mathrm{X}_j\mathrm{X}_{j+1}+\mathrm{Y}_j\mathrm{Y}_{j+1})$
and by linear combinations we obtain the $d$-qubit operator 
$a_4:=-\tfrac{i}{2}(\mathrm{X}\mathrm{Z}\cdots\mathrm{Z}\mathrm{X}+\mathrm{Y}\mathrm{Z}\cdots\mathrm{Z}\mathrm{Y})$. We further compute the commutators
$a_5:=[\mathrm{z}_2,[a_2,a_4]]=-\tfrac{i}{2}(\mathrm{I}\mathrm{X}\mathrm{Z}\mathrm{Z}\cdots\mathrm{Z}\mathrm{X}+\mathrm{I}\mathrm{Y}\mathrm{Z}\mathrm{Z}\cdots\mathrm{Z}\mathrm{Y})$ and
$a_6:=[\mathrm{z}_3,[a_3,a_5]]=-\tfrac{i}{2}(\mathrm{I}\mathrm{I}\mathrm{X}\mathrm{Z}\cdots\mathrm{Z}\mathrm{X}+\mathrm{I}\mathrm{I}\mathrm{Y}\mathrm{Z}\cdots\mathrm{Z}\mathrm{Y})$. Using the elements
$g_j$ and $a_6$ we can build the element
$a_6-\tfrac{i}{2}(\sum_{j=3}^{d-1} \mathrm{X}_j \mathrm{X}_{j+1} +\mathrm{Y}_j \mathrm{Y}_{j+1})$
which together with $a_2$ and the elements $\mathrm{z}_j$ generates $\uu(2)\oplus \uu(d-2)$. As $a_3$ is not contained in $\uu(2)\oplus \uu(d-2)$ we proved that
$\fk_d \supsetneq \uu(2)\oplus \uu(d-2)$.
\hfill$\blacksquare$
\end{proof}

\subsection{Note on the Hubbard Model with Periodic Boundary Conditions and Spin}
Including the spin $\sigma =\pm$ in the Hubbard model gives
\begin{subequations}
\label{Hubbard_spin}
\begin{align}
H = & -t \left[ \sum_{\sigma=\pm} \sum_{p=1}^d (f^{\dagger}_{p,\sigma} f_{p+1,\sigma} - f_{p,\sigma} f^{\dagger}_{p+1,\sigma})\right]
\label{Hubbard_spinA}
\\
& + \sum_{p=1}^d u_p (f^{\dagger}_{p,+} f_{p,+} - \tfrac{1}{2})(f^{\dagger}_{p,-} f_{p,-} - \tfrac{1}{2})\;,
\label{Hubbard_spinB}
\end{align}
\end{subequations}
where the anticommutation relations of Eq.~\eqref{anticomm} still hold among operators
with {\em equal spin values}, while operators with {\em different spin values} anticommute.
The spin degrees of freedom just split each of the original levels $p$ into two sub-levels.
Thus the image of the Hamiltonian form Eq.~\eqref{Hubbard_spin} under the Jordan-Wigner transformation 
operates on a space of squared dimension as compared to the case without spin
and the dynamic algebra is embedded in $\su(2^{2\,d})$.
The drift Hamiltonian of Eq.~\eqref{Hubbard_spinA} is mapped to $A_0\otimes \unity + \unity \otimes A_0$
with
\begin{equation*}
A_0:= \Big( \sum_{p=1}^{d-1} \underbrace{\mathrm{I}\cdots\mathrm{I}}_{p-1}
\mathrm{Y} \mathrm{Y} \underbrace{\mathrm{I}\cdots\mathrm{I}}_{d-1-p} +
\underbrace{\mathrm{I}\cdots\mathrm{I}}_{p-1}
\mathrm{X} \mathrm{X} \underbrace{\mathrm{I}\cdots\mathrm{I}}_{d-1-p}
\Big)+ 
\Big( 
 \mathrm{X} \underbrace{\mathrm{Z}\cdots\mathrm{Z}}_{d-2} \mathrm{X}+
 \mathrm{Y} \underbrace{\mathrm{Z}\cdots\mathrm{Z}}_{d-2} \mathrm{Y}
\Big)\;.
\end{equation*}
The control Hamltonians of Eq.~\eqref{Hubbard_spinB} are mapped to $A_p\otimes A_p$ where 
\begin{equation*}
A_p:=
\underbrace{\mathrm{I}\cdots\mathrm{I}}_{p-1}
\mathrm{Z} \underbrace{\mathrm{I}\cdots\mathrm{I}}_{d-p}\;.
\end{equation*}
For $d=2$, direct computation using 
the computer algebra system {\sf MAGMA}~\cite{MAGMA} 
gives the system Lie algebra $\su(2)\oplus\su(2)\oplus\uu(1)$
embedded in $\su(2^4)$. The general case appears more intricate and 
goes beyond the scope of this work.

\section{Bosonic Quantum Systems}\label{sec:bosonic}
Finally we comment on the bosonic case.
As opposed to the Pauli principle in the fermionic case, in bosons
the occupation number $n_p$ is not bounded and---even for 
a finite number $d$ of levels---the dynamic algebra of Hamiltonians of arbitrary order need not be finite
unless the particle number is also bounded. 
Yet, the dynamic algebra for {\em quadratic (pair-interaction) Hamiltonians} is given by the real symplectic
algebra $\usp(2d,\R)$ (see, e.g., p.~36 of Ref.~\cite{BR86}, p.~186 of Ref.~\cite{SW86}, or
p.~501 of Ref.~\cite{Gilmore94}). 
We have not yet found an appropriate spin system that would be dynamically
equivalent to the compact real form $\usp(d)$ of a quadratic bosonic system with
algebra $\usp(2d,\R)$.
However, in Secs.~\ref{sec:symplectic}--\ref{sec:alternating}, we have already presented 
spin systems with dynamic algebras $\usp(2^{n-1})$ 
which are actually {\em more powerful than required} and contain
the compact real form $\usp(d)$ of 
a quadratic bosonic system with algebra $\usp(2d,\R)$.
For further analysis of the bosonic case, the Holstein-Primakoff transformation
may be of help (see, e.g., p.~78 of Ref.~\cite{AS10}).---Finally, 
the results of mutually simulating quantum systems are summarised in Tab.~\ref{tab:simu}.

\section[Outlook: Quantum Simulation as an Observed Optimal-Control Problem]
{Outlook: \nopagebreak \\ \nopagebreak Quantum Simulation as an Observed Optimal-Control Problem}\label{sec:outlook}
Clearly, in view of experimental settings, one may take a more specific point of view 
by comparing the {\em time course} of two
{\em observed bilinear control systems} $(\Sigma_\mu)$, $\mu=a,b$
with respect to
(i) a set of Hermitian (and mutually orthogonal) 
observables $C_\nu^{(a)}$ and $C_{\nu'}^{(b)}$ with $\nu,\nu'\in\mathcal I\subseteq\{1,2,\dots,N^2-1\}$,
(ii) the initial states $\rho_0^\mu$, 
(iii) a given time interval $[0,T]$, and
(iv) admissible controls $u^\mu_j(t)\in\mathcal U^\mu\subseteq \R{}$
\begin{eqnarray}\label{eqn:bilinear_contr4}
        \dot \rho^\mu(t) &=& -i\Big[\big(H_0^\mu + \sum_{j=1}^m u^\mu_j(t) H^\mu_j\big) \;,\; \rho^\mu(t)\Big]
			\quad\text{with}\quad \rho^\mu(0)=\rho_0^\mu\\
             \expt{C}_\nu^\mu(t) &=& \tr\{(C_\nu^\mu)^\dagger \rho^\mu(t)\}
			\quad\text{with}\quad \{C_\nu^\mu\},
			\;\nu\in\mathcal I\;.
\end{eqnarray}
Now the comparison resorts to the expectation values $\expt{C}^\mu(t)$ 
via states $\rho^\mu(t)$, drifts $H_0^\mu$, controls $H^\mu_j$,
and control amplitudes $u^\mu_j(t)$. Note that $\{C_\nu^{(a)}\}$ and
$\{C_{\nu'}^{(b)}\}$ need not coincide, but if $\Sigma_a$ shall simulate $\Sigma_b$
it is convenient to require
$|\{C_\nu^{(a)}\}| \geq  |\{C_{\nu'}^{(b)}\}|$
so that (by invoking the above orthogonality of the observables with respect to the Hilbert-Schmidt
scalar product) one can ensure:
$\rank\spanR\{C_\nu^{(a)}\} \geq \rank\spanR\{C_{\nu'}^{(b)}\}$.

\medskip
Now for simultaneous measurement, it is useful to pick several observables $C_\nu^\mu$
as long as they are compatible (mutually commute), or, more generally, as long 
as they are mutually non-disturbing in the sense of the recent findings in 
Ref.~\cite{HeinWolf10}. Simultaneous expectation values are conveniently collected
in the observation vectors
\begin{equation}
[{\sf \expt{C}}^\mu(t)]:=[\expt{C}^\mu_1(t), \expt{C}^\mu_2(t), ...]^t\quad\text{with}\quad \mu=a,b\;.
\end{equation}

\noindent
Likewise, we define the respective dynamic system algebras of $\Sigma_a$ and $\Sigma_b$ as
\begin{equation}
\fk_\mu := \expt{iH_0^\mu, iH_j^\mu | j=1,2,\dots,m_\mu}_{\rm Lie} \quad\text{with}\quad \mu=a,b\;.
\end{equation}

\noindent
Clearly, $\fk_a\supseteq \fk_b$ implies $\Sigma_a$ can simulate $\Sigma_b$.
However, if $\Sigma_a$ comes with a larger set of observables $\{C_\nu^{(a)}\}$, 
the above condition is still sufficient, but it is no longer necessary. 
This is analogous to the fact that in {\em quantum systems} controllability implies observability,
whereas the converse need not hold \cite{dAless08} (for details see \cite{SS09}).
In {\em classical systems}, however, controllability and observability are dual to
one another (see, e.g., \cite{Sontag}), since no observables accounting for the quantum-specific 
measurement process are involved. --- Now the notion of {\em weak simulation}, for which
simulability can be seen as a strong condition, comes naturally:

\begin{proposition}
A dynamic system $\Sigma_a$ can weakly simulate another dynamic
system $\Sigma_b$ in time interval $[0,T]$ and with respect to the
two sets of observables $\{C_\nu^a\}$ and $\{C_\nu^b\}$,
if there exists a pair of initial conditions $\rho_0^a$ and $\rho_0^b$
(reachable form the respective equilibrium states)
and two sets of admissible control vectors $u^a_j(t)$ and $u^b_j(t')$
such that $M [\expt{{\sf C}}^a(t)]=[\expt{{\sf C}}^b(t')]$ for all $t\in[0,T]$ and
$t'\in[\tau(0),\tau(T)]$, where $\tau(t)$ is a bijective function of $t$
for all $t\in[0,T]$ and $M$ is a  
map 
$M: \R^n \to \R^m, [\expt{{\sf C}}^a(t)] \mapsto [\expt{{\sf C}}^b(t)]$
with $n\geq m$.
\end{proposition}

As will be described elsewhere, 
the previous proposition motivates to view simulability as a generic precondition to
formulate weak quantum simulation as an optimal-control task: 
minimise $||M [\expt{{\sf C}}^a(t)]-[\expt{{\sf C}}^b(t')]||^2_2$ subject to the
differential equations of motion given in Eqn.~\eqref{eqn:bilinear_contr4}.

\begin{table}[Ht!]
\caption{\label{tab:simu} Summarising Overview on Simulating Quantum Systems}
\begin{center}
\fontencoding{OT1}
\fontfamily{cmr}
\fontseries{m}
\fontshape{n}
\fontsize{8}{9.5}
\selectfont
\begin{tabular}{@{\hspace{2mm}}c@{\hspace{0mm}}}
\hline\hline\\
\begin{minipage}{11.7cm}
\vspace{-1mm}
\renewcommand{\footnoterule}{}
\begin{tabular}{lcccr}
system type & levels\footnote{\hspace{-1.5mm}\tiny{In second quantisation, the number of levels
for the fermionic or bosonic system usually arises
as a map\\ \phantom{X} from the number of qubits 
in the spin system. For fermions, the mapping is given by the 
Jordan-Wigner\\ \phantom{X} transformation.}}
 & fermionic & bosonic & system alg.\ \\[1mm]
 $n$\nb{}spins\nb$\tfrac{1}{2}$  & & \multicolumn{2}{c}{----------- order
 of coupling -----------}\\ \midrule
\raisebox{-2mm}{\includegraphics[width=30mm]{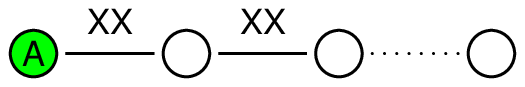}} & $n$ & quadratic (i.e.\ $2$)  & -- 
& $\so(2n+1)$\\[2mm]
\raisebox{-2mm}{\includegraphics[width=30mm]{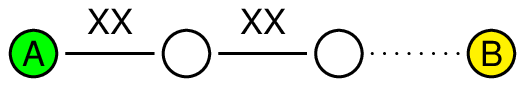}} & $n+1$ & quadratic (i.e.\ $2$)  & --
& $\so(2n+2)$ \\[2mm]
\raisebox{-2mm}{\includegraphics[width=30mm]{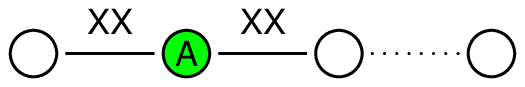}} \\[2mm]
	\; for $(n\bmod{4})\in\{0,1\}$ & $n$ & up to $n$ & -- & $\so(2^n)$ \\[1mm]
	\; for $(n\bmod{4})\in\{2,3\}$ & $n$ & -- & up to $n$ & $\usp(2^{n-1})$\\[2mm]
\raisebox{-2mm}{\includegraphics[width=30mm]{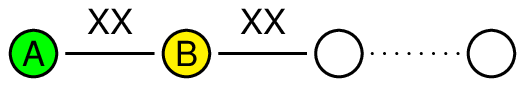}} & $n$ & up to $n$ & up to $n$ &
$\su(2^n)$
\end{tabular}
\vspace{-2mm}
\end{minipage}
\vspace{2.5mm}\\
\hline\hline
\end{tabular}
\end{center}
\end{table}

\section{Conclusion}

Often the presence or absence of symmetries in quantum hardware architectures
can already be assessed by inspection. Given the system Hamiltonian
as well as the control Hamiltonians, 
(i) we have provided a single necessary and sufficient symmetry condition ensuring
full controllability, and
(ii) in view of practical applications we have shown easy means 
(solving systems of homogeneous linear equations)
to determine the symmetry
of the dynamic system algebra $\fk$ merely in terms of its commutant or centraliser $\fk'$.
If the system Hamiltonian corresponds to a connected coupling graph, 
the absence of any symmetry can be further exploited to decide 
full controllability: it means the dynamic system algebra is irreducible and {\em simple}.
Now conjugation to simple orthogonal or symplectic candidate subalgebras can again
be decided solely on the basis of solving systems of homogeneous linear equations. 
The final identification task can now be settled because here we have given a {\em complete} list of 
irreducible simple subalgebras of $\su(N)$ compatible with the physical constituents 
as a dynamic pseudo-spin system. This avoids the usual and significantly more costly 
way of explicitly calculating Lie closures.  
We have thus made precise and easily accessible 
the following four conditions ensuring full controllability of a dynamic
qubit system in terms of its system algebra $\fk\subseteq \su(N)$:\\[1mm]
(1) the system must not show a symmetry ($\fk$ must have a trivial centraliser
		$\fk'$),\\
(2) the coupling graph of the control system must be connected, \\
(3) the system algebra $\fk$ must not be given in a symplectic or an orthogonal\\[-.3mm]
	\phantom{JX} representation, and finally\\
(4) if $\fk$ is given in a unitary representation, it must not be on the list of \\
	\phantom{JX} proper irreducible unitary simple subalgebras of $\su(N)$, in particular,
	$\fk\neq\fe_6$.

\medskip
\noindent
The system algebra completely determines the possible dynamics of controlled Hamiltonian
systems. Therefore, the lattice of irreducible simple subalgebras to $\su(N)$ given 
here also provides an easy means to assess not only the somewhat easier cases of {\em mutual simulability} 
but also the more intricate cases of {\em simulability with least overhead}
of dynamic systems of spin or fermionic or bosonic nature. In a number of 
examples (see also Tab.~\ref{tab:simu}), we have illustrated 
how controlled quadratic fermion and boson systems can be simulated by spin chains and in certain
cases also vice versa.

Finally, since full controllability entails observability (while in the quantum
domain the converse does not necessarily hold),
symmetry constraints immediately pertain to observability
as discussed in more detail in Ref.~\cite{SS09}.

\begin{acknowledgement}
This work was supported in part by the {\sc eu} programmes {\sc qap},
{\sc q-essence}, and the exchange with {\sc coquit}, moreover by the Bavarian excellence network {\sc enb}
via the International Doctorate Programme of Excellence
{\em Quantum Computing, Control, and Communication} ({\sc qccc})
as well as by the {\em Deutsche Forschungsgemeinschaft} ({\sc dfg}) in the
collaborative research centre {\sc sfb}~631. --- We are grateful to
Zolt{\'a}n Zimbor{\'a}s for fruitful discussion and to
Uwe Sander for helpful exchange on algorithmic aspects related to
Ref.~\cite{SS09}. R.Z.\ would like to thank the {\em Institut f{\"u}r 
Kryptographie und Sicherheit} ({\sc iks}) at the {\em Karlsruher Institut 
f{\"u}r Technologie} (Germany) for kindly permitting to use their 
computer resources for the computations. 
\end{acknowledgement}

\bigskip
\appendix
\section*{APPENDIX}
\section{Tensor-Product Structure in Qudit Systems with
	Many-Body Interactions}\label{sec:tensor2}

For quantum simulation, we generalise the discussion such as to
embrace {\em qudit} systems with (effective) many-body interactions.
Treating them as control systems embedded in $\su(N)$, 
now $\su(d_1) \hoplus \su(d_2) \hoplus \cdots \hoplus \su(d_n)$ 
is a \emph{tensor-product structure} of $\su(N)$, 
where $\prod_{j=1}^{n} d_j = N$ and $d_j \geq 2$. 
We consider the subalgebras $\su(d_j)$ as subsystems of the tensor-product structure.
We say that the tensor-product structure 
$\fh_1 \hoplus \fh_2 \hoplus  \cdots \hoplus \fh_n$
is a refinement of the tensor-product structure 
$\su(d_1) \hoplus \su(d_2) \hoplus \cdots \hoplus \su(d_n)$ if
$\fh_j$ is either equal to  $\su(d_j)$ or equal to
$\su(c_{j,1}) \hoplus \su(c_{j,2}) \hoplus \cdots
\hoplus \su(c_{j,m_j})$, where 
$\prod_{k=1}^{m_j} c_{j,k} = d_j$ and $m_j \geq 2$. 
We call $\fh_1 \hoplus \fh_2 \hoplus  \cdots \hoplus 
\fh_n$ a proper refinement if there is one $j$ such that
$\fh_j \neq \su(d_j)$.
For a given quantum system in $\su(N)$, there exists a common 
refinement $\su(p_1) \hoplus \su(p_2) \hoplus \cdots \hoplus \su(p_n)$
of all tensor-product structures, where
$\prod_{j=1}^n p_j$ is the factorization of $N$ into prime numbers.
The common refinement is unique up to permutations of subsystems. 

Now with respect to  tensor-product structure 
$\su(d_1) \hoplus \su(d_2) \hoplus \cdots \hoplus \su(d_n)$, 
again we write Hamiltonians as a linear combination ($c_k \in \R$) 
\begin{equation}\label{Hsum_gen}
H = \sum_{k=1}^m c_k \mathcal{H}_k
\end{equation}
of elements $\mathcal{H}_k=-\tfrac{i}{2} (\mathcal{H}_{k,1} \otimes 
\mathcal{H}_{k,2} \otimes \cdots \otimes \mathcal{H}_{k,n})$ forming a
\emph{tensor basis} of $\su(N)$. 
The elements $\mathcal{H}_{k,\ell} \in
  \mathcal{B}_\ell \cup \{ \unity_{d_\ell} \} $
are choosen relative to bases 
$$\left\{ -i \mathrm{A}\, | \, \mathrm{A} \in 
\mathcal{B}_\ell:=\{\mathrm{B}_{\ell,1}, \mathrm{B}_{\ell,2}, \ldots 
,\mathrm{B}_{\ell,(d_\ell)^2-1}\} \right\}$$
of $\su(d_\ell)$. In addition, we assume that the order $$\ord(\mathcal{H}_k):=
\# \{ \ell\, \colon \, \mathcal{H}_{k,\ell} \neq \unity_{d_\ell} \} \geq 1.$$

Recall  that the Hamiltonian $H$ has a \emph{coupling graph},
if its order $\ord(H)=\ord(\sum_{k=1}^m c_k \mathcal{H}_k)
:=\max(\{\ord( \mathcal{H}_k  )\, | \, k = 1,\ldots, m \})$
is equal to $2$, which is the case in pairwise coupling interactions. 
The vertices $j$ are given by the subsystems
$\su(d_j)$ and we get an edge 
between the nodes $k_1$ and $k_2$ with $k_1\neq k_2$
if there exists a $\mathcal{H}_k$ in Eqn.~\eqref{Hsum_gen}
such that $\{k_1,k_2\}=\{j\, \colon \, \mathcal{H}_{k,j} \neq \unity_{d_j} \}$.
If all control Hamiltonians are local, i.e.\ are contained in 
$\su(d_1) \hoplus \su(d_2) \hoplus \cdots \hoplus \su(d_n)$,
then we say that the coupling graph of the drift Hamiltonian $H_d$ is 
the coupling graph of the control system.

We say a control system on $\su(N)$ is \emph{weakly connected},
if the dynamic algebra $\fk$ contains for each proper partition 
of its tensor-product structure in
($\mathcal{I}_1\cup\mathcal{I}_2 = \{1,2,\ldots,m\}$,
$\mathcal{I}_1\cap \mathcal{I}_2 = \{\}$)
$$\fh_1=\hoplus_{j\in\mathcal{I}_1} \su(d_j)\quad \text{ and }\quad \fh_2=\hoplus_{j\in\mathcal{I}_2} \su(d_j)$$  
an element of
$\su(N)\setminus[\fh_1 \hoplus \fh_2]$.
For Hamiltonians $H$ of $\ord(H)=2$, this
is equivalent to the fact that the coupling graph is connected.
We will also use the stronger notion of a \emph{directly connected} control system
for which the dynamic algebra $\fk$ contains an element of
$\su(d_{j_1} d_{j_2})\setminus
[\su(d_{j_1}) \hoplus \su(d_{j_2})]$ for each pair of 
subsystems $\su(d_{j_1})$ and $\su(d_{j_2})$ with $j_1 \neq j_2$. --- 
With these notions, Theorem~\ref{thm:bilinear} generalises as follows.

\medskip
\begin{theorem}\label{thm:connected}
Consider a bilinear control system on $\su(\prod_{j=1}^{n} d_j)$, where \mbox{$d_j\geq 2$}.
Assume that the subsystems $\su(d_j)$ with $j\in \{1,\ldots,n\}$
are independently fully controllable
so the dynamic algebra 
$\fk \supseteq \su(d_1) \hoplus \su(d_2) \hoplus \cdots \hoplus \su(d_n)$. 
The control system is fully controllable, i.e.\ $\fk=\su(\prod_{j=1}^{n} d_j)$,
if and only if the control system is directly connected.
In particular, $\fk=\su(\prod_{j=1}^{n} d_j)$ is simple.
\end{theorem}
\begin{proof}
The `only if'-direction is obvious. We prove the `if'-direction.
First, we assume that $n=2$. As the subsystems are independently fully controllable, 
we obtain $\fk \supseteq \su(d_1) \hoplus \su(d_2)$. The dynamic algebra $\fk$ contains an element of
$\su(d_1 d_2)\setminus[\su(d_1) \hoplus \su(d_2)]$, as
the control system is directly connected.
It follows from Thm.~1.3 of \cite{Dynkin57}  that
$\su(d_1) \hoplus \su(d_2)$ is a maximal subalgebra of $\su(d_1 d_2)$.
As $\fk \supsetneq \su(d_1) \hoplus \su(d_2)$, 
this proves $\fk=\su(d_1 d_2)$. The general case follows by induction
on the number of subsystems.
We remark that $\su(\prod_{j=1}^{n} d_j)$ is simple so the last assertion follows.
\hfill$\blacksquare$
\end{proof}

This complements results on the controllability of quantum circuits \cite{Barenco},
where the controllability of continuous and discrete sets of unitary transformations 
is considered.
In particular, Theorems~4.1 and 4.2 of Ref.~\cite{BryBry02} (see also \cite{BDD02})
rely also on the maximality of the subgroup of local operations
on two qudits [i.e. on $\SU(d^2)\supset \SU(d)\otimes\SU(d)$].
Our controllability proof can be compared
to proofs relying on Cartan decompositions (see Thm.~5 of Ref.~\cite{NMRJOGO06} and
Prop.~2.4 of Ref.~\cite{DiHeGAMM08}). Unfortunately, one cannot 
substitute `directly connected' with `weakly connected' in 
Theorem~\ref{thm:connected}:

\medskip
\begin{example}
Consider a bilinear control system on $\su(8)$ with 
the tensor-product structure $\su(2)\hoplus\su(2)\hoplus\su(2)$.
We assume that the subsystems are independently fully controllable,
i.e.\ $$\fk \supseteq \langle i\mathrm{XII}, i\mathrm{YII}, i\mathrm{ZII},
 i\mathrm{IXI}, i\mathrm{IYI}, i\mathrm{IZI},
 i\mathrm{IIX}, i\mathrm{IIY}, i\mathrm{IIZ} \rangle_{\rm Lie}.$$
In addition, we have a drift Hamiltonian $H_d=\mathrm{ZZZ}$.
The control system is weakly connected and $\fk$ acts irreducibly.
The dynamic algebra is $\fk=\usp(4)\neq \su(8)$ and hence the system is
not fully controllable.
\end{example}


\section{Connected Control Systems in Qudit Systems with
	Many-Body Interactions\label{sec:simple}}

In this Appendix we build on Sec.~\ref{sec:sym_cons} and discuss a more general notion of
\emph{connected} control systems in qudit systems with
many-body interactions which do not necessarily have a natural coupling graph.
We freely use the notation of Appendix~\ref{sec:tensor2}.

Recall Example~\ref{ex1} of Sec.~\ref{sec:sym_cons}.
Motivated by this example 
one might conjecture that the dynamic algebra $\fk$ is simple 
if the control system is weakly connected and $\fk$ acts irreducibly.
Unfortunately, this is not true.

\begin{example}\label{ex2}
Assume that we have a bilinear control system on $\su(8)$ with two subsystems corresponding
to the tensor-product structure $\su(4)\hoplus \su(2)$. On the first subsystem we pick 
$\fh_1=\langle i\mathrm{XII}, i\mathrm{YII}, i\mathrm{ZII},
 i\mathrm{IXI}, i\mathrm{IYI}, i\mathrm{IZI} \rangle_{\rm Lie}$ as
the local dynamic Lie algebra.
On the second subsystem we pick the local dynamic Lie algebra
$\fh_2=\langle i\mathrm{IIX}, i\mathrm{IIY}, i\mathrm{IIZ} \rangle_{\rm Lie}$.
In addition, we have a drift Hamiltonian $H_d=\mathrm{IZZ}$.
The control system is weakly connected and $\fk$ acts irreducibly.
We obtain that the dynamic Lie algebra is $\fk=\su(2)\hoplus \su(4)$. It is neither 
simple and nor fully controllable. In particular, the dynamic Lie algebra
does not respect our chosen tensor-product structure.
The problem is that the
control system 
\begin{center}
\includegraphics{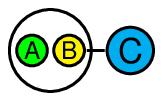}
\end{center}
is not weakly connected w.r.t.\ the tensor-product structures
$\su(2)\hoplus \su(4)$ and
$\su(2)\hoplus \su(2) \hoplus \su(2)$.
\end{example}

Generalising Sec.~\ref{sec:sym_cons},
we say that a control system is \emph{connected},
if the dynamic algebra $\fk$ contains an element of
$\su(N)\setminus[\su(e_1) \hoplus \su(e_2)]$
for each tensor-product structure $\su(e_1) \hoplus \su(e_2)$ with
$N=e_1 e_2$ and $e_1, e_2 \geq 2$. For control systems with pair interactions
this definition is equivalent to the one given in Sec.~\ref{sec:sym_cons}.

\begin{lemma}\label{lem:refine}
The following are equivalent:\\
(1) The control system is connected.\\
(2) The control system is weakly connected w.r.t.\ the 
common unique refinement\\ \phantom{FX} of its tensor-product structure.\\[0.3mm]
(3) The control system is weakly connected w.r.t.\ any
tensor-product structure. \phantom{XX}
\hfill$\blacksquare$
\end{lemma}

We now generalise Theorem~\ref{pair_connected} and 
prove that the dynamic algebra $\fk$
is simple if its centraliser is trivial 
and the corresponding control system is connected.

\begin{theorem}\label{lem:simplicity}
Assume that the dynamic algebra $\fk$ of a bilinear control system on $\su(N)$ has a 
trivial centraliser $\fk'$. Then one finds:\\
(1) The dynamic algebra $\fk$ is given in an irreducible representation.\\
(2) If $\fk$ is semi-simple but not simple, then $\fk \neq \su(N)$ and the control system 
is\\ \phantom{IX} not fully controllable.\\
(3) The dynamic algebra $\fk$ is simple iff
the control system is connected.
\end{theorem}
\begin{proof}
(1) immediately follows from $\fk'$ being trivial and Lemma~\ref{lem:centraliser}, 
while (2) is obvious, as $\su(N)$ is simple. We now prove the `if'-part of (3).
We obtain from Lemma~\ref{lem:semi-simplicity} that $\fk$ is simple or semi-simple.
In the following, we assume that $\fk$ is not simple. Thus
$\fk$ is a irreducible semi-simple (but not simple)  subalgebra of $\su(N)$.
Using Thm.~2.1 of Ref.~\cite{Dynkin57}, it follows that
$\fk=\fk_1 \hoplus \fk_2 \hoplus \cdots \hoplus \fk_m$ and that the
$\fk_j$ are irreducible simple subalgebras of some $\su(d_j)$ such that
$\fk \subseteq \su(d_1) \hoplus \su(d_2) \hoplus \cdots \hoplus \su(d_m)$, 
$\prod_{j=1}^m d_j = N$, and $m\geq 2$. In particular, we can choose 
two non-zero algebras 
$$\fh_1=\hoplus_{j\in\mathcal{I}_1} \fk_j\quad \text{ and }\quad \fh_2=\hoplus_{j\in\mathcal{I}_2} \fk_j,$$ 
where $\fk=\fh_1 \hoplus \fh_2\subseteq \su(c_1) \hoplus \su(c_2)$,
$\mathcal{I}_1\cup\mathcal{I}_2 = \{1,2,\ldots,m\}$, 
$\mathcal{I}_1\cap \mathcal{I}_2 = \{\}$, and $c_1 c_2=N$.
As the control system is connected, the dynamic algebra $\fk$
contains an element of  $\su(N)\setminus [\su(c_1) \hoplus \su(c_2)]$ for 
each tensor-product structure
$\su(c_1) \hoplus \su(c_2)$.
This is a contradiction to $\fk\subseteq \su(c_1) \hoplus \su(c_2)$ and 
the `if'-part of (3) follows. To prove the `only if'-part of (3) we assume that the control system is not connected.
It immediately follows that the dynamic algebra has to be a (non-trivial) direct sum.
Thus it cannot be simple, which proves the `only if'-part by contradiction.
\hfill$\blacksquare$
\end{proof}

In important special cases more convenient conditions hold:
\begin{corollary}\label{cor:connected}
Given a bilinear control system on $\su(N)$, where the centraliser $\fk'$
of the dynamic algebra $\fk$
is trivial.
We obtain:\\
(1) Assume that the subsystems of the tensor-product structure  are independently
fully controllable. The dynamic algebra $\fk$ is simple if and only if
the control system is weakly connected. \\
(2) Assume
that 
the tensor-product structure of the control system
is given by $\su(p_1) \hoplus \su(p_2) \hoplus \cdots \hoplus \su(p_n)$,
where $\prod_{j=1}^n p_j$ is a factorization of $N$ into prime numbers.
For example, $p_j=2$ for all $j$. The following are equivalent:\\
(a) The dynamic algebra $\fk$ is simple.\\
(b) The control system is weakly connected.\\
(c) The control system is connected.
\end{corollary}
\begin{proof} We first prove the case of (1).
As  
the subsystems are independently fully controllable,
any irreducible semi-simple (but not simple) dynamic algebra $\fk\supseteq \su(d_1) \hoplus \su(d_2) \hoplus \cdots \hoplus \su(d_m)$ has to be (irreducibly) contained
in  the algebra $\fh=\su(d'_1) \hoplus \su(d'_2) \hoplus \cdots \hoplus \su(d'_m)$ where 
$\su(d_1) \hoplus \su(d_2) \hoplus \cdots \hoplus \su(d_m)$ 
is a refinement of the tensor-product structure $\fh$.
All these cases are excluded as the control system is weakly connected, and (1) follows along the same lines as Theorem~\ref{lem:simplicity}.
As $\su(p_1) \hoplus \su(p_2) \hoplus \cdots \hoplus \su(p_n)$ is the common
unique refinement of all tensor-product structures
in the case of (2), the control system is weakly connected 
if and only if it is connected
(by Lemma~\ref{lem:refine}) 
and (2) follows by Theorem~\ref{lem:simplicity}. 
\hfill$\blacksquare$
\end{proof}

\section{Computational Techniques for Representation Theory}\label{app:repr}
For computationally exploiting Lie theory to list all irreducible representations of
a given dimension $N$ for all irreducible simple subalgebras of $\su(N)$, a
{\em self-consistent frame} is indispensable. It requires the highest weights and the
dimensions of their respective representations to be linked to the classification
by the standard Dynkin diagrams. Here we explicitly give all the details 
in such a consistent frame, since combining different literature sources 
runs the risk of arriving at erroneous results due to
possibly inconsistent conventions.

In particular, the appendix is meant to complement Sec.~\ref{sec:irred-sub-su}.
It describes the methods we used to compute the irreducible simple
subalgebras of $\su(N)$ and their inclusion relations.

\subsection{Highest Weights and Dimension Formulas\label{highestweigths}}
The irreducible simple subalgebras of $\su(N)$ are found by enumerating 
for all simple Lie algebras all their irreducible representations of dimension $N$.
The irreducible representations can be enumerated using highest weights
$(x_1,\ldots, x_{\ell})$ which are (non-negative) integer vectors.
The length $\ell$ of the highest weight is given by 
the rank (i.e.\ dimension of the maximal abelian subalgebra)
of the considered Lie algebra. Details on the theory of
highest weights can be found in Chap.~IX, Sec.~7 of Ref.~\cite{Bourb08b}.

Different orderings for the coefficients $x_i$ of the highest weights
are used in the literature. We use the so-called Bourbaki ordering which
is detailed in Tab.~\ref{tab:diag} by numbering the nodes
of the Dynkin diagrams (see Chap.~VI, Sec.~4.2, Thm.~3 of Ref.~\cite{Bourb08a}) 
for the compact simple Lie algebras. In 
Tab.~\ref{tab:standard} we present the highest weights and dimensions
for the standard representation of each compact simple Lie algebra. 
We put highest weights together if they differ only w.r.t.\
an outer automorphism, i.e.~an permutation which leaves the Dynkin diagram invariant.
Note that the standard representation is the lowest-dimensional (non-trivial) representation 
[with the exception of $\so(3)$, $\so(5)$, and $\so(6)$]
and is typically used to introduce the corresponding Lie algebra in  matrix
form.

\begin{table}[Ht!]
\caption{\label{tab:diag}
The Compact Simple Lie Algebras and their Dynkin Diagrams} 
\begin{center}
\begin{tabular}{c}
\hline\hline\\[-1mm]
\includegraphics{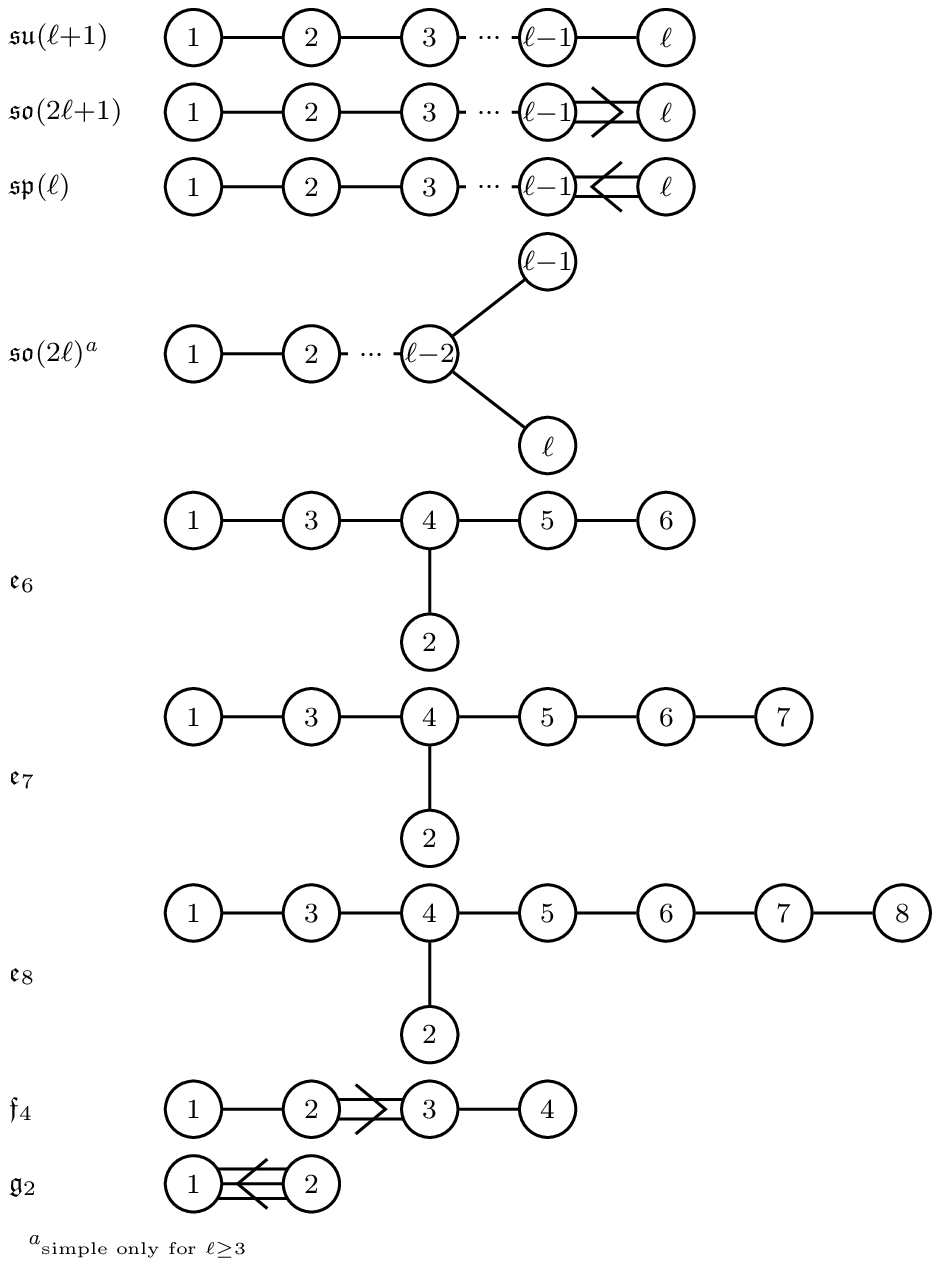}\\[1mm]
\hline\hline
\end{tabular}
\end{center}
\end{table}

\begin{table}[Ht!]
\caption{\label{tab:standard}
The Compact Simple Lie Algebras and their Standard Representations (with the corresponding
dimensions)}
\begin{center}
\fontencoding{OT1}
\fontfamily{cmr}
\fontseries{m}
\fontshape{n}
\fontsize{8}{9.5}
\selectfont
\begin{tabular}{@{\hspace{1mm}}c@{\hspace{-1.5mm}}}
\hline\hline\\[-0.3mm]
\begin{minipage}{10.3cm}
\renewcommand{\footnoterule}{}
\begin{tabular}[t]{c|c}
\begin{tabular}[t]{c@{\hspace{4pt}}l@{\hspace{4pt}}r}
algebra & highest weight(s) & dim\\ \midrule
$\su(\ell{+}1)$ & $\uebertwo{(1,\overbrace{0,\ldots,0}^{\ell{-}1}),}{(\underbrace{0,\ldots,0}_{\ell{-}1},1)\phantom{,}}$ & $\ell{+}1$ \\ \midrule
$\so(3)$ & $(2)$ & $3$\\ \midrule
$\uebertwo{\so(2\ell{+}1)}{\ell{\geq}2}$ & $(1,\underbrace{0,\ldots,0}_{\ell{-}1})$ & $2\ell{+}1$\\ \midrule
$\usp(\ell)$ & $(1,\underbrace{0,\ldots,0}_{\ell{-}1})$ & $2\ell$\\ \midrule
$\so(2)$\footnote{\hspace{-1.5mm}not simple} & $(1)$ & $2$\\ \midrule
$\so(4)$\tmark & $(1,1)$ & $4$\\ \midrule
\end{tabular}
&
\begin{tabular}[t]{c@{\hspace{4pt}}l@{\hspace{4pt}}r}
algebra & highest weight(s) & dim\\ \midrule
$\so(6)$ & $(1,0,0)$ & $6$\\ \midrule
$\so(8)$ & $
\begin{matrix}
(1,0,0,0),\\
(0,0,1,0),\\
(0,0,0,1)\phantom{,}
\end{matrix}$ & $8$\\ \midrule
$\uebertwo{\so(2\ell)}{\ell{\geq}5}$ & $(1,\underbrace{0,\ldots,0}_{\ell{-}1})$ & $2\ell$\\ \midrule
$\fe_6$ & $\uebertwo{(1,0,0,0,0,0),}{(0,0,0,0,0,1)\phantom{,}}$ & $27$ \\ \midrule
$\fe_7$ & $(0,0,0,0,0,0,1)$ & $56$ \\ \midrule
$\fe_8$ & $(0,0,0,0,0,0,0,1)$ & $248$ \\ \midrule
$\ff_4$ & $(0,0,0,1)$ & $26$ \\ \midrule
$\fg_2$ & $(1,0)$ & $7$ \\ \midrule
\end{tabular}
\end{tabular}
\vspace{-6mm}
\end{minipage}
\vspace{2mm}\\
\hline\hline
\end{tabular}
\end{center}
\end{table}

We already remarked in Sec.~\ref{sec:irred-sub-su} that the dimensions of
irreducible representations can 
be  efficiently computed using computer algebra systems such as {\sf LiE}~\cite{LIE222} 
and {\sf MAGMA}~\cite{MAGMA} via Weyl's dimension formula.
Now we present explicit formulas for the dimensions, 
which allowed us to speed up the computation of the dimensions
considerably. While for $\su(\ell+1)$, $\so(2\ell+1)$, and $\usp(\ell)$  
these formulas can readily be found  on pp.~340-341 
of Ref.~\cite{GW09}, we had to correct the one for $\so(2\ell)$,
since we could not find a reference with the proper formula either.
\begin{lemma}[Classical Lie algebras]\label{irred_dim}
Given the highest weight $(x_1,\ldots,x_{\ell})$
the dimensions of 
the corresponding
irreducible representations are:
\begin{enumerate}
\item $\su(\ell+1):\, \dim
= \prod_{1\leq i < j \leq \ell+1} \left\{ 1+ \frac{x_i+\cdots +x_{j-1}}{j-i} \right\}$
\item $\so(2\ell+1):\,
\dim = \prod_{1\leq i<j\leq \ell} \left\{ 1+ \frac{x_i+\cdots +x_{j-1}+2(x_j+\cdots 
+ x_{\ell-1})+x_{\ell}}{2\ell+1-i-j} \right\}
$\\
$\times
\prod_{1\leq i < j \leq \ell} \left\{ 1+ \frac{x_i+\cdots +x_{j-1}}{j-i} \right\}
\times
\prod_{1\leq i \leq \ell}
\left\{
1+ \frac{2(x_i+\cdots + x_{\ell-1}) + x_{\ell}}{2\ell+1-2i}
\right\}
$
\item $\usp(\ell):\,
\dim=
\prod_{1\leq i < j \leq \ell} \left\{ 1+ \frac{x_i+\cdots +x_{j-1}}{j-i} \right\}
\times
\prod_{1\leq i \leq \ell}
\left\{
1+ \frac{x_i+\cdots + x_{\ell}}{\ell+1-i}
\right\}
$\\
$\times
\prod_{1\leq i<j\leq \ell} \left\{ 1+ \frac{x_i+\cdots +x_{j-1}+2(x_j+\cdots 
+ x_{\ell})}{2\ell+2-i-j} \right\}
$
\item $\so(2\ell):\,
\dim =
\prod_{1\leq i < j \leq \ell} \left\{ 1+ \frac{x_i+\cdots +x_{j-1}}{j-i} \right\}
\times
\prod_{1\leq i \leq \ell-1}
\left\{
1+ \frac{x_i+\cdots + x_{\ell-2}+x_{\ell}}{\ell-i}
\right\}
$\\
$\times
\prod_{1\leq i<j\leq \ell-1} \left\{ 1+ \frac{x_i+\cdots +x_{j-1}+2(x_j+\cdots 
+ x_{\ell-2})+x_{\ell-1}+x_{\ell}}{2\ell-i-j} \right\}
$
$\hfill\blacksquare$
\end{enumerate}
\end{lemma}
Here we present the dimension formulas 
for the exceptional Lie algebras only for 
$\fg_2$ (cp.\ Ref.~\cite{Jacobson79}, pp.~257-258)
and $\ff_4$, 
ommitting the even longer and more complicated ones
for  $\fe_6$, $\fe_7$, and $\fe_8$.
We remark that these formulas are---in principle---well known but are usually not given
in the literature due to their complexity.
\begin{lemma}[$\fg_2$ and $\ff_4$]
Given the highest weight $(x_1,\ldots,x_{\ell})$
the dimensions of 
the corresponding
irreducible representations are:
\begin{enumerate}
\item $\fg_2:\,
\dim = 
(1+x_2)
(1+x_1)
\left(1{+}\frac{x_1 + x_2}{2}\right)
\left(1{+}\frac{x_1 + 2x_2}{3}\right)
\left(1{+}\frac{x_1 + 3x_2}{4}\right)
\left(1{+}\frac{2x_1 + 3x_2}{5}\right)
$
\item $\ff_4:\,
\dim = 
(1+x_4)
(1+x_3)
(1+x_2)
(1+x_1)
\left(1{+}\frac{x_3 + x_4}{2}\right)
\left(1{+}\frac{x_2 + x_3}{2}\right)
\left(1{+}\frac{x_1 + x_2}{2}\right)
$\\
$\times
\left(1{+}\frac{x_2 + x_3 + x_4}{3}\right)
\left(1{+}\frac{2x_2 + x_3}{3}\right)
\left(1{+}\frac{x_1 + x_2 + x_3}{3}\right)
\left(1{+}\frac{2x_2 + x_3 + x_4}{4}\right)
\left(1{+}\frac{x_1 + x_2 + x_3 + x_4}{4}\right)
$\\
$\times
\left(1{+}\frac{x_1 + 2x_2 + x_3}{4}\right)
\left(1{+}\frac{2x_2 + 2x_3 + x_4}{5}\right)
\left(1{+}\frac{x_1 + 2x_2 + x_3 + x_4}{5}\right)
\left(1{+}\frac{2x_1 + 2x_2 + x_3}{5}\right)
$\\
$\times 
\left(1{+}\frac{x_1 + 2x_2 + 2x_3 + x_4}{6}\right)
\left(1{+}\frac{2x_1 + 2x_2 + x_3 + x_4}{6}\right)
\left(1{+}\frac{x_1 + 3x_2 + 2x_3 + x_4}{7}\right)
$\\
$\times
\left(1{+}\frac{2x_1 + 2x_2 + 2x_3 + x_4}{7}\right)
\left(1{+}\frac{2x_1 + 3x_2 + 2x_3 + x_4}{8}\right)
\left(1{+}\frac{2x_1 + 4x_2 + 2x_3 + x_4}{9}\right)
$\\
$\times
\left(1{+}\frac{2x_1 + 4x_2 + 3x_3 + x_4}{10}\right)
\left(1{+}\frac{2x_1 + 4x_2 + 3x_3 + 2x_4}{11}\right)
$ 
\end{enumerate}
\end{lemma}
\begin{proof}
Computational explicit dimension formulas for the exceptional Lie algebras
were obtained using the computer algebra system {\sf MAGMA}~\cite{MAGMA} via Weyl's dimension formula.
$\hfill\blacksquare$
\end{proof}

We emphasize that in order to compute the dimensions efficiently,
one has to use the dimension formulas in the given factorized 
form. That is, one has to evaluate each factor and 
multiply the results. The alternative of evaluating the multiplied formula
is considerably less efficient. 

\subsection{Enumerating Representations\label{enum_repr}}
The aim of determining the irreducible simple subalgebras of $\su(N)$ 
for a given $N$ is reached by enumerating for all simple Lie algebras all 
their irreducible representations of dimension $N$. Therefore, we
have to enumerate for all simple Lie algebras all highest weights
$(x_1,\ldots, x_{\ell})$ corresponding to irreducible representations of dimension $N$. 
In doing so, how can one reduce the combined search space of Lie algebras and highest weights?

To this end, recall that the standard representation is the  lowest-dimensional (non-trivial) representation
[with the exception of $\so(3)$, $\so(5)$, and $\so(6)$].
It follows from the dimension formulas for the standard representations
in Tab.~\ref{tab:standard} that only a finite number of Lie algebras
have irreducible representations of dimension equal (or less than or equal) to a given $N$.  
Thus we have to search only through a finite set of Lie algebras.
In addition, we have to consider merely one instance of isomorphic Lie algebras 
[$\su(2) \cong \so(3) \cong \usp(1)$, $\so(5) \cong \usp(2)$, and 
$\su(4) \cong \so(6)$] and can neglect $\so(2)$ and $\so(4)$ as they are not simple. 
It follows from Chap.~IX, Sec.~8.5, Cor.~2 of Ref.~\cite{Bourb08b}
that for a Lie algebras the set of irreducible representations of dimension less than
or equal to $N$ is finite. We obtain:
\begin{lemma}\label{finite_number}
Each Lie algebra can only have a finite number of irreducible
representations of dimension equal (or less than or equal) to a given $N$.
Furthermore, only a finite number of Lie algebras have any irreducible
representations of dimension equal (or less than or equal) to a given $N$.
$\hfill\blacksquare$
\end{lemma}
It follows from Lemma~\ref{finite_number} that the search space for 
the highest weights is finite. Using the following Lemma, one can obtain 
stopping criteria for the search for the highest weights with dimension equal 
(or less than or equal) to a given $N$.
\begin{lemma}\label{prop_irred}
For a given Lie algebra, let $\dim[x]$
denote the dimension of an irreducible representation with
highest weight $x=(x_1,\ldots,x_{\ell})$.
\begin{enumerate}
\item The dimension is strongly monotonic ascending in each entry $x_i$ of the highest weight: 
$\dim[(x_1,\ldots,x_i+1,\ldots,x_{\ell})] > \dim[(x_1,\ldots,x_i,\ldots,x_{\ell})]$.
\item Let $e^i$ denote the vector such that $(e^i)_j=\delta_{i,j}$.
If $\sum x_i > 1$, then
$$\dim[x]\geq \min[\{\max[\dim(e^i),\dim(e^j)]\}_{i\neq j} \cup 
\{\dim(2 \e^i)\}_{1\leq i \leq \ell} ].$$
\end{enumerate}
\end{lemma}
\begin{proof}
See, e.g., Cor.~5.2 and Cor.~5.4 of Ref.~\cite{Stuck91}.
$\hfill\blacksquare$
\end{proof}

Let us fix a Lie algebra. We start our search for the highest weights of dimension less than or equal to
$N$ with all vectors $x=(x_1,\ldots,x_{\ell})$ such that $\sum x_i =1$.
In each step, we compute the dimension corresponding to the highest weight
$x=(x_1,\ldots,x_{\ell})$. If $\dim[x]\leq N$,
we include $x$ into the list of highest weights with dimension $\dim[x]$
and we branch our search to all $\tilde{x}=(x_1,\ldots,x_i+1,\ldots,x_{\ell})$.
If $\dim[x]>N$, we prune this branch in our search tree
(Part 1 of Lemma~\ref{prop_irred}).
We have to search through all Lie algebras such that the lowest-dimensional (non-trivial) 
irreducible representation is less than or equal to $N$.
One can further reduce the search space with respect to potential Lie algebras 
by using knowledge on the second-lowest-dimensional (non-trivial) irreducible representations:
\begin{theorem} For a given Lie algebra,
let $y=(y_1,\ldots,y_{\ell})$ denote the highest weight
of the second-lowest-dimensional (non-trivial) irreducible representation.
\begin{enumerate}
\item For $\su(\ell+1)$ and $\ell\geq 3$, we obtain 
$y=(0,1,0,\ldots,0)$ or $y=(0,\ldots,0,1,0)$. In addition, $\dim[y]=\ell(\ell+1)/2$.
\item For $\so(2\ell+1)$ and $\ell\geq 7$, we obtain 
$y=(0,1,0,\ldots,0)$ and $\dim[y]=(2\ell+1)\ell$.
\item For $\usp(\ell)$ and $\ell\geq 4$, we obtain 
$y=(0,1,0,\ldots,0)$ and $\dim[y]=2\ell^2-\ell-1$.
\item For $\so(2\ell)$ and $\ell\geq 8$, we obtain 
$y=(0,1,0,\ldots,0)$ and $\dim[y]=2\ell^2-\ell$.
\end{enumerate}
\end{theorem}
\begin{proof}
Again, let $e^i$ denote the vector such that $(e^i)_j=\delta_{i,j}$.
We apply Lemma~\ref{prop_irred} in each of the following cases:
\begin{enumerate}
\item $\su(\ell+1)$: Recall that $\dim[e^r]=\tbinom{\ell+1}{r}$ for $1\leq r \leq \ell$. 
It follows that $\dim[e^2]=\ell(\ell+1)/2$ for $\ell\geq 3$.
One can deduce from Lemma~\ref{irred_dim} that $\dim[2e^1]=(\ell+1)(\ell+2)/2$ and
that $\dim[(1,0,\ldots,0,1)]=\ell(\ell+2)$ for $\ell\geq 2$. 
We obtain that $\dim[e^2]<\dim[2e^1]<\dim[(1,0,\ldots,0,1)]$
for $\ell\geq 3$. The first part follows.
\item $\so(2\ell+1)$: Recall that $\dim[e^{\ell}]=2^{\ell}$ and
$\dim[e^r]=\tbinom{2\ell+1}{r}$ for $1\leq r \leq \ell-1$
(see, e.g., p.~340 of Ref.~\cite{GW09}). It follows that $\dim[e^2]=(2\ell+1)\ell$ for $\ell\geq 3$. 
We obtain that $\dim[e^2]<\dim[e^{\ell}]$ for $\ell\geq 7$. One can deduce from Lemma~\ref{irred_dim} 
that $\dim[2e^1]=\ell(2\ell+3)>\dim[e^2]$ for $\ell\geq 3$. The second part follows.
\item $\usp(\ell)$: Recall that 
$\dim[e^1]=2\ell$ and
$\dim[e^r]=\tbinom{2\ell}{r}-\tbinom{2\ell}{r-2}$ for $2\leq r \leq \ell$
(see, e.g., p.~341 of Ref.~\cite{GW09}). It follows that
$\dim[e^2]=2\ell^2-\ell-1$. We obtain that $\dim[e^r]=\frac{2\ell+2-2r}{2\ell+2}\tbinom{2\ell+2}{r}$. One can deduce that
$\dim[e^3]=\frac{4}{3}\ell^3-2\ell^2-\frac{4}{3}\ell>\dim[e^2]$ for $\ell\geq 4$.
If $r\geq 4$, it follows that $\dim[e^r]-\dim[e^2]=\frac{2\ell+2-2r}{2\ell+2}\tbinom{2\ell+2}{r}-2\ell^2+\ell+1\geq \frac{1}{\ell+1} \tbinom{2\ell+2}{4}-2\ell^2+\ell+1=\frac{2}{3} \ell^3 -2\ell^2+\frac{5}{6}\ell+1>0$ for $\ell\geq 4$. One can obtain from Lemma~\ref{irred_dim} that $\dim[2e^1]=2\ell^2+\ell>\dim[e^2]$. The third part follows.
\item $\so(2\ell)$: Recall that 
$\dim[e^{\ell-1}]=\dim[e^{\ell}]=2^{\ell-1}$ and
$\dim[e^r]=\tbinom{2\ell}{r}$ for $1\leq r \leq \ell-2$
(see, e.g., p.~341 of Ref.~\cite{GW09}). It follows that $\dim[e^2]=2\ell^2-\ell$ for $\ell\ge 4$. 
We obtain that $\dim[e^2]<\dim[e^\ell]$ for $\ell \geq 8$. One can deduce from Lemma~\ref{irred_dim} 
that $\dim[2e^1]=2\ell^2+\ell-1>\dim[e^2]$ for $\ell\ge 4$. The fourth part follows.$\hfill\blacksquare$
\end{enumerate}
\end{proof}
Now one obtains bounds on $\ell$ such that the dimension of the second-lowest-dimensional 
(non-trivial) irreducible representation is greater than $N$:
\begin{corollary}\label{second_lowest}
For a given Lie algebra,
let $y=(y_1,\ldots,y_{\ell})$ denote the highest weight
of the second-lowest-dimensional (non-trivial) irreducible representation.
\begin{enumerate}
\item For $\su(\ell+1)$ and $\ell\geq 3$, we obtain:
$\dim[y]>N \Leftrightarrow \ell > \sqrt{1/4+2N}-1/2
$
\item For $\so(2\ell+1)$ and $\ell\geq 7$, we obtain:
$\dim[y]>N \Leftrightarrow \ell > (\sqrt{1+8N}-1)/4
$
\item For $\usp(\ell)$ and $\ell\geq 4$, we obtain:
$\dim[y]>N \Leftrightarrow \ell > (\sqrt{9+8N}+1)/4
$
\item For $\so(2\ell)$ and $\ell\geq 8$, we obtain:
$\dim[y]>N \Leftrightarrow \ell > (\sqrt{1+8N}+1)/4
$$\hfill\blacksquare$
\end{enumerate}
\end{corollary}
Now we explain how to use Corollary~\ref{second_lowest} in order to reduce the search space. 
Consider  $\su(k+1)$ and $k\geq 3$. If $N\geq k+1$ but $k > \sqrt{1/4+2N}-1/2$ then the standard 
representation of $\su(k+1)$ occurs with dimension less than or equal to $N$. But no other
(non-trivial) 
irreducible representation of $\su(k+1)$ has dimension less than or equal to $N$. We include the 
highest weight of the standard representation in the list corresponding to the dimension $k+1$. 
But we do not have to search for other irreducible representations. The search space is reduced 
from a size linear in $N$ to a square-root in $N$.

\subsection{Inclusion Relations\label{incl_rel}}
Once having obtained all irreducible simple subalgebras of $\su(N)$,
one can determine their inclusion relations
following the work of 
Dynkin~\cite{Dynkin57} (see, e.g., Chap.~6, Sec.~3.2 of Ref.~\cite{GOV94}).
Refer to \cite{Onishchik00} for related literature.
For example, Refs.~\cite{Seitz87,Testerman88}
generalise the work of Dynkin \cite{Dynkin57} 
to classical and exceptional Lie algebras over prime fields.
References~\cite{Min06,DM10,deGraaf11} contain most recent findings.
It follows from Theorem~1.5 in \cite{Dynkin57} that almost all irreducible simple subalgebras
(of dimension $\dim$) are maximal in $\su(\dim)$, $\usp(\dim/2)$, or $\so(\dim)$. 
Relying on Tab.~I of \cite{Dynkin57}, the exceptions
are listed in Tab.~\ref{tab:Dynkin}, which contains 
irreducible simple subalgebras of $\su(\dim)$ including the algebra in which the subalgebra is maximal. 
In addition, the highest weights of the corresponding representations as well as the type of the 
subalgebra (s for symplectic, o for orthogonal, and u for unitary) are given. 
For reference, we give the Malcev classification~\cite{Malcev50} (see also, e.g., \cite{Dynkin57,BP70,McKay81})
of symplectic, orthogonal, and unitary representations:
\begin{theorem}[Malcev]\label{Malcev}
Let $x=(x_1,\ldots,x_{\ell})$ denote the highest weight corresponding to an irreducible representation $\phi$
of a Lie algebra $\fk$. As $\phi$ is irreducible, the different cases
of symplectic, orthogonal, and unitary representations are mutually exclusive:
\begin{enumerate}
\item $\fk=\su(\ell+1)$:\footnote{Recall that $\mathrm{div}$ denotes integer division, e.g., $(5\, \mathrm{div}\, 2)=2$.}
\begin{enumerate}
\item $\phi$ is symplectic if $x$ is symmetric, $(\ell\bmod{4}) = 1$,
and $x_{((\ell-1)\, \mathrm{div}\, 2)+1}$ is odd.
\item $\phi$ is orthogonal if $x$ is symmetric
as well as either (i) $(\ell\bmod{4}) = 1$ and $x_{((\ell-1)\, \mathrm{div}\, 2)+1}$ is even or (ii)
$(\ell\bmod{4}) \neq 1$.
\item $\phi$ is unitary if $x$ is not symmetric.
\end{enumerate}
\item $\fk=\so(2\ell+1)$ for $\ell \geq 2$:
\begin{enumerate}
\item $\phi$ is symplectic if $(\ell\bmod{4}) \in \{1,2\}$ and $x_\ell$ is odd.
\item $\phi$ is orthogonal if either $(\ell\bmod{4}) \in \{0,3\}$ or $x_\ell$ is even.
\end{enumerate}
\item $\fk=\usp(\ell)$ for $\ell \geq 2$:
$\phi$ is symplectic if $\sum_{1\leq 2j+1 \leq \ell} x_{2j+1}$ is odd ($j\in\N\cup\{0\}$). Otherwise, $\phi$ is orthogonal.
\item $\fk=\so(2\ell)$ for $\ell \geq 3$:
\begin{enumerate}
\item $\phi$ is symplectic if $(\ell\bmod{4})=2$ and $x_{\ell-1}+x_{\ell}$ is odd.
\item $\phi$ is orthogonal if either (i)
$(\ell\bmod{4})=2$ and $x_{\ell-1}+x_{\ell}$ is even,\\
(ii) $(\ell\bmod{4})=0$, or (iii)
$\ell$ is odd and $x_{\ell-1}= x_{\ell}$.
\item $\phi$ is unitary if $\ell$ is odd and $x_{\ell-1}\neq x_{\ell}$.
\end{enumerate}
\item $\fk=\fg_2$, $\fk=\ff_4$, or $\fk=\fe_8$: $\phi$ is always orthogonal.
\item $\fk=\fe_6$: $\phi$ is orthogonal if $x_1=x_6$ and $x_3=x_5$. Otherwise, $\phi$ is unitary.
\item $\fk=\fe_7$: $\phi$ is symplectic if $x_2+x_5+x_7$ is odd. Otherwise, $\phi$ is orthogonal. $\hfill\blacksquare$
\end{enumerate}
\end{theorem}

\begin{table}[Ht!]
\caption{\label{tab:Dynkin} 
Irreducible Simple Subalgebras not Maximal in $\su(\dim)$, $\usp(\dim/2)$, or $\so(\dim)$}
\begin{center}
\fontencoding{OT1}
\fontfamily{cmr}
\fontseries{m}
\fontshape{n}
\fontsize{8}{9.5}
\selectfont
\begin{tabular}{@{\hspace{3mm}}c@{\hspace{1mm}}}
\hline\hline\\[-0.9mm]
\begin{minipage}{11.2cm}
\renewcommand{\footnoterule}{}
\begin{tabular}[t]{c@{\hspace{4pt}}c@{\hspace{4pt}}l@{\hspace{-10pt}}r@{\hspace{4pt}}l@{\hspace{-10pt}}r}
subalgebra & type & highest weight(s) & algebra & highest weight(s) & dim\\ \midrule
$\ueber{\su(\ell+1)}{\ell\geq 4}$  & u & $\ueber{(1,0,1,0,\ldots,0),}{(0,\ldots,0,1,0,1)\phantom{,}}$ & $\su\hspace{-0.05cm}\left[\ell(\ell{+}1)/2\right]$ & $\ueber{(0,1,0,\ldots,0),}{(0,\ldots,0,1,0)\phantom{,}}$  & $3\tbinom{\ell+2}{4}$\\ \midrule[0pt]
$\ueber{\su(\ell+1)}{\ell\geq 3}$  & u & $\ueber{(2,1,0,\ldots,0),}{(0,\ldots,0,1,2)\phantom{,}}$ & $\su\hspace{-0.075cm}\left[\tfrac{\ell(\ell+3)}{2}{+}1\right]$ & $\ueber{(0,1,0,\ldots,0),}{(0,\ldots,0,1,0)\phantom{,}}$  & $3\tbinom{\ell+3}{4}$\\ \midrule[0pt]
$\su(2)$  & o & $(6)$ & $\fg_2$ & $(1,0)$  & $7$\\ \midrule[0pt]
$\su(6)$  & o & $(0,1,0,1,0)$ & $\usp(10)$ & $(0,1,0,\ldots,0)$  & $189$\\ \midrule[0pt]
$\ueber{\so(4k{+}3)}{\ueber{k\geq 1, m\geq 1}{(\text{but not}\hspace{1mm} k{=}m{=}1)\footnote{\hspace{-1.5mm}\tiny corrected, for $k{=}m{=}1$ we have $\so(7){\subset}\so(8){\subset}\su(8)$}
}}$  
& s/o\footnote{\hspace{-1.5mm}\tiny if $(k+1)m$ is odd then s else o} & $(0,\ldots,0,m)$  & $\so(4k{+}4)$ & $\ueber{(0,\ldots,0,m,0),}{(0,\ldots,0,0,m)\phantom{,}}$  & \footnote{\hspace{-1.5mm}\tiny corrected, $\prod_{s=1}^{2k+1} \left[ \tbinom{m+2s-1}{m}/\tbinom{m+s-1}{m} \right]$}\\ \midrule[0pt]
$\so(9)$  & o & $(1,0,0,1)$ & $\so(16)$ & $\ueber{(0,\ldots,0,1,0),}{(0,\ldots,0,0,1)\phantom{,}}$  & $128$\\ \midrule[0pt]
$\usp(3)$  & o & $(0,2,0)$ & $\usp(7)$ & $(0,1,0,0,0,0,0)$  & $90$\\ \midrule[0pt]
$\usp(3)$  & s & $(0,2,1)$ & $\usp(7)$ & $(0,0,1,0,0,0,0)$  & $350$\\ \midrule[0pt]
$\so(10)$  & u & $\ueber{(0,1,0,1,0)}{(0,1,0,0,1)}$ & $\su(16)$ & $\ueber{(0,0,1,0,\ldots,0)}{(0,\ldots,0,1,0,0)}$  & $560$\\ \midrule[0pt]
$\so(12)$  & o & $(0,0,0,1,0,0)$ & $\usp(16)$ & $(0,1,0,0,\ldots,0)$  & $495$\\ 
\midrule[0pt]
$\so(12)$  & s & $\ueber{(0,0,1,0,1,0)}{(0,0,1,0,0,1)}$ & $\usp(16)$ & $(0,0,1,0,\ldots,0)$  & $4928$\\ \midrule[0pt]
$\fe_6$  & u & $\ueber{(0,0,1,0,0,0)}{(0,0,0,0,1,0)}$ & $\su(27)$ & $(0,1,0,0,0,\ldots,0)$  & $351$\\ \midrule[0pt]
$\fe_6$  & u & $\ueber{(0,1,1,0,0,0)}{(0,1,0,0,1,0)}$ & $\su(27)$ & $(0,0,0,1,0,\ldots,0)$  & $17550$\\ \midrule[0pt]
$\fe_7$  & o & $(0,0,0,0,0,1,0)$ & $\usp(28)$ & $(0,1,0,\ldots,0)$  & $1539$\\ \midrule[0pt]
$\fe_7$  & s & $(0,0,0,0,1,0,0)$ & $\usp(28)$ & $(0,0,1,0,\ldots,0)$  & $27664$\\ \midrule[0pt]
$\fe_7$  & o & $(0,0,0,1,0,0,0)$ & $\usp(28)$ & $(0,0,0,1,0,\ldots,0)$  & $365750$\\ \midrule[0pt]
$\fe_7$  & s & $(0,1,1,0,0,0,0)$ & $\usp(28)$ & $(0,0,0,0,1,0,\ldots,0)$  & $3792096$\\ \midrule[0pt]
$\ueber{\fg_2}{m\geq 2}$  & o & $(m,0)$ & $\so(7)$ & $(m,0,0)$  & $\tfrac{2m+5}{5}\tbinom{m+4}{4}$\\ \midrule[0pt]
\end{tabular}
\vspace{-0.3cm}
\end{minipage}
\vspace{3mm}\\
\hline\hline
\end{tabular}
\end{center}
\end{table}

Results for dimension $\dim \leq 16$ can be found
in Tab.~\ref{tab:repr}, where the irreducible simple subalgebras
of $\su(\dim)$ are given again with their type 
(s for symplectic, o for orthogonal, and u for unitary) 
plus the highest weight of the corresponding 
irreducible representation. This information is essential
for deriving Tab.~\ref{tab:su_subalg}.

\begin{table}[Hp!]
\caption{\label{tab:repr} 
Highest Weights of the Irreducible Representations up to Dimension $16$ 
}
\begin{center}
\fontencoding{OT1}
\fontfamily{cmr}
\fontseries{m}
\fontshape{n}
\fontsize{8}{9.5}
\selectfont
\begin{tabular}{@{\hspace{1mm}}c@{\hspace{0mm}}}
\hline\hline\\[-0.8mm]
\begin{tabular}[t]{l|l}
\begin{tabular}[t]{c@{\hspace{4pt}}c@{\hspace{4pt}}c@{\hspace{4pt}}l}
$\dim$ & algebra & type & highest weight(s)\\ \addlinespace[0.1cm] \midrule \addlinespace[0.2cm]
$2$ & $\su(2)$ & s & $(1)$\\ \addlinespace[0.1cm] \midrule \addlinespace[0.2cm]
$3$ & $\su(3)$ & u & $(1,0)$, $(0,1)$\\
& $\su(2)$ & o & $(2)$\\ \addlinespace[0.1cm] \midrule \addlinespace[0.2cm]
$4$ & $\su(4)$ & u & $(1,0,0)$, $(0,0,1)$\\
& $\usp(2)$ & s & $(1,0)$\\
& $\su(2)$ & s &$(3)$\\ \addlinespace[0.1cm] \midrule \addlinespace[0.2cm]
$5$ & $\su(5)$ & u & $(1,0,0,0)$, $(0,0,0,1)$\\
& $\so(5)$ & o & $(1,0)$\\
& $\su(2)$ & o & $(4)$\\ \addlinespace[0.1cm] \midrule \addlinespace[0.2cm]
$6$ & $\su(6)$ & u & $(1,0,0,0,0)$,\\
& & & $(0,0,0,0,1)$\\
& $\usp(3)$ & s & $(1,0,0)$\\
& $\su(2)$ & s & $(5)$\\
& $\so(6)$ & o & $(1,0,0)$\\
& $\su(3)$ & u & $(2,0)$, $(0,2)$\\ \addlinespace[0.1cm] \midrule \addlinespace[0.2cm]
$7$ & $\su(7)$ & u & $(1,0,0,0,0,0)$,\\
& & & $(0,0,0,0,0,1)$\\
& $\so(7)$ & o & $(1,0,0)$\\
& $\fg_2$ & o & $(1,0)$\\
& $\su(2)$ & o & $(6)$\\ \addlinespace[0.1cm] \midrule \addlinespace[0.2cm]
$8$ & $\su(8)$ & u & $(1,0,0,0,0,0,0)$,\\
& & & $(0,0,0,0,0,0,1)$\\
& $\usp(4)$ & s & $(1,0,0,0)$\\
& $\su(2)$ & s & $(7)$\\
& $\so(8)$ & o & $(1,0,0,0)$, $(0,0,1,0)$\\
& & & $(0,0,0,1)$\\
& $\su(3)$ & o & $(1,1)$\\
& $\so(7)$ & o & $(0,0,1)$\\ \addlinespace[0.1cm] \midrule \addlinespace[0.2cm]
$9$ & $\su(9)$ & u & $(1,0,0,0,0,0,0,0)$,\\
& & & $(0,0,0,0,0,0,0,1)$\\
& $\so(9)$ & o & $(1,0,0,0)$\\
& $\su(2)$ & o & $(8)$\\ \addlinespace[0.1cm] \midrule \addlinespace[0.2cm]
$10$ & $\su(10)$ & u & $(1,0,0,0,0,0,0,0,0)$,\\
& & & $(0,0,0,0,0,0,0,0,1)$\\
& $\usp(5)$ & s & $(1,0,0,0,0)$\\
& $\su(2)$ & s & $(9)$\\
& $\so(10)$ & o & $(1,0,0,0,0)$\\
& $\so(5)$ & o & $(0,2)$\\
& $\su(3)$ & u & $(3,0)$, $(0,3)$\\
& $\su(4)$ & u & $(2,0,0)$, $(0,0,2)$\\
& $\su(5)$ & u & $(0,1,0,0)$, $(0,0,1,0)$
\end{tabular}
&
\begin{tabular}[t]{c@{\hspace{4pt}}c@{\hspace{4pt}}c@{\hspace{4pt}}l}
$\dim$ & algebra & type & highest weight(s)\\ \addlinespace[0.1cm] \midrule \addlinespace[0.2cm]
$11$ & $\su(11)$ & u & $(1,0,\ldots,0)$,\\
& & &  $(0,\ldots,0,1)$\\
& $\so(11)$ & o & $(1,0,0,0,0)$\\
& $\su(2)$ & o & $(10)$\\ \addlinespace[0.1cm] \midrule \addlinespace[0.2cm]
$12$ & $\su(12)$ & u & $(1,0,\ldots,0)$,\\
& & & $(0,\ldots,0,1)$\\
& $\usp(6)$ & s & $(1,0,0,0,0,0)$\\
& $\su(2)$ & s & $(11)$\\
& $\so(12)$ & o & $(1,0,0,0,0,0)$\\ \addlinespace[0.1cm] \midrule \addlinespace[0.2cm]
$13$ & $\su(13)$ & u & $(1,0,\ldots,0)$,\\
& & & $(0,\ldots,0,1)$\\
& $\so(13)$ & o & $(1,0,0,0,0,0)$\\
& $\su(2)$ & o & $(12)$\\ \addlinespace[0.1cm] \midrule \addlinespace[0.2cm]
$14$ & $\su(14)$ & u & $(1,0,\ldots,0)$,\\
& & & $(0,\ldots,0,1)$\\
& $\usp(7)$ & s & $(1,0,0,0,0,0,0)$\\
& $\su(2)$ & s & $(13)$\\
& $\usp(3)$ & s & $(0,0,1)$\\
& $\so(14)$ & o & $(1,0,0,0,0,0,0)$\\
& $\so(5)$ & o & $(2,0)$\\
& $\usp(3)$ & o & $(0,1,0)$\\
& $\fg_2$ & o & $(0,1)$\\ \addlinespace[0.1cm] \midrule \addlinespace[0.2cm]
$15$ & $\su(15)$ & u & $(1,0,\ldots,0)$,\\
& & & $(0,\ldots,0,1)$\\
& $\so(15)$ & o & $(1,0,0,0,0,0,0)$\\
& $\su(2)$ & o & $(14)$\\
& $\so(6)$ & o & $(0,1,1)$\\
& $\su(3)$ & u & $(4,0)$, $(0,4)$\\
& $\su(5)$ & u & $(2,0,0,0)$,\\
& & & $(0,0,0,2)$\\
& $\su(6)$ & u & $(0,1,0,0,0)$,\\
& & & $(0,0,0,1,0)$\\
& $\su(3)$ & u & $(2,1)$, $(1,2)$\\ \addlinespace[0.1cm] \midrule \addlinespace[0.2cm]
$16$ & $\su(16)$ & u & $(1,0,\ldots,0)$,\\
& & & $(0,\ldots,0,1)$\\
& $\usp(8)$ & s & $(1,0,0,0,0,0,0,0)$\\
& $\su(2)$ & s & $(15)$\\
& $\usp(2)$ & s & $(1,1)$\\
& $\so(16)$ & o & $(1,0,0,0,0,0,0,0)$\\
& $\so(9)$ & o & $(0,0,0,1)$\\
& $\so(10)$ & u & $(0,0,0,1,0)$,\\
& & & $(0,0,0,0,1)$
\end{tabular}
\end{tabular}
\vspace{3mm}\\
\hline\hline
\end{tabular}
\end{center}
\end{table}

\subsection{Examples}
We illustrate our
methods by  two concrete examples:
\begin{example}
We use the methods of Appendix~\ref{enum_repr} 
in the case of dimension $N=7$
and compute the irreducible simple subalgebras of $\su(7)$, where the corresponding irreducible representations
are specified by highest weights.\footnote{The definition of the highest weight is discussed in
Appendix~\ref{highestweigths}} We find the following  irreducible simple
subalgebras (see Tab.~\ref{tab:repr}): 
$\su(7)$ with highest weights\footnote{The two irreducible representations of $\su(7)$ 
are related by an inner automorphism.} 
$(1,0,0,0,0,0)$ and $(0,0,0,0,0,1)$,
$\so(7)$ with $(1,0,0)$, $\fg_2$ with $(1,0)$, as well as
$\su(2)$ with $(6)$. We conclude from Theorem~\ref{Malcev} that the given irreducible representations of $\su(7)$ are unitary and that all the other ones  are orthogonal (see Tab.~\ref{tab:repr}). It follows that
$\so(7)$ is directly embedded in $\su(7)$:
\begin{center}
\includegraphics{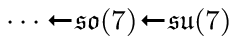}
\end{center}
The algebras $\fg_2$ and $\su(2)$ are embedded in $\so(7)$, but we still have to determine the inclusion relations. 
All algebras not listed
with the corresponding highest weight in Tab.~\ref{tab:Dynkin} 
are {\em directly contained} either in $\su(N)$, $\usp(N/2)$ [for $N$ even], or in
$\so(N)$ depending on whether the irreducible representation is unitary, symplectic, or orthogonal.
We find the algebra $\su(2)$ with the highest weight $(6)$ in the third row of Tab.~\ref{tab:Dynkin}. 
Thus the algebra $\su(2)$ is contained in $\so(7)$ but only indirectly so---via $\fg_2$:
\begin{center}
\includegraphics{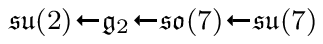}
\end{center}
\end{example}

\begin{example}
Consider the case of $N=16$. First, we obtain all the irreducible simple subalgebras of
$\su(16)$ (see Tab.~\ref{tab:repr}):
$\su(16)$ with highest weights $(1,0,\ldots,0)$ and $(0,\ldots,0,1)$ as well as
$\so(10)$ with $(0,0,0,1,0)$. The cases of irreducible symplectic representations are
$\usp(8)$ with highest weight $(1,0,0,0,0,0,0,0)$,
$\su(2)$ with $(15)$, and $\usp(2)$ with $(1,1)$. The irreducible orthogonal representations
are given by $\so(16)$ with highest weight $(1,0,0,0,0,0,0,0)$ and
$\so(9)$ with $(0,0,0,1)$. We immediately conclude that the algebras
$\usp(8)$, $\so(16)$, and $\so(10)$ are directly embedded in $\su(16)$:
\begin{center}
\includegraphics{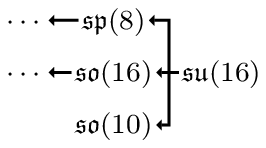}
\end{center}
It follows that all the other cases are {\em directly contained} either
in $\usp(8)$ or $\so(16)$ as they are not listed in Tab.~\ref{tab:Dynkin}: 
\begin{center}
\includegraphics{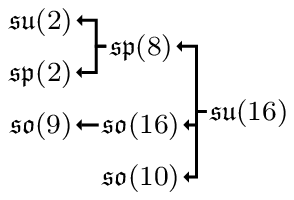}
\end{center}
\end{example}

\section{Alternating and Symmetric Squares of Representations\label{AltSym}}
In this Appendix, we enumerate for all compact semisimple Lie algebras
those representations whose alternating and symmetric squares
are both irreducible. We obtain that only the standard representation
of $\su(\ell+1)$  with $\ell \geq 0$ has this property.
This result is used in Sec.~\ref{sufficient}. We freely use the notation
of Appendix~\ref{app:repr}.

Assume that $\phi$ is a representation of a compact semisimple Lie algebra $\fg$ 
on a finite-dimensional vector space $V$ with basis $\{ v_1,\ldots ,v_k\}$. 
The representation $\phi$ is given as a map from $\fg$ to
the set of complex $k\times k$ matrices $\gl(k,\C)$. 
Starting from the representation $\phi$ we construct its
tensor square $\phi^{\otimes 2}=\phi\otimes\unity_k + \unity_k\otimes \phi$
which acts on the $k^2$-dimensional vector space $V\otimes V$ with basis
$\{ v_{i_1} \otimes v_{i_{2}}\,|\, i_1, i_2 \in \{1,\ldots,k\} \}$. This action
is defined on the basis  by ($g \in \fg$)
$$\phi^{\otimes 2}(g) [v_{i_1} \otimes v_{i_2}]:=
[\phi(g) v_{i_1}] \otimes v_{i_2} + v_{i_1} \otimes [\phi(g) v_{i_2}]
$$
and it can be extended to the full vector space $V\otimes V$ by linearity. 
Now we can define for $\phi$ its alternating square $\Alt^2 \phi:=\phi^{\otimes 2}|_{\Alt^2 V}$
by restricting $\phi^{\otimes 2}$ to the $k(k-1)/2$-dimensional subspace
$\Alt^2 V \subset V \otimes V$ with basis
$$\{v_{i_1} \otimes v_{i_2} - v_{i_2} \otimes v_{i_1}\,|\, i_1, i_2 \in \{1,\ldots,k\} \text{ and } i_1 \neq i_2   \}.$$
It is clear that  $\Alt^2 \phi$ is well defined as ($g\in \fg$)
\begin{eqnarray*}
(\Alt^2 \phi)(g)[v_{i_1}\otimes v_{i_2} - v_{i_2} \otimes v_{i_1}]
& = & \left(
[\phi(g) v_{i_1}] \otimes v_{i_2} - v_{i_2} \otimes [\phi(g) v_{i_1}]
\right)\\
& + & \left(
v_{i_1} \otimes [\phi(g) v_{i_2}] - [\phi(g) v_{i_2}] \otimes v_{i_1}
\right)
\end{eqnarray*}
is contained in $\Alt^2 V$. Similarily, one defines the symmetric square 
$\Sym^2 \phi:=\phi^{\otimes 2}|_{\Sym^2 V}$ as the restriction to 
the $k(k+1)/2$-dimensional subspace
$\Sym^2 V \subset V \otimes V$ with basis
$\{v_{i_1} \otimes v_{i_2} + v_{i_2} \otimes v_{i_1}\,|\, i_1, i_2 \in \{1,\ldots,k\} \}$.
We obtain that the tensor square $\phi^{\otimes 2}=\Alt^2 \phi \oplus \Sym^2 \phi$
decomposes in a direct sum,
exactly as the tensor product $V \otimes V=\Alt^2 V \oplus \Sym^2 V$. 
Dynkin~\cite{Dynkin57} classified the cases when $\Alt^2 \phi$ is irreducible:

\begin{theorem}[Dynkin]\label{alttable}
Assume $\phi$ is a (finite-dimensional) representation of a compact semi\-simple Lie algebra $\fg$.
The representation $\Alt^2 \phi$ is irreducible if and only if $\phi$ is irreducible and the pair ($\fg$, $\phi$) 
is  (up to an outer automorphism of $\fg$) given in the following table:
\begin{center}
\fontencoding{OT1}
\fontfamily{cmr}
\fontseries{m}
\fontshape{n}
\fontsize{8}{9.5}
\selectfont
\begin{tabular}{@{\hspace{3mm}}c}
\hline\hline\\[-1mm]
\begin{minipage}{9.3cm}
\begin{tabular}[t]{l@{\hspace{4pt}}c@{\hspace{4pt}}l@{\hspace{1pt}}r@{\hspace{4pt}}l@{\hspace{1pt}}r}
case & $\fg$ & $\phi$ & dim$(\phi)$ & $\Alt^2 \phi$ & dim$(\Alt^2 \phi)$\\ \midrule
(1a) & $\ueber{\so(2\ell+1)}{\ell > 2}$ &
$(1,0,\ldots,0)$ & $2\ell{+}1$ & $(0,1,0,\ldots,0)$  & $(2\ell{+}1)\ell$\\ \midrule[0pt]
(1b) & $\so(5)$ &  $(1,0)$ & $5$ & $(0,2)$ & $10$\\ \midrule[0pt]
(2a) & $\ueber{\so(2\ell)}{\ell > 3}$ &
$(1,0,\ldots,0)$ & $2\ell$ & $(0,1,0,\ldots,0)$ & $(2\ell{-}1)\ell$\\ \midrule[0pt]
(2b) & $\so(6)$ & $(1,0,0)$ & $6$ & $(0,1,1)$ & $15$\\ \midrule[0pt]
(3) & $\ueber{\su(\ell+1)}{\ell\geq 3}$ &
$(0,1,0,\ldots,0)$ & $\tfrac{\ell(\ell+1)}{2}$ & $(1,0,1,0,\ldots,0)$ & $3 \tbinom{\ell+2}{4}$\\ \midrule[0pt]
(4) & $\ueber{\su(\ell+1)}{\ell\geq 2}$ &
$(2,0,\ldots,0)$ & $\tfrac{(\ell+1)(\ell+2)}{2}$ & $(2,1,0,\ldots,0)$ & $3 \tbinom{\ell+3}{4}$\\ \midrule[0pt]
(5) & $\so(10)$ &
$(0,0,0,1,0)$ & $16$ & $(0,0,1,0,0)$ & $120$\\ \midrule[0pt]
(6) & $\fe_6$ &
$(1,0,0,0,0,0)$ & $27$ & $(0,0,1,0,0,0)$ & $351$\\ \midrule[0pt]
(7) & $\ueber{\su(\ell+1)}{\ell\geq 1}$ &
$(1,0,\ldots,0)$ & $\ell{+}1$ & $(0,1,0,\ldots,0)$ & $\tfrac{\ell(\ell+1)}{2}$\\ \midrule[0pt]
\end{tabular}
\vspace{1mm}
\end{minipage}\\
\hline\hline
\end{tabular}
\end{center}
\end{theorem}
\begin{proof}
If $\phi=\phi_1 \oplus \phi_2$ is not irreducible then neither is $\Alt^2(\phi_1 \oplus \phi_2)=
\Alt^2 \phi_1 \oplus (\phi_1 \otimes \phi_2) \oplus \Alt^2 \phi_2$ irreducible.
The Theorem follows from Thm.~4.7 and Tab.~6 of Ref.~\cite{Dynkin57}.
$\hfill\blacksquare$
\end{proof}

\noindent
Relying on Theorem~\ref{alttable} we obtain:

\begin{theorem}\label{thm_altsym}
Assume $\phi$ is a (finite-dimensional) representation of a compact semi\-simple Lie algebra $\fg$.
The representations $\Alt^2 \phi$ and $\Sym^2 \phi$ are both irreducible if and only if
$\fg=\su(\ell+1)$ with $\ell \geq 1$ and $\phi$ is 
(up to an outer automorphism of $\fg$) the standard representation
[i.e.\ its highest weight is $(1,0,\ldots,0)$]. 
\end{theorem}
\begin{proof}
We go through the cases of Theorem~\ref{alttable}.
Let us denote by $\phi_x$ the representation with highest weight $x$.
In the cases (1a)-(2b), it follows from Ex.~19.21  of Ref.~\cite{FH91}
that $\Sym^2 \phi_{(1,0,\ldots,0)}= \phi_{(2,0,\ldots,0)} \oplus \phi_{(0,\ldots,0)}$ decomposes.
In the case of (3), we can use a Pieri-type formula (see Prop.~15.25(ii) of Ref.~\cite{FH91})
to show that $\Sym^2 \phi_{(0,1,0)} =
\phi_{(0,0,0)} \oplus \phi_{(0,2,0)}$ and $\Sym^2 \phi_{(0,1,0,\ldots,0)} =
\phi_{(0,0,0,1,0,\ldots,0)} \oplus \phi_{(0,2,0,\ldots,0)}$ decompose. 
In the case of (4), we can use again a Pieri-type formula (see Prop.~15.25(i) of Ref.~\cite{FH91})
to show that $\Sym^2 \phi_{(2,0,\ldots,0)} = \phi_{(0,2,0,\ldots,0)} \oplus \phi_{(4,0,\ldots,0)}$
decomposes. 
In the cases (5) and (6), we explicitly compute the decomposition using computer algebra systems 
such as {\sf LiE}~\cite{LIE222} 
and {\sf MAGMA}~\cite{MAGMA}.
We get for (5) 
that $\Sym^2 \phi_{(0,0,0,1,0)} = \phi_{(0,0,0,2,0)} \oplus \phi_{(1,0,0,0,0)}$
and for (6) that $\Sym^2 \phi_{(1,0,0,0,0,0)}=
\phi_{(0,0,0,0,0,1)} \oplus \phi_{(2,0,0,0,0,0)}$.
In the case of (7), we use again a Pieri-type formula (see Prop.~15.25(i) of Ref.~\cite{FH91})
to show that $(\phi_{(1,0,\ldots,0)})^{\otimes 2}= \phi_{(2,0,\ldots,0)} \oplus
\phi_{(0,1,0,\ldots,0)}= \Sym^2 \phi_{(1,0,\ldots,0)} \oplus \Alt^2 \phi_{(1,0,\ldots,0)}$.
Therefore, case (7) is the only case for which both the alternating and symmetric square
are irreducible.
\phantom{XXX}$\hfill\blacksquare$
\end{proof}



\bibliographystyle{spphys}
\bibliography{conty}

\end{document}